\newif\ifEC
\ECfalse

\ifEC
\documentclass[format=acmsmall, review=false]{acmart}
\usepackage{acm-ec-22}
\usepackage{booktabs} 
\usepackage[ruled]{algorithm2e} 

\SetAlFnt{\small}
\SetAlCapFnt{\small}
\SetAlCapNameFnt{\small}
\SetAlCapHSkip{0pt}
\IncMargin{-\parindent}

\setcitestyle{authoryear}
\else
\documentclass{article}
\usepackage[margin=1in]{geometry}

\fi

\usepackage{amsmath, amsthm, amssymb}
\usepackage{graphicx}
\usepackage{verbatim}
\usepackage{natbib}
\usepackage{caption}
\usepackage{subcaption}
\usepackage{fancyvrb}
\usepackage{enumerate}
\usepackage{relsize}
\usepackage{mathtools}
\usepackage{tikz}
\usepackage{float}
\usepackage{hyperref}
\usepackage{bbm}
\hypersetup{colorlinks,citecolor=blue,urlcolor=blue,linkcolor=blue}

\usepackage{stefan_tex}
\graphicspath{{./figures/}}

\def\independenT#1#2{\mathrel{\rlap{$#1#2$}\mkern2mu{#1#2}}}

\theoremstyle{plain}
\newtheorem{prop}{Proposition}

\newtheorem{coro}[prop]{Corollary}
\newtheorem{lemm}[prop]{Lemma}
\newtheorem{theo}[prop]{Theorem}

\theoremstyle{definition}
\newtheorem{exam}{Example}
\newtheorem{netexam}{Network Example}

\newtheorem{assu}{Assumption}

\theoremstyle{remark}

\newcommand{\rhon}{\rho_n}
\newcommand{\ngb}{\mathcal{N}}
\newcommand{\SDE}{\operatorname{SDE}}
\newcommand{\LDE}{\operatorname{LDE}}
\newcommand{\LTE}{\operatorname{LTE}}

\newcommand{\Ps}{P^{\star}}
\newcommand{\Ys}{Y^{\star}}
\newcommand{\Qs}{Q^{\star}}
\newcommand{\Zs}{Z^{\star}}
\newcommand{\bv}{\boldsymbol{v}}
\def\trans{^{\scriptscriptstyle \sf T}}
\newcommand{\pss}{p^{\star}}
\newcommand{\qss}{q^{\star}}
\newcommand{\bu}{\boldsymbol{u}}
\newcommand{\Ph}{\hat{P}}

\newcommand{\Dh}{\hat{D}}
\newcommand{\Wh}{\hat{W}}
\newcommand{\uh}{\boldsymbol{\hat{u}}}
\newcommand{\bxi}{\boldsymbol{\xi}}

\begin{document}

\ifEC
\author{Submission 374}
\else
\author{Shuangning Li \and Stefan Wager}
\date{Stanford University}
\fi
\title{Network Interference in Micro-Randomized Trials}

\ifEC
\else
\maketitle
\fi

\begin{abstract}
The micro-randomized trial (MRT) is an experimental design that can be used to develop optimal mobile health interventions. In MRTs, interventions in the form of notifications or messages are sent through smart phones to individuals, targeting a health-related outcome such as physical activity or weight management. Often, mobile health interventions have a social media component; an individual's outcome could thus depend on other individuals' treatments and outcomes.  In this paper, we study the micro-randomized trial in the presence of such cross-unit interference. We model the cross-unit interference with a network interference model; 
the outcome of one individual may affect the outcome of another individual if and only if they are connected by an edge in the network. Assuming the dynamics can be represented as a Markov decision process, we analyze the behavior of the outcomes in large sample asymptotics and show that they converge to a mean-field limit when the sample size goes to infinity. Based on the mean-field result, we give characterization results and estimation strategies for various causal estimands including the short-term direct effect of a binary intervention, its long-term direct effect and its long-term total effect.
\end{abstract}

\ifEC
\begin{titlepage}
\maketitle
\end{titlepage}
\fi

\section{Introduction}

In mobile health studies, the micro-randomized trial (MRT) is an experimental design that is often used to help evaluate and optimize dynamic interventions \citep{battalio2021sense2stop,dempsey2020stratified,  klasnja2015microrandomized, liao2016sample, liao2021off, walton2020micro}. In an MRT, the intervention assigned to any given unit is sequentially randomized at many decision points over the course of the experiment. For example, in the case of a wellness app that seeks to encourage non-sedentary behavior, \citet{klasnja2019efficacy} ran an MRT that, many times a day, randomly assigned study participants to either receive one of a number of available messages or no message. MRTs are a powerful statistical tool, in that they can be be used to assess both short- and long-term effects of a number of actions, and to design optimal dynamic treatment regimes \citep{hernan2020whatif,hu2021offpolicy,kallus2020double,liao2021off,luckett2019estimating,robins1986new}.

Most available studies of MRTs, including the ones cited above, assume that
there is no-cross unit interference. In other words, although treatments given
through time to a single unit may induce complex dependence patterns, it is assumed
that treatments given to one unit cannot affect outcomes for a different unit. There
has been growing interest, however, in setting areas where this assumption is
not applicable, including settings where the effect of an intervention is mediated
through interactions on a social network \citep{aral2011creating,eckles2016estimating}.
For example, \citet{aral2017exercise} describe a fitness app where users are notified
if any of their friends went running---thus potentially encouraging them to go running also.

The goal of this paper is to lay conceptual and methodological groundwork for the study
of micro-randomized trials in the presence of cross-unit interference. Following a number
of recent studies, we start by modeling MRTs as a Markov decision process
\citep{kallus2020double,liao2021off,luckett2019estimating}. We then incorporate cross-unit effects
using a network interference model, where units are placed at vertices of an exposure
graph and may interfere with each other if there is an edge connecting them
\citep{aronow2017estimating,athey2018exact,leung2020treatment,li2020random}. Our main
results pertain to characterization and estimation of a number of causal targets in this
model.

This paper is structured as follows.
In Section \ref{section:MDP}, we describe the basic problem setting and introduce the key assumptions we make through our this work. We also discuss the causal estimands of interest. In Section \ref{section:mean_field}, we analyze the behavior of our system under large-sample asymptotics and show that they converge to a mean-field limit when the sample size goes to infinity. Based on the mean-field result, in Section \ref{section:estimation}, we give characterization results and estimation strategies separately for each of the causal estimands of interest. We note that even though we motivate the problem with the mobile health studies, our problem setting, methodology and results can be applied to other application areas as well (e.g., see Example \ref{exam:shopping}).

\subsection{Related work}

The micro-randomized trial was introduced by \citet{liao2016sample} and \citet{klasnja2015microrandomized} as an experimental design for developing just-in-time adaptive interventions. Subsequent works have applied the micro-randomized trials to study various mobile health interventions, which are designed to increase physical activity among sedentary individuals \citep{klasnja2015microrandomized, klasnja2019efficacy}, to support support stress-management in newly abstinent smokers \citep{battalio2021sense2stop}, or to provide weight loss support \citep{qian2021estimating}, etc. Our work differs from the previous work in allowing for cross-unit interference, which is especially relevant when people involved in the study interact with each other virtually or in-person. 

The existing literature on treatment effect estimation under cross-unit interference has mostly focused on the case where there is a single time point, i.e., one gets to observe the treatment and the outcome only once for each unit \citep{aronow2017estimating, basse2019randomization, hudgens2008toward,  leung2020treatment, li2020random, savje2021average}. The problem we consider in this paper has an additional time axis, which gives rise to a number of new phenomena and challenges.
First, the causal estimands of interest are different. Unlike in the static setting where
natural causal effects, including direct, indirect and total causal effects, are time-independent \citep{hu2021average, hudgens2008toward, savje2021average}, we now need to address lagged
and/or compounding treatment effects over time.
Second, the dynamic nature of our problem fundamentally alters estimation considerations:
While it gives rise to challenges, the increased data from repeated sampling also opens
the door to new estimation strategies.


From a theoretical perspective, our work is also related to the network game literature. In a network game, an agent's payoff depends on her own strategy and the strategy of other agents connected to her through a network. A major question of interest is to study the Nash equilibrium \citep{ ballester2006s,bramoulle2007public,bramoulle2014strategic,galeotti2010network,galeotti2020targeting,jackson2014games,parise2019graphon}. In this paper, we model the dynamic of the outcomes by a Markov decision process. With a time invariant policy, the stationary distribution of our process is closely connected to the Nash equilibrium of the network game. Therefore, tools and results in establishing existence, uniqueness, and existence of certain limit of the Nash equilibrium in network games can be transferred here to study the stationary distribution. Some regularity conditions we use in this paper are also closely related to some regularity conditions used by \citet{ballester2006s} and \citet{parise2019graphon} in the setting of network games.
On a conceptual level, our work is also related to several recent papers that
used mean-field modeling to study causal inference in complex systems
\citep{johari2022experimental,wager2021experimenting}.

\section{Problem Setup}
\label{section:MDP}
Suppose that there are $n$ subjects of interests indexed by $i = 1, \dots, n$. Assume that there is a undirected graph with indices corresponding to the $n$ subject. We call the undirected graph the \textit{interference network} or the \textit{interference graph}. We use $\cb{E_{ij}}$ to denote the edge set of the graph. Let $\mathcal{N}_i = \cb{j: E_{ij} = 1}$ be the set of neighbors of subject $i$. Let $Y_{i t} \in \cb{0,1}$ denote the outcome of interest at time $t$ and $W_{i t}  \in \cb{0,1}$ be the treatment at time $t$. In this work, we focus on binary outcomes and binary treatments. The results can be easily extended to categorical variables.
For concreteness, one could consider a running app example loosely motivated by the work of 
\citet{aral2017exercise}, were the subjects of interest correspond the users of the app,
the interference network is the friendship network of the app,
$Y_{it} \in \cb{0, \, 1}$ denotes whether user $i$ went running in the $t$-th time period, and
$W_{i t}  \in \cb{0,\,1}$ corresponds to whether user $i$ received a motivational message 
in the $t$-th time period.

We model the dynamic of the outcomes with a Markov decision process (MDP). At each time $t$, we define the state $Y_t = (Y_{1t}, Y_{2t}, \dots, Y_{nt})$ to be outcomes of all users at time $t$. Let the action $W_t = (W_{1t}, W_{2t}, \dots, W_{nt})$ be the vector of treatments at time $t$. We make the following assumption on the transition probabilities. 

\begin{assu}[MDP with Network Interference]
\label{assu:MDP}
Each unit $i = 1, \, \ldots, \, n$ is characterized by an activation function $f_i(\cdot)$
such that, conditionally on $Y_{1:t}$ and $W_{1:t}$, 
\begin{equation}
\label{eqn:MDP}
Y_{i (t+1)} \sim \operatorname{Ber}(f_i(Y_{i t}, W_{i t}, Z_{i t})) \textnormal{ independently,}
\end{equation}
where $Z_{i t} = \sum_{j \in \mathcal{N}_i} Y_{j t}$, and $0 < f_i(y, \, w, \, z) < 1$
for all $y, \, w \in \cb{0, \, 1}$ and $z \in \RR_+$.
\end{assu}

In the context of our running example, Assumption \ref{assu:MDP} can be interpreted as the following: The probability of an individual goes running tomorrow depends on whether she went running today ($Y_{i t}$), whether she receives any encouraging message from the app ($W_{i t}$), and the total number of her friends that went running today ($Z_{i t}$) and her individual characteristics ($f_i$). 

Assumption \ref{assu:MDP} has essentially two key components: It formalizes the
assumed Markovian dynamics, and specifies the form of the cross-unit interference.
The Markovian assumption is a common assumption made in the statistical and
reinforcement learning literature in modeling similar problems
\citep[e.g.,][]{antos2008learning,kallus2020double,russel2010artificial}.
Closely related to us, \citet{liao2021off} study the micro-randomized trials and
model the dynamics of the outcome of interest with a MDP.

The second part of Assumption \ref{assu:MDP} requires the interference to act
along the network, and also to be ``anonymous" in the sense that $Y_{i(t+1)}$
does not depend on the specific behavior of a neighbor but only on $Z_{it}$, the total number of
neighbors $j \in \mathcal{N}_i$ with $Y_{jt} = 1$. This concept is related to the ``anonymous interference" assumption proposed by \citet{hudgens2008toward}. However, unlike in \citet{hudgens2008toward} where the outcome of subject $i$ depends on the total number of treated neighbors, Assumption \ref{assu:MDP} states that the outcome $Y_{i(t+1)}$ depends on the total number of neighbors with $Y_{jt} = 1$.
The reason we specify the assumption in this way is that we believe that, in many leading
examples of MRTs with cross-unit effects, especially over a social network, it seems more
likely that the $i$-th unit would be responding to past behaviors of their neighbors, rather
then their past treatment assignments.


Next, we make assumptions about how treatments are assigned in the MRT. Here,
we use the simplest possible assumption, namely that the treatment given to each unit
is Bernoulli-randomized in each time period, and that the underlying treatment-assignment
probabilities are time-invariant:

\begin{assu}[Bernoulli treatment]
\label{assu:Bern}
The treatments $W_{it} \sim \operatorname{Ber}(\pi_i)$ independently for each $i$ and each $t$. 
\end{assu}

With these two assumptions, we will show next that in the long run, the distribution of the state $Y_{t} = (Y_{1t}, Y_{2t}, \dots, Y_{nt})$ will converge to a distribution $\mu(\pi)$. We use $\mu(\pi)$ to emphasize the dependency on treatment probabilities $\pi$. This $\mu(\pi)$ also corresponds to the stationary distribution of the Markov chain induced by the policy defined in Assumption \ref{assu:Bern};
we show existence of $\mu(\pi)$ in Section \ref{sec:stationary}.

\subsection{Causal estimands}
\label{subsection:estimands}
In this paper, we would like to study the following few causal estimands. We discuss their definition and interpretations in this section and will study the characterization and estimation of them in Section \ref{section:estimation}. 

\paragraph{Short-term direct effect.} The short-term direct effect quantifies the immediate effect of the unit's treatment on its own outcome. The direct effect defined below takes the average of the effects for each unit. 
\begin{equation}
\label{eqn:SDE_def}
\tau_{\SDE, t} = \frac{1}{n}\sum_{i=1}^n f_i\p{Y_{i t}, 1, Z_{i t}} - f_i\p{Y_{i t}, 0, Z_{i t}} 
\end{equation}

\paragraph{Long-term direct effect.} The long-term direct effect captures the long-term effect of the unit's treatment on its own outcome, again averaged over units. Proposition \ref{prop:MDP_stationary} (in Section \ref{sec:stationary}) implies that the ``long-term" outcomes can be described by the stationary distribution of $Y$, thus we define the long-term direct effect in term of the stationary distribution. In other words, $\tau_{\LDE}(\gamma_1, \gamma_2)$ defined below concerns the effect of a unit's treatment on its expected outcome under the stationary distribution. More specifically, $\tau_{\LDE}(\gamma_1, \gamma_2) $ quantifies the average effect on unit $i$'s of changing the treatment probability of unit $i$ from $\gamma_2$ to $\gamma_1$.
\begin{equation}
\label{eqn:LDE_def}
\tau_{\LDE}(\gamma_1, \gamma_2) = \frac{1}{n} \sum_{i = 1}^n \p{\EE[\mu(\pi_i = \gamma_1, \pi_{-i})]{Y_i} - \EE[\mu(\pi_i = \gamma_2, \pi_{-i})]{Y_i} }. 
\end{equation}
Here $\mu(\pi_i = \gamma, \pi_{-i})$ stands for the stationary distribution of MDP \eqref{eqn:MDP} when the treatment probability of the $i$-th unit has been changed to $\gamma$.

\paragraph{Long-term total effect.} Different from the previous two estimands, the long-term total effect focuses on the effect of changing the treatment probability for every user at the same time. It measures the effect of changing the entire treatment vector from $\pi_2$ to $\pi_1$ on the expected average outcome under the stationary distribution. 
\begin{equation}
\label{eqn:LTE_def}
\tau_{\LTE}(\pi_1, \pi_2) = \frac{1}{n} \sum_{i = 1}^n \cb{\EE[\mu(\pi_1)]{Y_i} - \EE[\mu(\pi_2)]{Y_i} } 
\end{equation}
Without cross-unit interference, if $\pi_1 = (\gamma_1, \dots, \gamma_1)$ and $\pi_2 = (\gamma_2, \dots, \gamma_2)$, then the long-term direct effect will be the same as the long-term total effect: $\tau_{\LDE}(\gamma_1, \gamma_2) = \tau_{\LTE}(\pi_1, \pi_2)$. 

Below, we illustrate through a few examples what the three causal estimands correspond to and how to interpret them. 

\begin{exam}[Fitness app]
\label{exam:fitness}
In the fitness app example, recall that the units are users, the outcomes $Y$ correspond to whether users go running on a particular day, and the treatments $W$ correspond to the motivational messages. 
Researchers are interested in the effect of the motivational messages on users' level of physical activity---running---in this case. In this example, the short-term direct effect answers the following question: how much difference does a message send to unit $i$ make in user $i$'s probability of running tomorrow?  While the short-term direct effect concerns the immediate effect of the message, the two long term effects concern how much the messages can change users' long-term habit. The long-term direct effect quantifies on average by how much the messages directed at user $i$ can shift user $i$'s habit. The long-term total effect focuses on the overall effect: if we increase the frequency of motivational messages to everyone in the community, by how much can we change the overall running habits? 
\end{exam}

\begin{exam}[Shopping habit]
\label{exam:shopping}
Imagine there are two Grocery stores in a neighborhood, Store A and Store B. There are $n$ people living in the neighborhood who go shopping at Grocery stores each week. Let $Y_{it} = 1$ if the $i$-th person go to Store A in week $t$, and let $Y_{it} = 0$ otherwise. In order to attract customers, Store A starts to send coupons to people. Let $W_{it} = 1$ if the $i$-th person receives a coupon in week $t$, and let $W_{it} = 0$ otherwise. Let the interference graph be the natural friendship network.
In this context, Assumption \ref{assu:MDP} means that the choice of Grocery store depends on the choice last week, whether the person receives a coupon and the behavior of friends. Interference exists in this case, because of conformity: people tend to go to the stores their friends go. 
In this example, the short-term direct effect concerns people's behavior exploiting the coupon; the coupon provides direct financial incentives for customers to come and shop. This effect matters when a store needs a quick increase in the number of customers.
The long-term direct effect answers the question of (on average) how many coupons Store A needs to give in order to shift the shopping habit of one customer. This effect is most relevant if the store targets a small proportion of people and hope to increase their frequency of shopping. 
The long-term total effect is about how many coupons in total the store need to give to everyone in the neighborhood to increase the total number of customers in the long term. This can be relevant for long-term overall planning of the store.

\end{exam}

\subsection{Notation}
Throughout this paper, we use $C_1, C_2, \dots$ for constants not depending on $n$. Note that $C_1, C_2, \dots$ might mean different things in different settings. 
We let $f'_i \p{y, w, z}$, $f''_i \p{y, w, z}$, etc., denote derivatives of $f_i$ with respect to the third argument $z$.
Let $\law(X)$ denote the law of a random variable or random vector $X$.
We use the notation $a \vee b = \max(a,b)$ and $a \wedge b = \min(a,b)$. 
We use $ \oo(), \oo_p(), o_p()$ in the following sense: $a_n = \oo(b_n)$ if $|a_n| \leq C b_n$ for $n$ large enough. $X_n = \oo_p(b_n)$, if for any $\delta>0$, there exists $M,N > 0$, s.t. $\PP{|X_n| \geq M b_n} \leq \delta$ for any $n > N$. $X_n = o_p(b_n)$, if $\lim\PP{|X_n| \geq \epsilon b_n} \to 0$ for any $\epsilon > 0$. We write $X_n \Rightarrow X$ to say that $X_n$ converges in distribution to $X$. 

\section{Weak convergence}
\label{section:mean_field}

In this section, we consider large-sample behavior of our system.
First, we verify existence of a stationary distribution, justifying our
above discussions. Second, we show that, as $n \rightarrow \infty$, our
system is tightly coupled to an auxiliary stochastic process, which we find to
be easier to analyze. This coupling result lets us establish mean-field convergence,
and yields powerful tools to analyze the MDP in \eqref{eqn:MDP} that we then use to study
our causal estimands in later sections. 

\subsection{Existence of a stationary distribution}
\label{sec:stationary}

Since both $Y_{i t}$ and $W_{i t}$ are binary random variables, the function $f_i$ can be decomposed into four terms:
\begin{equation}
\label{eqn:f_decomposition}
f_i(y, w, z) = a_i(z) + b_i(z) w + c_i (z) y + d_i(z) wy. 
\end{equation}

\begin{assu}[Boundedness and Lipschitzness]
\label{assu:bound_fi}
Each function $f_i$ is Lipschitz with Lipschitz constant $L_n$ in its third argument. The functions $c_i$ and $d_i$ satisfy $\abs{c_i(z) + d_i(z) w} \leq B$ for any $z \in \RR_+$ and any $w \in \cb{0,1}$. 
\end{assu}

\begin{assu}[Node Degree]
\label{assu:largest_degree}
The largest node degree of the interference graph is bounded by $D_n$. 
\end{assu}

\begin{assu}[Contraction]
\label{assu:contraction}
The constants $B$, $L_n$ and $D_n$ defined in Assumptions \ref{assu:bound_fi} and \ref{assu:largest_degree} satisfy
\begin{equation}
B + L_n D_n \leq C < 1. 
\end{equation}
\end{assu}

Assumption \ref{assu:bound_fi} states that the function $f_i$ is continuous in $z$ and it cannot change much if $z$ change a little. $L_n$ captures the maximum possible effect of the number of neighbors on the outcome. The term $c_i(z) + d_i(z) w$ can be written as $c_i(z) + d_i(z) w = f_i(1, w, z) - f_i(0, w, z)$; thus $B$ quantifies the maximum possible effect of outcome at time $t$ on outcome at time $t+1$. In Assumption \ref{assu:contraction}, the term $L_n D_n$ bounds the total effect of neighbors on a subject's outcome: if on average each neighbor's outcome at time $t$ is increased by $\delta$, then the change in the outcome at time $t+1$ is expected to be bounded by $L_n D_n \delta$. Bounding the sum of $B$ (individual effect) and $L_n D_n$ (neighbor effect), Assumption \ref{assu:contraction} limits the effect of outcomes (outcomes of the subject $i$ and its neighbors) at time $t$ on the outcome at time $t+1$. 

Assumptions \ref{assu:bound_fi} - \ref{assu:contraction} are crucial to ensure the uniqueness of the stationary distribution of the Markov chain induced by the policy in Assumption \ref{assu:Bern}. Assumption \ref{assu:contraction} is related to the contraction assumption that is often made in the network game literature (see, e.g., \cite{ballester2006s} and \cite{parise2019graphon}). There, the contraction assumption is usually stated to guarantee the influence of neighbors does not eclipse a subject's own strategy. This assumption is often used to obtain uniqueness of the Nash equilibrium of a network game.

Under Assumptions \ref{assu:MDP} - \ref{assu:contraction}, Proposition \ref{prop:MDP_stationary} shows that in the long run, the distribution of the state $Y_{t}$ converges to a distribution $\mu(\pi)$, which corresponds to the stationary distribution of the Markov chain induced by the policy defined in Assumption \ref{assu:Bern}. 
\begin{prop}[Stationary distribution]
\label{prop:MDP_stationary}
\begin{enumerate}
\item Under Assumptions \ref{assu:MDP} and \ref{assu:Bern}, there exists a stationary distribution $\mu(\pi)$, such that if $Y_0 \sim \mu(\pi)$, then $Y_t \sim \mu(\pi)$ for any $t \geq 0$. 
\item Furthermore, under Assumptions \ref{assu:MDP} - \ref{assu:contraction}, the Markov chain induced by the policy defined in Assumption \ref{assu:Bern} is ergodic, i.e., the stationary distribution $\mu(\pi)$ is unique. For any initial distribution of $Y_0$, we have that $Y_t \Rightarrow \mu(\pi)$ as $t \to \infty$. 
\end{enumerate}
\end{prop}
Proposition \ref{prop:MDP_stationary} establishes existence of the stationary distribution and ergodicity of the system. It has nice implications: On the one hand, if researchers are interested in the long-term behavior of the outcomes, it suffices to study the stationary distribution of Markov chain. On the other hand, the ergodicity enables us to approximate the stationary distribution by taking averages along the history of the outcomes.

Finally, to get a better intuition for the role of the scaling constants $L_n$ and $D_n$,
we consider the following models of the interference graph. We emphasize that we do not
rely on these random graph models for our analysis; rather, in this paper, they are simply
used to interpret scalings.

\begin{netexam}[Erd\H{o}s-R\'enyi]
\label{netexam:erdos}
Each edge is included in the interference graph with probability $\rhon$, independently from every other edge, i.e., $E_{ij} \sim \operatorname{Bernoulli} (\rhon)$ independently. 
\end{netexam}
In this example, each unit is expected to have $n \rhon$ neighbors. We can easily show with concentration inequalities that the largest node degree is bounded by $2 n \rhon$ with high probability. Given this high probability event, one possible choice of $D_n$ is $2 \rhon n$. Then in order for Assumption \ref{assu:contraction} to hold, we need $L_n \leq (C - B)/(2 n\rhon)$. Recalling that $L_n$ is the Lipschitz constant of the function $f_i$ in the $z$ argument. In this case, the bound on $L_n$ suggests that if we instead treat $f_i$ as a function of $z/(n \rhon)$, then the function should be Lipschitz with constant $(C - B)/2$. Since each unit is expected to have $n \rhon$ neighbors, $Z_{it}/(n\rhon)$ can be treated as the average of neighbors' outcomes of unit $i$. Thus, the first part of Assumption \ref{assu:bound_fi} can be reinterpreted as the following: The effect of the average of neighbors' outcomes on a subject's outcome cannot exceed a constant. Throughout this paper, it is helpful to think of this setting and treat $L_n = \oo_p\p{1/(n \rhon)}$ and $D_n = \oo_p\p{n \rhon}$. 

Though easy-to-understand, the Erd\H{o}s-R\'enyi graph may not be the best choice to model real life networks. A better way of modeling the interference network is the following. 
\begin{netexam}[Graphon]
\label{netexam:graphon}
The interference graph is randomly generated as follows. Each subject has
a random type \smash{$U_i \simiid \operatorname{Uniform}[0, \, 1]$}. There is a $\rhon \in (0, \, 1]$ and a symmetric measurable function $G : [0, \, 1]^2 \rightarrow [0, \, 1]$
called a \emph{graphon} such that $E_{ij} \sim \operatorname{Bernoulli}\p{\rhon G(U_i, \, U_j)}$ independently for all $i < j$. 
\end{netexam}
In this example, unit $i$ and unit $j$ are connected in the interference graph with probability $\rhon G(U_i, \, U_j)$. Unlike the Erd\H{o}s-R\'enyi graph, the edge-forming probability depends on random types of the units. If we go through similar analysis as in the above example, we find that we can still have $L_n = \oo_p\p{1/(n \rhon)}$ and $D_n = \oo_p\p{n \rhon}$. 

\subsection{Mean-field characterization}

Our next goal is to show that, as $n$ gets large, our MDP of interest
is coupled with the following dynamical system. Let $P_t = (P_{1 t}, P_{2 t}, \dots, P_{n t}) \in [0,1]^n$ be the state of the dynamical system at time $t$. The evolution rule of the system is the following:\footnote{The function $f_i$ is initially only defined for $y \in \cb{0,1}$ and $w \in \cb{0,1}$, but the form in \eqref{eqn:f_decomposition} explicitly extends the domain of $f_i$ to $y \in [0,1]$ and $w \in [0,1]$.} 
\begin{equation}
\label{eqn:dynamic_system}
\begin{split}
P_{i (t+1)} &= f_i(P_{ i t}, \pi_i, Q_{i t}) \\
&= a_i\p{Q_{i t}} + b_i\p{Q_{i t}} \pi_i +   c_i\p{Q_{i t}} P_{i t} + d_i\p{Q_{i t}} \pi_i P_{i t},
\end{split}
\end{equation}
where $Q_{i t} = \sum_{j \in \ngb_i} P_{j t}$, and the functions $f_i$, $a_i, \dots, d_i$
are defined in \eqref{eqn:MDP} and \eqref{eqn:f_decomposition}.

In comparing this process with \eqref{eqn:MDP}, one can interpret
$P_{i,t}$ as the probability that $Y_{i,t} = 1$ given past information. The
key difference from \eqref{eqn:dynamic_system} is that, here, the state of the
$i$-th user depends directly on its neighbors probabilities rather than their realized
outcomes---and removing this extra layer of randomness makes the process \eqref{eqn:dynamic_system}
considerably easier to study.



Intuitively, this form of the stochastic process is easier to analyze for at least two reasons. First, the probabilities are non-random numbers, thus if the outcomes $Y_{it} \sim \operatorname{Bern}(P_{i t})$ independently, then the $Y_{it}$'s are independent. This is not the case for the MDP we considered in the previous section. Indeed, we can imagine, if individuals $i$ and $j$ have many common friends, then the corresponding $Z_{i(t-1)}$ and $Z_{j(t-1)}$ would be very much correlated and thus $Y_{it}$ and $Y_{jt}$ will be correlated. Independence makes it much easier to separate the direct effect from the total effect. Second, it is easier to analyze the fixed point of the system \eqref{eqn:dynamic_system} than to find the stationary distribution of the MDP defined in \eqref{eqn:MDP}. The fixed point of system \eqref{eqn:dynamic_system} can be characterized by a vector of length $n$, whereas the stationary distribution of the MDP essentially needs a vector of length $2^n$. 

We will then study the properties of the system \eqref{eqn:dynamic_system}. Specifically, we show the existence and uniqueness of the fixed point and we demonstrate the closeness of stationary distribution of the MDP to the fixed point of system \eqref{eqn:dynamic_system}. We call our results \textit{mean-field} results, because individual's behavior is independent under system \eqref{eqn:dynamic_system} and thus they are almost independent under the MDP defined in Section \ref{section:MDP}. 

Proposition \ref{prop:system_fixed} establishes the existence and uniqueness of the fixed point of system \eqref{eqn:dynamic_system}. 
\begin{prop}[Fixed point]
\label{prop:system_fixed}
\begin{enumerate}
\item If the functions $f_i$'s are continuous and satisfy $0 \leq f_i \leq 1$, then there exists a fixed point $P^{\star} \in [0,1]^n$ of the system \eqref{eqn:dynamic_system}, i.e., if $P_{t} = P^{\star}$, then $P_{t+1} = P^{\star}$. 
\item Under Assumptions \ref{assu:bound_fi} - \ref{assu:contraction}, the fixed point is unique, and for any value of $P_0$, we have $P_t \to \Ps$ as $t \to \infty$.  
\end{enumerate}
\end{prop}

With the above proposition, it is natural to ask how good the mean-field approximation is in terms of the long-term behavior. More precisely, if we let $\Ys_{it} \sim \operatorname{Ber}(\Ps_i)$ independently, we would like to study how close the law of $\Ys_{it}$ is to $\mu(\pi)$, which is the stationary distribution of the MDP in \eqref{eqn:MDP}. To this end, we define an $L_1$-Wasserstein distance between two laws $\nu_1$ and $\nu_2$
\begin{equation}
\label{eqn:wasser_L1}
W_{L_1}\p{\nu_1, \nu_2} = \inf \cb{\EE{\|X_1 - X_2\|_1}: \law\p{X_1} = \nu_1, \law\p{X_2} = \nu_2}. 
\end{equation}
We also define a graph dependent distance between two random vectors and two laws respectively. Let $k \in \NN_+$ be a positive integer. For random vectors $X, Y \in \cb{0,1}^n$, let
\begin{equation}
\label{eqn:dek_def}
d_{E,k}\p{X, Y} = \max_i \bigg(\mathbb{E}\bigg[ \Big(\sum_{j \in \mathcal{N}_i} \abs{X_j - Y_j}\Big)^k\bigg]\bigg)^{\frac{1}{k}}. 
\end{equation}
Here we write $E$ in the subscript to emphasize the dependence of this distance metric on the friendship network (graph) $E$. We also define a Wasserstein version of the metric $d_{E,k}$. For two laws $\nu_1$ and $\nu_2$, let
\begin{equation}
\label{eqn:wasser_dek}
W_{d_{E,k}}\p{\nu_1, \nu_2} = \inf\cb{d_{E,k}\p{X_1, X_2}: \law(X_1) = \nu_1, \law(X_2) = \nu_2 }. 
\end{equation}
We can easily verify that all three distances are indeed well defined distances satisfying the triangular inequality. (See more details in Appendix \ref{subsection:triangular}.) When $k = 1$, we omit the $k$ in the subscript for simplicity, i.e., we write $d_{E} = d_{E,1}$, and $W_{E} = W_{E,1}$. 

Roughly speaking, the $W_{L_1}\p{\law(X), \law(Y)}$ distance is the expected total number of elements that are different in the two random vectors $X$ and $Y$ under the best coupling. The distance $d_{E}\p{X, Y}$ is a different metric. For unit $i$, $\sum_{j \in \mathcal{N}_i} \abs{X_j - Y_j}$ is the number of different outcomes among the neighbors of $i$ in $X$ and $Y$. Thus, $d_{E,k}$ concerns the maximum expected number of different outcomes among neighbors, and same for $W_{d_{E,k}}$. At a higher level, $W_{L_1}$ is a ``collective" metric, measuring the overall distance between two distributions. When $W_{L_1}\p{\law(X), \law(Y)}$ is small, we know most of the elements in $X$ are the same as that of $Y$ under the best coupling, but there could still be some proportion of elements that are always different. 
The distance $W_{d_{E,k}}$, nevertheless, is an ``individualized" metric. When $W_{d_{E,k}}$ is small, then under the best coupling, for any unit $i$, the number of different outcomes among the neighbors of $i$ in $X$ and $Y$ is small. This implies that for any unit $i$, $f_i(y, w, \sum_{j \in \mathcal{N}_i} X_j)$ will be close to $f_i(y, w, \sum_{j \in \mathcal{N}_i} Y_j)$, and thus the outcome at the next time point will have similar behavior. 

Theorem \ref{theo:mean_field} establishes that the $L_1$-Wasserstein distance between $\mu$ and $\law\p{\Ys}$ is small, while Theorem \ref{theo:mean_field2} establishes similar results for the graph specific Wasserstein distance between $\mu$ and $\law\p{\Ys}$. 

\begin{theo}
\label{theo:mean_field}
Under Assumptions \ref{assu:MDP} - \ref{assu:contraction}, let $\mu$ be the stationary distribution of the MDP \eqref{eqn:MDP} and $\Ps$ be the fixed point of system \eqref{eqn:dynamic_system}. Assume that $\Ys_{i} \sim \operatorname{Ber}(\Ps_i)$ independently, then
\begin{equation}
W_{L_1}\p{\mu, \law(\Ys)} \leq n \sqrt{L_n} \sqrt{C}/(2(1-C)),
\end{equation} 
where $W_{L_1}$ is the $L_1$-Wasserstein distance defined in \eqref{eqn:wasser_L1}.
\end{theo}

\begin{theo}
\label{theo:mean_field2}
Under the conditions of Theorem \ref{theo:mean_field},
\begin{equation}
\begin{split}
 W_{d_E}\p{\mu, \law(\Ys)} &\leq \sqrt{D_n} C/(2(1-C)),\\
W_{d_{E,3}}\p{\mu, \law(\Ys)} &\leq (2C \sqrt{D_n} + 1)/(1-C),
\end{split}
\end{equation}
where $W_{d_E} = W_{d_1}$ and $W_{d_E,3}$ are the graph dependent distances defined in \eqref{eqn:wasser_dek}. 
\end{theo}

To interpret the results, consider the network examples \ref{netexam:erdos} and \ref{netexam:graphon}. There, we have that $L_n = \oo_p\p{1/(n \rhon)}$ and $D_n = \oo_p\p{n \rhon}$. Thus, Theorem \ref{theo:mean_field} implies that $W_{L_1}\p{\mu, \law(\Ys)} = \oo_p(n/\sqrt{n\rhon})$. Hence, as long as $n \rhon \to \infty$, $W_{L_1}\p{\mu, \law(\Ys)} =  o_p(n)$. Similarly, Theorem \ref{theo:mean_field2} implies that in examples \ref{netexam:erdos} and \ref{netexam:graphon}, $W_{d_E}\p{\mu, \law(\Ys)} = \oo_p(\sqrt{n\rhon})$. Hence, as long as $n \rhon \to \infty$, $W_{L_1}\p{\mu, \law(\Ys)} =  o_p(n \rhon)$. We note here that all the randomness inside the $\oo_p(\cdot)$ and $o_p(\cdot)$ comes from randomly generating the interference graph. Theorem \ref{theo:mean_field} and \ref{theo:mean_field2} give inequalities that hold almost surely conditioning on the interference graph, i.e., if $D_n$ and $L_n$ are treated as fixed numbers.

\section{Characterization and estimation of the causal estimands}
\label{section:estimation}

We now move on to study the causal estimands discussed in Section \ref{section:MDP}. We ask and aim at addressing the following questions: Can they be simplified into forms that are easier to analyze? Can they be consistently estimated?

\subsection{The short-term direct effect}


To study the short-term direct effect, we look at the micro-randomized trial at one specific time $t$. A natural estimator to use here is the inverse propensity weighted (IPW) estimator:
\begin{equation}
\label{eqn:IPW_estimator}
\htau_{\operatorname{IPW}, t} =  \frac{1}{n}\sum_{i=1}^n Y_{i (t+1)} \p{\frac{W_{t}}{\pi_i} - \frac{1 - W_{t}}{1 - \pi_i}}. 
\end{equation}

Theorem \ref{theo:SDE_estimation} establishes that the IPW estimator is consistent for the short-term direct effect. 
\begin{theo}[Short-term direct effect estimation]
\label{theo:SDE_estimation}
Under Assumption \ref{assu:MDP} and \ref{assu:Bern},
\begin{equation}
\label{eqn:SDE_estimation}
\htau_{\operatorname{IPW}} =  \tau_{\SDE, t} + \oo_p\p{\frac{1}{\sqrt{n}}},
\end{equation}
where the estimator $\htau_{\operatorname{IPW}}$ is defined in \eqref{eqn:IPW_estimator}, and the estimand $ \tau_{\SDE, t}$ is defined in \eqref{eqn:SDE_def}. 
\end{theo}

This result enables researchers to draw meaningful conclusions even when they only look at the micro-randomized trial at one decision point. By focusing on the treatments and the outcomes following immediately, we can consistently estimate the short-term direct effect even in the presence of interference. 
In such short time, the cross-unit interference has not come into effect. Indeed,  conditional on $Y_{t}$, each pair of $(Y_{i(t+1)}, W_{it})$ are independent. Thus the behavior of the above IPW estimator in the micro-randomized trial does not differ too much from that in the standard randomized control trial. 

Similar consistency results of the IPW estimator for the direct effect have been established by \citet{savje2021average}; see also \citet{li2020random}. Nevertheless, unlike in \citet{savje2021average} where the outcome of subject $i$ depends on the current treatments of its neighbors, we assume in this paper that the outcome of subject $i$ depends on the outcomes of its neighbors in previous time periods. This gives rise to a difference in the statistical properties of the IPW estimator. 


\subsection{The long-term direct effect}
\label{section:LTD}

A more challenging question to ask is how the treatments influence outcomes in the long term. In this section, we focus on the long-term effect of one unit's treatment on its own outcome. For a time independent policy of treatment, it has been established in Proposition \ref{prop:MDP_stationary} that in the long run, the outcomes will converge in distribution to the stationary distribution $\mu(\pi)$. Therefore a simpler way of writing down the long-term effect is to express it in terms of the stationary distribution. 

More precisely, we are interested in
\begin{equation}
\tau_{\LDE}(\gamma_1, \gamma_2) 
 = \frac{1}{n} \sum_{i = 1}^n \p{\EE[\mu(\pi_i = \gamma_1, \pi_{-i})]{Y_i} - \EE[\mu(\pi_i = \gamma_2, \pi_{-i})]{Y_i} }. 
\end{equation}

To estimate this quantity, one key question is whether we can express this quantity in terms of the stationary distribution of the original experiments. In other words, if we have only implemented the experiments for a fixed $\pi$ for a long time but have never really changed our treatment probability, is it possible to estimate this estimand in an ``off-policy" sense? Theorem \ref{theo:LDE_char} gives a positive answer to the above questions. We establish in Theorem \ref{theo:LDE_char} that we can rewrite each term in $\tau_{\LDE}(\gamma_1, \gamma_2)$ in terms of expectations under the original stationary distribution $\mu(\pi)$.

To gain some intuition, consider the following approximation for the
``tilted" stationary distribution $\mu(\pi_i = \gamma, \pi_{-i})$, i.e., the stationary distribution corresponding to a treatment vector with $\pi_i = \gamma$ and the other $\pi_j$ remaining unchanged. Under this stationary distribution, 
\begin{equation}
\begin{split}
&\EE[\mu(\pi_i = \gamma, \pi_{-i})]{Y_i} \\
&\qquad = \EE[\mu(\pi_i = \gamma, \pi_{-i})]{f_i(Y_i, W_i, Z_i)}\\
 &\qquad=  \EE[\mu(\pi_i = \gamma, \pi_{-i})]{a_i(Z_i) +  b_i(Z_i) W_i + c_i (Z_i) Y_i+ d_i(Z_i) W_i Y_i}\\
&\qquad \approx \EE[\mu(\pi_i = \gamma, \pi_{-i})]{a_i(Z_i)} +  \EE[\mu(\pi_i = \gamma, \pi_{-i})]{b_i(Z_i)} \gamma\\
&\qquad \qquad   +  \EE[\mu(\pi_i = \gamma, \pi_{-i})]{c_i (Z_i)} \EE[\mu(\pi_i = \gamma, \pi_{-i})]{Y_i} +  \EE[\mu(\pi_i = \gamma, \pi_{-i})]{d_i(Z_i)} \gamma \EE[\mu(\pi_i = \gamma, \pi_{-i})]{Y_i},
\end{split}
\end{equation}
if $W_i$, $Y_i$ and $Z_i$ are roughly independent. 
Moving all terms involving $\EE[\mu(\pi_i = \gamma, \pi_{-i})]{Y_i}$ to the left hand side, we get
\begin{equation}
\EE[\mu(\pi_i = \gamma, \pi_{-i})]{Y_i} \approx \frac{ \EE[\mu(\pi_i = \gamma, \pi_{-i})]{a_i(Z_i)} +\EE[\mu(\pi_i = \gamma, \pi_{-i})]{b_i(Z_i)} \gamma }{1 - \EE[\mu(\pi_i = \gamma, \pi_{-i})]{c_i (Z_i)}  -  \EE[\mu(\pi_i = \gamma, \pi_{-i})]{d_i(Z_i)} \gamma}.
\end{equation}
Now if we assume that $Z_i$'s are not influenced too much by the treatment probability of unit $i$, then 
\begin{equation}
\begin{split}
\EE[\mu(\pi_i = \gamma, \pi_{-i})]{Y_i} 
&\approx \frac{ \EE[\mu(\pi_i = \gamma, \pi_{-i})]{a_i(Z_i)} +\EE[\mu(\pi_i = \gamma, \pi_{-i})]{b_i(Z_i)} \gamma }{1 - \EE[\mu(\pi_i = \gamma, \pi_{-i})]{c_i (Z_i)}  -  \EE[\mu(\pi_i = \gamma, \pi_{-i})]{d_i(Z_i)} \gamma}\\
&\approx \frac{ \EE[\mu(\pi)]{a_i(Z_i)} +\EE[\mu(\pi)]{b_i(Z_i)} \gamma }{1 - \EE[\mu(\pi)]{c_i (Z_i)}  -  \EE[\mu(\pi)]{d_i(Z_i)} \gamma}.
\end{split}
\end{equation}
The following result establishes that, given our assumptions, this heuristic in fact
correctly recovers the long-term direct effect.

\begin{theo}[Long-term direct effect characterization]
\label{theo:LDE_char}
Under Assumptions \ref{assu:MDP} - \ref{assu:contraction},
\begin{equation}
\label{eqn:LDE_char}
\EE[\mu(\pi_i = \gamma, \pi_{-i})]{Y_i}  =  \frac{\EE[\mu(\pi)]{a_i(Z_i) + b_i(Z_i) \gamma}}{\EE[\mu(\pi)]{1 - c_i(Z_i) - d_i(Z_i) \gamma}} + \oo\p{ \sqrt{L_n}}.
\end{equation}
\end{theo}
An immediate corollary of Theorem \ref{theo:LDE_char} is that we can express $\tau_{\LDE}(\gamma_1, \gamma_2)$ in term of expectations under $\mu(\pi)$ with a small error. 
\begin{equation}
\label{eqn:LDE_char2}
\tau_{\LDE}(\gamma_1, \gamma_2) = \frac{1}{n} \sum_{i=1}^n  \frac{\EE[\mu(\pi)]{a_i(Z_i) + b_i(Z_i) \gamma_1}}{\EE[\mu(\pi)]{1 - c_i(Z_i) - d_i(Z_i)\gamma_1}} - \frac{1}{n} \sum_{i=1}^n  \frac{\EE[\mu(\pi)]{a_i(Z_i) + b_i(Z_i) \gamma_2}}{\EE[\mu(\pi)]{1 - c_i(Z_i) - d_i(Z_i)\gamma_2}} +  \oo\p{\sqrt{L_n}}.  
\end{equation}
This characterization allows us to estimate $\tau_{\LDE}(\gamma_1, \gamma_2)$ using data from experiments with treatment probability $\pi$. Before going into details of the estimation problem, we make a few remarks on Theorem \ref{theo:LDE_char}. 

We first look at the error term $\sqrt{L_n}$. Assumption \ref{assu:contraction} requires that $L_n D_n < 1$. Thus as long as the largest node degree converges goes to infinity, the error term $\sqrt{L_n} < 1/\sqrt{D_n} \to 0$. Interestingly, the bound on the error term does not come from averaging causal effects from different units; indeed, Theorem \ref{theo:LDE_char} gives characterization for each single unit of the expected outcome under a different policy.

Another close look at the expression \eqref{eqn:LDE_char} shows that the long-term causal effect is in general not linear in the treatment probability. If we assume $\EE[\mu(\pi)]{a_i(Z_i)}>0, \dots, \EE[\mu(\pi)]{d_i(Z_i)} > 0$, then we can immediately verify that $\p{\EE[\mu(\pi)]{a_i(Z_i) + b_i(Z_i) \gamma}}/\p{\EE[\mu(\pi)]{1 - c_i(Z_i) - d_i(Z_i) \gamma}}$, as a function of $\gamma$, is increasing and convex. In the fitness app example, the above claim says that when a user responds positively to the encouraging messages, tends to behave similarly as yesterday, and is even more excited about running with both messages and the exercise from yesterday, then it is relatively easier to change her long-term running probability (habit) from 0.8 to 0.9 than from 0.1 to 0.2. 

Finally, we discuss estimation strategies for the long-term direct effect. Formula \eqref{eqn:LDE_char2} shows that it suffices to estimate $\EE[\mu(\pi)]{a_i(Z_{i t})}, \dots, \EE[\mu(\pi)]{d_i(Z_{i t})}$ for each $i$, which is equivalent to estimating $\EE[\mu(\pi)]{f_i(0,0,Z_i)}$, $\EE[\mu(\pi)]{f_i(0,1,Z_i)}$, $\EE[\mu(\pi)]{f_i(1,0,Z_i)}$ and $\EE[\mu(\pi)]{f_i(1,1,Z_i)}$. Here we take $\EE[\mu(\pi)]{f_i(0,0,Z_i)}$ as an example. Recall that $\EE{Y_{i(t+1)}\mid W_t, Y_t} = f_i(Y_i,W_i,Z_i)$. Therefore, a natural strategy is to take the time points at which $W_{it} = 0$ and $Y_{it} = 0$, and to take the average of the corresponding $Y_{i(t+1)}$'s. Define
\begin{equation}
\widehat{f_i(0,0)}= \frac{\frac{1}{T} \sum_{t = 1}^T Y_{i (t+1)} (1 - W_{i t}) (1 - Y _{i t})}{\frac{1}{T} \sum_{t = 1}^T (1 - W_{i t}) (1 - Y _{i t})}. 
\end{equation}
Proposition \ref{prop:MDP_stationary} implies that the Markov chain induced by Bernoulli treatments is ergodic. Therefore, when $T$ is large, average over the history approximates the stationary distribution well. Roughly, this implies $\frac{1}{T} \sum_{t = 1}^T Y_{i (t+1)} (1 - W_{i t}) (1 - Y _{i t}) \approx \EE[\mu(\pi)]{Y_{i (t+1)} (1 - W_{i t}) (1 - Y _{i t})}$ and same for the denominator. Thus $\widehat{f_i(0,0)} \approx \EE[\mu(\pi)]{Y_{i(t+1)}\mid W_{it} =0, Y_{it} = 0}  = \EE[\mu(\pi)]{f_i(0,0,Z_{it}) \mid W_{it} =0, Y_{it} = 0}$. Again, if $Z_i$ does not depend much on $Y_i$ under $\mu(\pi)$, then $\EE[\mu(\pi)]{f_i(0,0,Z_{it}) \mid W_{it} =0, Y_{it} = 0} \approx \EE[\mu(\pi)]{f_i(0,0,Z_{it})}$. 

Put formally, for $w \in \cb{0,1}$ and $y \in \cb{0,1}$, define
\begin{equation}
\widehat{f_i(y,w)}= \frac{\frac{1}{T} \sum_{t = 1}^T Y_{i (t+1)} \mathbbm{1} \cb{W_{i t} = w, Y _{i t} = y}}{\frac{1}{T} \sum_{t = 1}^T \mathbbm{1} \cb{W_{i t} = w, Y _{i t} = y}}. 
\end{equation}
Proposition \ref{prop:estimate_LDE} establishes that $\widehat{f_i(y,w)}$ is close to the goal $\EE[\mu\p{\pi}]{f_i(y,w, Z_i)}$. 

\begin{prop}
\label{prop:estimate_LDE}
Under Assumptions \ref{assu:MDP} - \ref{assu:contraction}, assume further that there exist constant $C_{y}, C_{\pi} > 0$ such that $C_y<\EE[\mu(\pi)]{Y_{it}} < 1- C_y$ and $C_{\pi} < \pi_i < 1-C_{\pi}$, then
\[ \EE{\p{\widehat{f_i(y,w)} - \EE[\mu\p{\pi}]{f_i(y,w, Z_i) }}^2} \leq C_1 \p{\frac{1}{T} + L_n},  \]
for some constant $C_1$ not depending on $i$ or $n$. 
\end{prop}

With the estimators for $ \EE[\mu\p{\pi}]{f_i(y,w, Z_i)}$, we can transform them into estimators for $\EE[\mu]{a_i(Z_i)}$, $\dots,$ $\EE[\mu]{d_i(Z_i)}$, and combine them to form an estimator for the long-term direct effect. We take $\hat a_i = \widehat{f_i(0,0)}$, $\hat b_i = \widehat{f_i(0,1)} - \widehat{f_i(0,0)}$, $\hat c_i = \widehat{f_i(1,0)} - \widehat{f_i(0,0)}$, $\hat d_i = \widehat{f_i(1,1)} + \widehat{f_i(0,0)} - \widehat{f_i(0,1)} - \widehat{f_i(1,0)}$, and
\begin{equation}
\hat{\tau}_{\LDE}(\gamma_1, \gamma_2) = \frac{1}{n} \sum_{i=1}^n  \frac{\hat{a}_i + \hat{b}_i \gamma_1}{1 - \hat{c}_i - \hat{d}_i\gamma_1} - \frac{1}{n} \sum_{i=1}^n  \frac{\hat{a}_i + \hat{b}_i \gamma_2}{1 - \hat{c}_i - \hat{d}_i\gamma_2}. 
\end{equation}
It then follows directly from Proposition \ref{prop:estimate_LDE} that $\hat{\tau}_{\LDE}(\gamma_1, \gamma_2) $ is consistent for $\tau_{\LDE}(\gamma_1, \gamma_2)$ when the time horizon $T \to \infty$ and $L_n \to 0$. 

\begin{coro}[Long-term direct effect estimation]
Under Assumptions \ref{assu:MDP} - \ref{assu:contraction}, assume further that there exist constant $C_{y}, C_{\pi} > 0$ such that $C_y<\EE[\mu(\pi)]{Y_{it}} < 1- C_y$ and $C_{\pi} < \pi_i < 1-C_{\pi}$, then
\[\hat{\tau}_{\LDE}(\gamma_1, \gamma_2) -  \tau_{\LDE}(\gamma_1, \gamma_2)
= \oo_p\p{\frac{1}{\sqrt{T} } + \sqrt{L_n}}. \]
\end{coro}


\subsection{The long-term total effect}
\label{section:LTE}
In this section, we study estimation of the long-term total effect of the treatment. The question of interest is what happens in the long term to the outcomes if we simultaneously change the treatment probabilities for everyone. Again, as argued in the previous section, to quantify a long-term effect, we study how the stationary distribution depends on treatment probabilities. This is because in the long term, the distribution of outcomes will eventually converge to the stationary distribution. We are interested in
\begin{equation}
\tau_{\LTE}(\pi_1, \pi_2) = \frac{1}{n} \sum_{i = 1}^n \p{\EE[\mu(\pi_1)]{Y_i} - \EE[\mu(\pi_2)]{Y_i} }.
\end{equation}
Here we focus on the case where $\pi_2 = \pi$, i.e., we take $\pi_2$ to be the true probability of which the treatment is given in the experiment. We write $\pi_1 = \pi + \Delta \bv$, where $\Delta \in \mathbb{R}$ captures the scale of the difference between $\pi_1$ and $\pi_2$, while $\bv \in \mathbb{R}^n$ is the direction of the difference. We enforce the constraint $\Norm{\bv} = \sqrt{n}$. 
Two common choices of $\bv$ are $\bv = \mathbf{1} = (1,1,\dots, 1)^T$ and $\bv = \pi\sqrt{n/\Norm{\pi}}$. In the shopping habit example, imagine store A has been running the micro-randomized trial with different coupon sending probabilities to different people. 
Taking $\bv = \mathbf{1}$ corresponds to asking the question of what the effect is of increasing coupon sending probability by $0.2$ to everyone? The other choice of $\Norm{\bv} \propto \pi$ has a different meaning. It focuses on the effect of increasing everyone's coupon sending probability by, for instance, $30\%$ of the original probability. There are, of course, other choices of $\bv$. For example, store A may be curious about a policy, which unlike the current one, send coupons to everyone with a fixed probability 0.25. 

We will then move on to study the characterization and estimation of $\tau_{\LTE}(\pi + \Delta \bv, \pi)$. 
We have established in Section \ref{section:MDP} that the stationary distribution of MDP \eqref{eqn:MDP} is close to the fixed point of system \eqref{eqn:dynamic_system}. We start with a simpler task: analyzing the estimand under the dynamical system \eqref{eqn:dynamic_system}. 
We adopt the notation from Section \ref{section:mean_field} and let $\Ps$ be the fixed point of system \eqref{eqn:dynamic_system}. Here we write $\Ps(\pi)$ to emphasize its dependence on the treatment probability $\pi$. 
Recall that $P^*$ satisfies 
\begin{equation}
\label{eqn:p_star_pi}
\Ps_{i}(\pi) = a_i\p{\Qs_{i}(\pi)} + b_i\p{\Qs_{i}(\pi)} \pi_i +  c_i\p{\Qs_{i}(\pi)} \Ps_{i}(\pi) + d_i\p{\Qs_{i}(\pi)} \pi_i \Ps_{i}(\pi). 
\end{equation}
Here the slightly different but easier goal is:
\begin{equation}
\widetilde{\tau}_{\LTE}( \pi + \Delta \bv,  \pi) = \frac{1}{n} \sum_{i = 1}^n \p{\Ps_i(\pi + \Delta\bv) - \Ps_i(\pi) } .
\end{equation}
When $\Delta$ is small, the above is close to $\Delta/n \sum_{i = 1}^n \p{\nabla_{\pi} \Ps_i(\pi) \trans \bv}$. 

In order to evaluate $\nabla_{\pi} \Ps_i(\pi) \trans \bv$, we take a closer look at equation \eqref{eqn:p_star_pi}. 
Evaluated at $\pi + \Delta \bv$, equation \eqref{eqn:p_star_pi} becomes
\begin{equation}
\label{eqn:fix_point}
\begin{split}
 \Ps_{i}(\pi + \Delta \bv) 
&= a_i\p{\Qs_{i}(\pi +\Delta\bv)} + b_i\p{\Qs_{i}(\pi + \Delta\bv)} (\pi_i + \Delta v_i) + \\
&  \qquad  c_i\p{\Qs_{i}(\pi + \Delta\bv)} \Ps_{i}(\pi + \Delta\bv) + d_i\p{\Qs_{i}(\pi + \Delta\bv)} (\pi_i + \Delta\bv) \Ps_{i}(\pi + \Delta\bv). 
\end{split}
\end{equation}
Taking derivative with respect to $\Delta$ of both hand sides and evaluating at $\Delta = 0$, we get
\begin{equation}
\begin{split}
 \pss_{i}(\pi) 
&= b_i\p{\Qs_{i}(\pi )} v_i + d_i\p{\Qs_{i}(\pi)} \Ps_{i}(\pi) v_i + \sqb{c_i\p{\Qs_{i}(\pi)} + d_i\p{\Qs_{i}(\pi)} \pi_i }  \pss_{i}(\pi)+ \\
& \qquad \sqb{a_i'\p{\Qs_{i}(\pi)} + b_i'\p{\Qs_{i}(\pi )} \pi_i + c_i'\p{\Qs_{i}(\pi)} \Ps_{i}(\pi) + d_i'\p{\Qs_{i}(\pi)} \pi_i \Ps_i(\pi)  } \qss_{i}(\pi )\\
\end{split}
\end{equation}
where $\pss_i(\pi) = \nabla_{\pi} \Ps_i(\pi) \trans \bv$ and $\qss_i(\pi) = \nabla_{\pi} \Qs_i(\pi) \trans \bv$. Since $\Qs_i(\pi) = \sum_{j \in \ngb_i} \Ps_j(\pi)$, by linearity of derivatives, we have $\qss(\pi) = \sum_{j \in \ngb_i} \pss_j(\pi)$. Thus 
\begin{equation}
\begin{split}
 \pss_{i}(\pi) 
&= b_i\p{\Qs_{i}(\pi )}v_i  + d_i\p{\Qs_{i}(\pi)} \Ps_{i}(\pi) v_i + \sqb{c_i\p{\Qs_{i}(\pi)} + d_i\p{\Qs_{i}(\pi)} \pi_i }  \pss_{i}(\pi)+ \\
& \qquad \sqb{a_i'\p{\Qs_{i}(\pi)} + b_i'\p{\Qs_{i}(\pi )} \pi_i + c_i'\p{\Qs_{i}(\pi)} \Ps_{i}(\pi) + d_i'\p{\Qs_{i}(\pi)} \pi_i \Ps_i(\pi)}\sum_{j \in \ngb_i} \pss_j(\pi). 
\end{split}
\end{equation}
Note that this question holds for every $i$. Thus, now we have a set of linear equations for $\pss_i(\pi)$'s. Let $\pss(\pi)$ be the vector of $(\pss_1(\pi), \dots, \pss_n(\pi))$. By solving the linear system, we can write $\pss(\pi)$ in terms of $\Ps(\pi)$, the functions $a_i$, $b_i$, $c_i$, $d_i$, and their derivatives $a_i'$, $b_i'$, $c_i'$, $d_i'$. Specifically, 
\begin{equation}
\label{eqn:first_dir_matrix}
\pss(\pi) = (I - D A - W) ^{-1} \bu,  
\end{equation}
where $A$ is the adjacency matrix of the interference graph, \sloppy{$D = \operatorname{diag}\big(a_i'\p{\Qs_{i}(\pi)} + b_i'\p{\Qs_{i}(\pi )} \pi_i + c_i'\p{\Qs_{i}(\pi)} \Ps_{i}(\pi) + d_i'\p{\Qs_{i}(\pi)} \pi_i \Ps_i(\pi)\big)$, $W =  \operatorname{diag} \big( c_i\p{\Qs_{i}(\pi)} + d_i\p{\Qs_{i}(\pi)} \pi_i\big)$, and $\bu = \operatorname{vec} \big( v_i (b_i(\Qs_{i}(\pi )) + d_i\p{\Qs_{i}(\pi)} \Ps_{i}(\pi))  \big)$.} Summarizing the above findings, we have
\begin{equation}
\label{eqn:theo_LTE_char0}
\begin{split}
\widetilde{\tau}_{\LTE}( \pi + \Delta \bv,  \pi) 
&= \frac{1}{n}\sum_{i = 1}^n\p{ \Ps_i(\pi + \Delta\bv) - \Ps_i(\pi) }
\approx \frac{\Delta}{n}\sum_{i = 1}^n \p{\nabla_{\pi} \Ps_i(\pi) \trans \bv}
= \frac{\Delta}{n}\sum_{i = 1}^n \pss_i(\pi) \\
&= \frac{\Delta}{n}\mathbf{1}\trans \pss(\pi)
= \frac{\Delta}{n} \mathbf{1}\trans (I - D A - W) ^{-1} \bu. 
\end{split}
\end{equation}

Given the above characterization for $\widetilde{\tau}_{\LTE}( \pi + \Delta \bv,  \pi) $, there are two remaining questions: 
How close is the simpler goal $\widetilde{\tau}_{\LTE}(\pi + \Delta \bv,  \pi)$ to our real goal $\tau_{\LTE}( \pi + \Delta \bv,  \pi)$?
How close is $\frac{\Delta}{n} \mathbf{1}\trans (I - D A - W) ^{-1} \bu$ to our real goal $\tau_{\LTE}( \pi + \Delta \bv,  \pi)$?

\begin{theo}
\label{theo:LTE_char1}
Under Assumptions \ref{assu:MDP} - \ref{assu:contraction}, 
\begin{equation}
\label{eqn:theo_LTE_char1}
\tau_{\LTE}( \pi + \Delta \bv,  \pi) = \widetilde{\tau}_{\LTE}( \pi + \Delta \bv,  \pi) +  \oo\p{ \sqrt{L_n}}. 
\end{equation}
\end{theo}
Theorem \ref{theo:LTE_char1} provides an answer to the first question. It establishes that if the Lipschitz constant $L_n \to 0$, then the difference between $\tau_{\LTE}( \pi + \Delta \bv,  \pi)$ and $\widetilde{\tau}_{\LTE}( \pi + \Delta \bv,  \pi)$ converges to 0. This result follows directly from Theorem \ref{theo:mean_field}, where we showed that the stationary distribution of MDP \eqref{eqn:MDP} is close to the fixed point of the dynamical system \eqref{eqn:dynamic_system}. 

To answer the second question, we need stronger assumptions. Specifically, in order for the approximation step $\frac{1}{n}\sum_{i = 1}^n\p{ \Ps_i(\pi + \Delta\bv) - \Ps_i(\pi) }
\approx \frac{\Delta}{n}\sum_{i = 1}^n \p{\nabla_{\pi} \Ps_i(\pi) \trans \bv}$ to be accurate, we need $\Delta$ to be small, and we need to impose smoothness assumptions on the functions $f_i$. Theorem \ref{theo:LTE_char2} establishes that when $\Delta$ does not converge to zero too quickly, then $\frac{1}{n} \mathbf{1}\trans (I - D A - W) ^{-1} \bu$ approximates the long-term total effect (scales by $1/\Delta$) well.
\begin{assu}[Smoothness]
\label{assu:smooth}
The functions $f_i$'s satisfy $\abs{f''_i(y,w,z)} \leq L_{2,n}$ for any $y \in \cb{0,1}$, $w \in \cb{0,1}$ and $z \in \RR_+$, where all derivatives  are taken with
respect to the third argument. There exists a constant $C_{L,2}$ such that $L_{2,n} D_n^2 \leq C_{L,2}$, where $D_n$ is defined in Assumption \ref{assu:largest_degree}. 
\end{assu}

\begin{theo}[Long-term total effect characterization]
\label{theo:LTE_char2}
Under Assumptions \ref{assu:MDP} - \ref{assu:smooth},
assume further that $\Delta$ varies with $n$. We write $\Delta_n$ to emphasize such dependency. Assume that there exists a constant $C_v$ such that $\max \abs{v_i} \leq C_v$. Then
\begin{equation}
\label{eqn:LTE_char1}
 \frac{\tau_{\LTE}( \pi + \Delta_n \bv,  \pi)}{\Delta_n} = \frac{1}{n} \mathbf{1}\trans (I - D A - W) ^{-1} \bu + \oo\p{\frac{\sqrt{L_n}}{\Delta_n} + \Delta_n}. 
\end{equation}
\end{theo}

Next, we study estimation of the quantity $ \frac{1}{n} \mathbf{1}\trans (I - D A - W) ^{-1} \bu$. The estimating strategy is fairly straight forward: We estimate $D, W$ and $\bu$ separately. This essentially requires two things: estimating the function values $a_i(\Qs(\pi)), \dots, d_i(\Qs(\pi))$, and estimating the derivatives $a_i'(\Qs(\pi)) \dots d_i'(\Qs(\pi))$. The first task can be easily done using the estimators from Section \ref{section:LTD}. Specifically, we have consistent estimators for $\EE[\mu(\pi)]{a_i(Z_{it})}, \dots, \EE[\mu(\pi)]{d_i(Z_{it})},$ when $T$ is large and $\sqrt{L_n}$ is small. Theorem \ref{theo:mean_field2} further ensures that $\EE[\mu(\pi)]{a_i(Z_{it})}$ is close to $a_i(\Qs_{i}(\pi))$. For the second task, we need to estimate the derivative of the functions $a_i, \dots, d_i$ evaluated at $\Qs_i(\pi)$. Specifically, the quantity of interest is the diagonal of the $D$ matrix, whose $i$-th element is
\begin{equation}
\begin{split}
&a_i'\p{\Qs_{i}(\pi)} + b_i'\p{\Qs_{i}(\pi )} \pi_i + c_i'\p{\Qs_{i}(\pi)} \Ps_{i}(\pi) + d_i'\p{\Qs_{i}(\pi)} \pi_i \Ps_i(\pi)\\
&\qquad \qquad=  f_i'(\Ps_i(\pi), \pi_i, \Qs_i(\pi))\\
 &\qquad\qquad = (1 -  \pi_i )( 1 - \Ps_i(\pi)) f_i'(0, 0, \Qs_i(\pi)) + \pi_i(1 - \Ps_i(\pi)) f_i'(0, 1, \Qs_i(\pi)) +\\
&\qquad\qquad \qquad \qquad \Ps_i(\pi) (1 - \pi_i)  f_i'(1, 0, \Qs_i(\pi)) +  \Ps_i(\pi) \pi_i f_i'(1, 1, \Qs_i(\pi)). 
\end{split}
\end{equation}
Thus the goal becomes estimating $f_i'(y, w, \Qs_i(\pi))$ for $y, w \in \cb{0,1}$. 

To estimate the derivative, one natural idea is to run a regression of $Y_{i(t+1)}$ on $Z_{it}$ conditioning on $Y_{it} = y$ and $W_{it} = w$ using data from different time points. Intuitively, the method works because of the following reason. Imagine the ideal case where $f_i$ is a linear function in $z$, i.e., $f_i(y, w, z) = \beta_i(y, w) z  + \beta_{0,i}(y,w)$. Then $f'_i(y,w,z) = \beta_i(y, w)$. Since $\EE{Y_{i(t+1)} \mid Y_t, W_t} = f_i(Y_{it}, W_{it}, Z_{it})$, the regression would give a good estimator for $\beta(y, w)$. Now there are two problems remaining. The first is that $f_i$ is not linear, and the second is that the data points are not i.i.d., and thus standard results from linear regression do not apply directly. The first problem can be solved since the mean-field results give that $Z_{it}$ would be close to sum of independent random variables. Together with the fact that the Markov chain converges in distribution to the stationary distribution, one can establish that $Z_{it}$ won't be far from $ \Qs_i(\pi)$. In other words, the fluctuation in $Z_{it}$ is small, so as long as the function $f_i$ is smooth enough, locally around $\Qs_i(\pi)$, $f_i$ can be treated as a linear function. The second problem can be solved using the ergodic property of the Markov chain.

We make the above arguments precise below. Let $\delta_T$ be a sequence of positive numbers such that $\delta_T \to 0$ as $T \to \infty$.
For $y, w \in \cb{0,1}$, let
\begin{equation}
\widehat{f_i'(y,w)}_{\delta_T} = \frac{\sum_{t = 1}^T \mathbbm{1}\cb{Y_{it} = y, W_{it} = w} \p{Y_{i(t+1)} - \bar{Y_i}(y,w)} \p{Z_{it} - \bar{Z_i}(y,w) }}{D_n T \delta_T \vee \sum_{t = 1}^T \mathbbm{1} \cb{Y_t = y, W_t = w}  \p{ Z_{it} -  \bar{Z_i}(y,w) }^2},\footnote{The term $D_n T \delta$ is used to guarantee that the denominator will not be too close to zero.}
\end{equation}
where 
\begin{equation}
\bar{Y_i}(y,w) = \frac{\sum_{t = 1}^T \mathbbm{1}\cb{Y_{it} = y, W_{it} = w} Y_{i(t+1)} }{\sum_{t = 1}^T \mathbbm{1} \cb{Y_t = y, W_t = w} }, \textnormal{ and}
\end{equation}
\begin{equation}
\bar{Z_i}(y,w) = \frac{\sum_{t = 1}^T \mathbbm{1}\cb{Y_{it} = y, W_{it} = w} Z_{i(t+1)} }{\sum_{t = 1}^T \mathbbm{1} \cb{Y_t = y, W_t = w} }. 
\end{equation}
Theorem \ref{theo:LTE_estimation} shows that $\widehat{f_i'(y,w)}_{\delta_T} $ is a good estimator of $f_i'(y,w, \Qs_i(\pi))$ when we observe a long enough trajectory of treatments and outcomes. 

\begin{theo}
\label{theo:LTE_estimation}
Under Assumptions \ref{assu:MDP} - \ref{assu:smooth}. Suppose that there exist constants $C_{\pi}, C_y, C_P, C_{f}, C_l > 0$ such that $\pi_i \in (C_{\pi}, 1- C_{\pi})$, $\EE[\mu(\pi)]{Y_i} \in (C_y, 1 - C_y)$, $\Ps_i(\pi) \in (C_P, 1 - C_P)$, $\abs{\ngb_i} \geq C_l D_n$, and $f_i \in (C_{f}, 1- C_{f})$. Then
\[\widehat{f_i'(y,w)}_{\delta_T}   = f_i'(y,w, \Qs_i(\pi)) + \epsilon_i,\]
where $\EE{\epsilon_i^2} =  \oo\p{1/(T \delta_T^2) + D_n^{-3}}$. 
\end{theo}

Finally, we combine our estimators for $a_i(\Qs(\pi)), \dots, d_i(\Qs(\pi))$, and $a_i'(\Qs(\pi)) \dots d_i'(\Qs(\pi))$ and produce a final estimator for $\frac{1}{n} \mathbf{1}\trans (I - D A - W) ^{-1} \bu$. 
Let $\eta_n \in (0,1)$ and $\kappa_n \in (0,1)$ be two sequences of numbers such that $\eta_n \to 0$ and $\kappa_n \to 0$ as $n \to \infty$. 
Let $\hat{P}_i = \sum_{i = 1}^T Y_{it}/T$, 
$\hat{D}_i = (1 -  \pi_i )( 1 - \Ph_i(\pi)) \widehat{f_i'(0,0)}  + \pi_i(1 - \Ph_i(\pi)) \widehat{f_i'(0,1)}  + \Ph_i(\pi) (1 - \pi_i)  \widehat{f_i'(1,0)}  +  \Ph_i(\pi) \pi_i \widehat{f_i'(1,1)}$,
$\Dh = \operatorname{diag}(\hat{D}_i)$, 
$\hat{\omega}_{i, \kappa_n} =  \min(1 - \kappa_n, \hat{c}_i + \hat{d}_i \pi_i)$,\footnote{Similar to the role of $\delta_T$ in the previous theorem,  $\kappa_n$ and $\eta_n$ are used to guarantee that the matrix $(M_{\eta_n} - \Dh A - \Wh_{\kappa_n})$ is invertible and that its smallest eigenvalue (in absolute value) is not too close to zero.}
$\Wh_{\kappa_n} = \operatorname{diag}(\hat{\omega}_{i, \kappa_n})$, 
$\uh = \operatorname{vec} \big( \hat{b}_i v_i + \hat{d}_i \Ph_{i}(\pi) v_i \big)$, and $M_{\eta_n} = \operatorname{diag}\p{\max\p{1, \hat{d}_i D_n/(1- \eta_n) + \hat{\omega}_{i, \kappa_n}} }$. Combining the above definitions, we define
\begin{equation}
\hat{\tau}_{\LTE}( \pi + \Delta \bv,  \pi)
= \frac{\Delta}{n} \mathbf{1}\trans (M_{\eta_n} - \Dh A - \Wh_{\kappa_n}) ^{-1} \uh. 
\end{equation}
Theorem \ref{theo:LTE_estimation2} shows that when the functions $f_i$'s are smooth enough, we observe long enough trajectory, and the difference between treatment probabilities is of the right scale, then our estimator $\hat{\tau}_{\LTE}( \pi + \Delta \bv,  \pi)$ is close to the long-term total effect under large sample asymptotics. 

\begin{theo}[Long-term total effect estimation]
\label{theo:LTE_estimation2}
Under the conditions of Theorem \ref{theo:LTE_estimation}, 
assume further that $\Delta$ varies with $n$. We write $\Delta_n$ to emphasize such dependency. Furthermore,  assume that there exists a constant $C_v$ such that $\max \abs{v_i} \leq C_v$. Then
\begin{equation}
\frac{\hat{\tau}_{\LTE}( \pi + \Delta_n \bv,  \pi)}{\Delta_n} = \frac{\tau_{\LTE}( \pi + \Delta_n \bv,  \pi)}{\Delta_n} + \oo_p\p{\frac{1}{\Delta_n \sqrt{D_n}} + \Delta_n + \frac{1}{\eta_n \kappa_n} \p{\frac{D_n}{\sqrt{T}\delta_t} +\frac{1}{\sqrt{D_n}}}}. 
\end{equation}
\end{theo}

\section{Discussion}

The problem of treatment effect estimation in systems where agents may interfere
with each other over a network has received a considerable amount of attention in
the literature \citep{aronow2017estimating,athey2018exact,leung2020treatment,li2020random}.
Available work on this problem, however, is focused on a static setting with a single
time-step, whereas several important applications---especially mobile health
applications---rely on micro-randomized trials where treatment dynamics play a
key role \citep{liao2016sample,klasnja2015microrandomized}.
In this paper we took a first step towards studying network interference
in a dynamic (Markovian) setting by characterizing a number of natural causal estimands,
and by proposing consistent estimators for them.

The core technical results in this paper were all built around a mean-field
approximation result for the stationary distribution of the Markov process
we used to model micro-randomized trials with network interference
(Theorems \ref{theo:mean_field} and \ref{theo:mean_field2}). Specifically, we showed
that, at stationarity, our system looks as though it were governed by independent
draws from Bernoulli random variables parametrized by a fixed point $P^{\star}$ of the system
\eqref{eqn:dynamic_system}.

We note that the fixed point $P^{\star}$ of the system \eqref{eqn:dynamic_system} is closely related
to the Nash equilibrium of some network games \citep{ballester2006s,parise2019graphon}.
Consider an $n$-agent network game, where $P_i$ represents the strategy of the $i$-th agent.
If we define the payoff function of agent $i$ to be $U_i = -(P_i - f_i(P_i, \pi_i, Q_i))^2$
where, like in our setting, $Q_i = \sum_{j \in \ngb_i} P_j$ sums the strategies of the $i$-th unit's neighbors,
the fixed point $P^{\star}$ is a coordinate-wise critical point of the payoff function.
\citet{parise2019graphon} then discuss settings where, given regularity conditions analogous to our
Assumption \ref{assu:contraction}, $P^{\star}$ is the unique Nash equilibrium of the network game.
Interestingly, in network games it is natural to assume that agents reason directly about each
others' strategies (rather then stochastic realized outcomes), and so the system \eqref{eqn:dynamic_system}
arises directly from the model; in contrast, in our setting, the system \eqref{eqn:dynamic_system} is only
useful in a large-sample (mean-field) approximation to the natural Markov process that arises from modeling problem primitives.
Further investigation of connections between dynamic treatment effect estimation
under network interference and network games may lead to new insights relevant to both
problem settings.


\ifEC
\bibliographystyle{ACM-Reference-Format}
\bibliography{references}

\else

\bibliographystyle{plainnat}
\bibliography{references}

\fi

\newpage
\begin{appendix}
\section{Proofs}

\subsection{Some Lemmas}

\begin{lemm}[$L_1$ Contraction]
\label{lemma:L1_contraction}
Consider two processes $X_t$ and $Y_t$ satisfying Assumptions \ref{assu:MDP} and \ref{assu:Bern}. Assume that at each time $t$, they share the same treatment vector $W_t$ and random seed, i.e., conditional on $X_t$ and $Y_t$, there exists $U_{it} \sim \operatorname{Unif}[0,1]$ independently, such that $Y_{i(t+1)} = 1$ if $U_{it} \leq f_i(Y_{it}, W_{it}, Z_{it})$ and $X_{i(t+1)} = 1$ if $U_{it} \leq f_i(X_{it}, W_{it}, V_{it})$, where $V_{it} = \sum_{j \in \mathcal{N}_i} X_{jt}$ and $Z_{it} = \sum_{j \in \mathcal{N}_i} Y_{jt}$. Under Assumptions \ref{assu:bound_fi} - \ref{assu:contraction}, we have 
\begin{equation}
\EE{ \Norm{X_{t+1} - Y_{t+1}}_1} \leq C \EE{\Norm{X_t - Y_t}},
\end{equation}
where $C$ is the contraction constant in Assumption \ref{assu:contraction}. 
Furthermore, 
\begin{equation}
\label{eqn:inequality_contraction_WL1}
W_{L_1}(\law(X_{t+1}), \law(Y_{t+1})) \leq C W_{L_1}(\law(X_t), \law(Y_t)). 
\end{equation}
\end{lemm}

\begin{proof}
The term $\EE{ \Norm{X_{t+1} - Y_{t+1}}_1}$ can be rewritten into
\[
\begin{split}
\EE{ \Norm{X_{t+1} - Y_{t+1}}_1} 
&= \sum_{i = 1}^n \EE{ \abs{X_{i(t+1)} - Y_{i(t+1)}}}\\
&= \sum_{i = 1}^n \EE{ \abs{f_i(X_{it}, W_{it}, V_{it}) - f_i(Y_{it}, W_{it}, Z_{it})}}\\
& \leq L_n \sum_{i = 1}^n\EE{\abs{V_{it} - Z_{it}}} + B \EE{\abs{X_{it} - Y_{it}}},
\end{split}
\]
where $V_{it} = \sum_{j \in \mathcal{N}_i} X_{jt}$ and $Z_{it} = \sum_{j \in \mathcal{N}_i} Y_{jt}$. Note that the term $\sum_i \abs{V_{it} - Z_{it}}$ can be further decomposed into 
\[
\begin{split}
\sum_{i = 1}^n \abs{V_{it} - Z_{it}} &= \sum_{i = 1}^n \sum_{j \in \mathcal{N}_i} \abs{X_{jt} - Y_{jt}} 
= \sum_{j = 1}^n  \sum_{i \in \mathcal{N}_j} \abs{X_{jt} - Y_{jt}}\\
&= \sum_{j = 1}^n \abs{\mathcal{N}_j}  \abs{X_{jt} - Y_{jt}} \leq D_n \Norm{X_t - Y_t}_1. 
\end{split}
\]
Therefore,
\[\EE{ \Norm{X_{t+1} - Y_{t+1}}_1}  \leq (L_n D_n + B) \Norm{X_t - Y_t}_1 \leq C\EE{\Norm{X_t - Y_t}}. \]

The inequality \eqref{eqn:inequality_contraction_WL1} can be shown easily by taking the coupling of $X_t$ and $Y_t$ such that $\EE{\Norm{X_t - Y_t}_1} = W_{L_1}(\law(X_t), \law(Y_t))$. 
\end{proof}

\begin{lemm}[$d_E$ Contraction]
\label{lemma:contraction_d_e}
Under the conditions of Lemma \ref{lemma:L1_contraction}, we have 
\begin{equation}
d_E(X_{t+1}, Y_{t+1})\leq C d_E(X_t , Y_t),
\end{equation}
where $C$ is the contraction constant in Assumption \ref{assu:contraction}. 
Furthermore, 
\begin{equation}
\label{eqn:inequality_contraction_Wde}
W_{d_E}(\law(X_{t+1}), \law(Y_{t+1})) \leq C W_{d_E}(\law(X_t), \law(Y_t)). 
\end{equation}
\end{lemm}

\begin{proof}
Recall that $d_E(X_{t+1}, Y_{t+1}) = \max_i \sum_{j \in \mathcal{N}_i} \EE{\abs{X_{j(t+1)} - Y_{j(t+1)}}}$. We use a similar decomposition as in Lemma \ref{lemma:L1_contraction}: 
\[
\begin{split}
\max_i \sum_{j \in \mathcal{N}_i} \EE{\abs{X_{j(t+1)} - Y_{j(t+1)}}}
&=  \max_i \sum_{j \in \mathcal{N}_i} \EE{ \abs{f_j(X_{jt}, W_{jt}, V_{jt}) - f_j(Y_{jt}, W_{jt}, Z_{jt})}}\\
& \leq \max_i \sum_{j \in \mathcal{N}_i} \p{L_n \EE{\abs{V_{jt} - Z_{jt}}} + B \EE{\abs{X_{jt} - Y_{jt}}}}\\
& \leq L_n \max_i \sum_{j \in \mathcal{N}_i} \EE{\abs{V_{jt} - Z_{jt}}} + B \max_i \sum_{j \in \mathcal{N}_i} \EE{\abs{X_{jt} - Y_{jt}}}\\
& = L_n \max_i \sum_{j \in \mathcal{N}_i} \EE{\abs{V_{jt} - Z_{jt}}} + B d_E(X_{t}, Y_{t}),
\end{split}
\]
where $V_{jt} = \sum_{k \in \mathcal{N}_j} X_{kt}$ and $Z_{jt} = \sum_{k \in \mathcal{N}_j} Y_{kt}$. 
The term $\sum_{j \in \mathcal{N}_i} \EE{\abs{V_{jt} - Z_{jt}}} $ satisfies
\[
\begin{split}
\sum_{j \in \mathcal{N}_i} \EE{\abs{V_{jt} - Z_{jt}}} \leq \sum_{j \in \mathcal{N}_i} \sum_{k \in \mathcal{N}_j} \EE{\abs{X_{kt} - Y_{kt}}} \leq D_n d_E(X_{t}, Y_{t}). 
\end{split}
\]
Thus $\max_i \sum_{j \in \mathcal{N}_i} \EE{\abs{V_{jt} - Z_{jt}}} \leq D_n d_E(X_{t}, Y_{t})$ as well. 
Therefore,
\[
\begin{split}
d_E(X_{t+1}, Y_{t+1}) &=  \max_i \sum_{j \in \mathcal{N}_i} \EE{\abs{X_{j(t+1)} - Y_{j(t+1)}}} \leq L_n D_n d_E(X_{t}, Y_{t}) + B d_E(X_{t}, Y_{t}) \\
&\leq C d_E(X_{t}, Y_{t}).  
\end{split}
\]

Again, as in Lemma \ref{lemma:L1_contraction}, the inequality \eqref{eqn:inequality_contraction_Wde} can be shown easily by taking the coupling of $X_t$ and $Y_t$ such that $d(X_t, Y_t) = W_{d_E}(\law(X_t), \law(Y_t))$. 
\end{proof}

\begin{lemm}[$d_{E,3}$ Contraction]
\label{lemma:contraction_d_e3}
Under the conditions of Lemma \ref{lemma:L1_contraction}, we have 
\begin{equation}
d_{E,3}(X_{t+1}, Y_{t+1}) \leq C d_{E,3}(X_t , Y_t) + 1,
\end{equation}
where $C$ is the contraction constant in Assumption \ref{assu:contraction}. 
Furthermore, 
\begin{equation}
\label{eqn:inequality_contraction_Wde3}
W_{d_E}(\law(X_{t+1}), \law(Y_{t+1})) \leq C W_{d_E}(\law(X_t), \law(Y_t)) + 1. 
\end{equation}
\end{lemm}

\begin{proof}
By definition, 
\[
d_{E,3}(X_{t+1}, Y_{t + 1})^3 = \max_i \EE{\Big( \sum_{j \in \ngb_i} \abs{X_{j(t+1)} - Y_{j(t+1)}} \Big)^3}. 
\]
By construction, we know that condition on $W_t$ and $Y_t$, $\abs{X_{j(t+1)} - Y_{j(t+1)}} \sim \operatorname{Ber}(a_{jt})$, where
\[ a_{jt} = \abs{f_j(X_{jt}, W_{jt}, V_{jt}) - f_j(Y_{jt}, W_{jt}, Z_{jt})},\]
$V_{jt} = \sum_{k \in \ngb_j} X_{kt}$ and $Z_{jt} = \sum_{k \in \ngb_j} Y_{kt}$. 
Let $\epsilon_{j(t+1)} = \abs{X_{j(t+1)} - Y_{j(t+1)}} - a_{jt}$. We note that $a_{jt}$ is measurable with respect to $W_t$ and $Y_t$, whereas $\epsilon_{j(t+1)}$'s are independent and mean zero conditional on  $W_t$ and $Y_t$. Therefore, 
\[
\begin{split}
 \EE{\Big( \sum_{j \in \ngb_i} \abs{X_{j(t+1)} - Y_{j(t+1)}} \Big)^3}
& = \EE{\Big( \sum_{j \in \ngb_i} \p{ a_{jt} +\epsilon_{j(t+1)}}  \Big)^3}\\
& = \EE{\Big( \sum_{j \in \ngb_i} a_{jt}\Big)^3 + 3 \sum_{j,k \in \ngb_i} a_{jt} \epsilon_{k(t+1)}^2 }. 
\end{split}
\]
Since $\EE{\epsilon_{k(t+1)}^2 \mid Y_t, W_t} = a_{kt} (1-a_{kt}) \leq a_{kt}$, the above expression can further be bounded by 
$ \EE{\p{ \sum_{j \in \ngb_i} a_{jt}}^3} + 3 \EE{\p{\sum_{j \in \ngb_i} a_{jt}}^2 }$. We start with bounding the first term. To this end, note that 
\[a_{jt} = \abs{f_j(X_{jt}, W_{jt}, V_{jt}) - f_j(Y_{jt}, W_{jt}, Z_{jt})}
\leq B \abs{X_{jt} - Y_{jt}} + L_n \abs{V_{jt} - Z_{jt}}. 
\]
Thus 
\[\sum_{j \in \ngb_i} a_{jt} \leq B \sum_{j \in \ngb_i} \abs{X_{jt} - Y_{jt}} + L_n \sum_{j \in \ngb_i} \abs{V_{jt} - Z_{jt}}
= B \sum_{j \in \ngb_i} \abs{X_{jt} - Y_{jt}} + L_n \sum_{j \in \ngb_i} \sum_{k \in \ngb_j} \abs{X_{kt} - Y_{kt}}
 .  \]
Thus
\[
\begin{split}
& \EE{\Big(\sum_{j \in \ngb_i} a_{jt}\Big)^3}
 \leq \EE{\Big(B \sum_{j \in \ngb_i} \abs{X_{jt} - Y_{jt}} + L_n \sum_{j \in \ngb_i} \sum_{k \in \ngb_j} \abs{X_{kt} - Y_{kt}}\Big)^3}\\
& \qquad \leq \p{B \mathbb{E} \bigg[\Big(\sum_{j \in \ngb_i} \abs{X_{jt} - Y_{jt}}\Big)^3\bigg]^{\frac{1}{3}} + L_n \sum_{j\in\ngb_i} \mathbb{E} \bigg[\Big(\sum_{k \in \ngb_j} \abs{X_{kt} - Y_{kt}}\Big)^3\bigg]^{\frac{1}{3}}}^3\\
& \qquad \leq \p{B d_{E,3}(X_t, Y_t) + L_n \rhon d_{E,3}(X_t, Y_t)}^3
\leq C^3 d_{E,3}(X_t, Y_t)^3. 
\end{split}
\]
At the same time, the second term can be easily bounded using bounds of the first term. Specifically, by H\"older's inequality, $\EE{\p{\sum_{j \in \ngb_i} a_{jt}}^2} \leq \EE{\p{\sum_{j \in \ngb_i} a_{jt}}^3}^{\frac{2}{3}} \leq C^2 d_{E,3}(X_t, Y_t)^2$. Therefore, combining the results, we have
\[
\begin{split}
d_{E,3}(X_{t+1}, Y_{t + 1})^3 &= \max_i \EE{\Big( \sum_{j \in \ngb_i} \abs{X_{j(t+1)} - Y_{j(t+1)}} \Big)^3}\\
& \leq \max_i\p{\mathbb{E}\bigg[\Big( \sum_{j \in \ngb_i} a_{jt}\Big)^3\bigg] + 3 \mathbb{E}\bigg[\Big( \sum_{j \in \ngb_i} a_{jt}\Big)^2\bigg] }\\
& \leq C^3 d_{E,3}(X_t, Y_t)^3 + 3C^2 d_{E,3}(X_t, Y_t)^2
\leq (C d_{E,3}(X_t, Y_t) + 1)^3. 
\end{split}
\]

\end{proof}

\begin{lemm} 
\label{lemma:close_dyn_MDP}
Let $\Ps$ be the fixed point of system \eqref{eqn:dynamic_system}. Under Assumptions \ref{assu:MDP} and \ref{assu:Bern}, assume further that the initial distribution is defined as $Y_{0 i} \sim \operatorname{Ber}(\Ps_i)$ independently. 
Under Assumptions \ref{assu:bound_fi} and \ref{assu:largest_degree}, we have 
\begin{equation}
W_{L_1} \p{\law(Y_0), \law(Y_1)} \leq L_n \sqrt{D_n}/2, 
\end{equation}
\begin{equation}
W_{d_E} \p{\law(Y_0), \law(Y_1)} \leq L_n D_n^{\frac{3}{2}}/2, 
\end{equation}
and 
\begin{equation}
W_{d_{E,3}} \p{\law(Y_0), \law(Y_1)} \leq 2 L_n D_n^{\frac{3}{2}}. 
\end{equation}
\end{lemm}

\begin{proof}
We start by noting that if we let $
\widetilde{Y}_{i1} \sim \operatorname{Ber}\p{f\p{Y_{i0}, W_{i0}, \Qs_{i}}}$  independently for each $i$,  where $\Qs_{i} = \sum_{j \in \mathcal{N}_i} \Ps_j$, then $\law(\widetilde{Y}_{1}) = \law\p{Y_{0}}$. This is because 
\[
\begin{split}
\EE{\widetilde{Y}_{i1}} &= \EE{f\p{Y_{i0}, W_{i0}, \Qs_{i}}}
= a_i\p{\Qs_{i}} + b_i\p{\Qs_{i}} \EE{W_{0i}} +   c_i\p{\Qs_{i}} \EE{Y_{i 0}} + d_i\p{\Qs_{i}} \EE{W_{0i} Y_{i 0}}\\
& = a_i\p{\Qs_{i}} + b_i\p{\Qs_{i}} \pi_i+   c_i\p{\Qs_{i}} \Ps_i + d_i\p{\Qs_{i}} \pi_i \Ps_i
= f\p{\Ps_{i}, \pi_{i}, \Qs_{i}} = \Ps_i, 
\end{split}
\]
and each $\widetilde{Y}_{i1}$ are indeed independent. 

Then we will focus on $\widetilde{Y}_{1}$ and show that it is close to $Y_1$. Specifically, assume that $\widetilde{Y}_{1}$ and $Y_1$ are ``generated" with the same treatment vector $W_0$ and random seed $U_0$. More precisely, assume that conditional on $Y_0$, $U_{i0}$'s are generated from Unif$[0,1]$ independently, and  $\widetilde{Y}_{i1} = \mathbbm{1}\cb{U_0 \leq f\p{Y_{i0}, W_{i0}, \Qs_{i}}}$ and $Y_{i1} = \mathbbm{1}\cb{U_0 \leq f\p{Y_{i0}, W_{i0}, Z_{i0}}}$. Thus, 
\[
\begin{split}
\EE{ \abs{\widetilde{Y}_{i1} - Y_{i1}}}
= \EE{\abs{f\p{Y_{i0}, W_{i0}, \Qs_{i}} - f\p{Y_{i0}, W_{i0}, Z_{i0}}}}
\leq L_n \EE{\abs{\Qs_{i} -  Z_{i0}}}. 
\end{split}
\]
Recall that $Z_{i0} = \sum_{j \in \mathcal{N}_i} Y_{i0}$, $\Qs_i = \sum_{j \in \mathcal{N}_i} \Ps_i$ and $Y_{0 i} \sim \operatorname{Ber}(\Ps_i)$. Therefore, 
$\EE{\p{\Qs_{i} -  Z_{i0}}^2} = \sum_{j \in \mathcal{N}_i} \Ps_i(1 - \Ps_i) \leq D_n/4$. This further implies that $\EE{\abs{\Qs_{i} -  Z_{i0}}} \leq \sqrt{D_n}/2$. 

Thus
\begin{align*}
 W_{L_1}\p{\law(Y_0), \law(Y_1)} 
= W_{L_1}\p{\law(\widetilde{Y}_1), \law(Y_1)} \leq \sum_{i=1}^n \EE{ \abs{\widetilde{Y}_{i1} - Y_{i1}}}  
\leq n L_n \sqrt{D_n}/2. 
\end{align*}

For $W_{d_E}$, we can bound things similarly. 
\begin{align*}
 W_{d_E}\p{\law(Y_0), \law(Y_1)} 
= W_{d_E}\p{\law(\widetilde{Y}_1), \law(Y_1)} 
\leq \max_i \sum_{j \in \mathcal{N}_i} \EE{ \abs{\widetilde{Y}_{j1} - Y_{j1}}}  
\leq L_n D_n^{\frac{3}{2}}/2. 
\end{align*}

For $W_{d_{E,3}}$, things are slightly more complicated. We make use of the same construction of $\widetilde{Y}$ and $Y$ as above. Note that $d_{E,3}(\widetilde{Y}_1, Y_1)^3 = \max_i \EE{\p{\sum_{j \in \ngb_i} \abs{\widetilde{Y}_{j1} - Y_{j1}}}^3}$. With the above construction, $\abs{\widetilde{Y}_{j1} - Y_{j1}} \sim \operatorname{Ber}(b_{j0})$, where 
\[b_{it} = \abs{f\p{Y_{i0}, W_{i0}, \Qs_{i}} - f\p{Y_{i0}, W_{i0}, Z_{i0}}}. 
\]
Let $\eta_{j1} = \abs{\widetilde{Y}_{j1} - Y_{j1}} - b_{j0}$. We note that $b_{j0}$ is measurable with respect to $W_0$ and $Y_0$, whereas $\eta_{j1}$'s are independent and mean zero conditional on  $W_0$ and $Y_0$. Therefore, 
\[
\begin{split}
 \EE{\Big( \sum_{j \in \ngb_i} \abs{\widetilde{Y}_{j1} - Y_{j1}} \Big)^3}
& = \EE{\Big( \sum_{j \in \ngb_i} \p{ b_{j0} +\eta_{j0}}  \Big)^3}\\
& = \EE{\Big( \sum_{j \in \ngb_i} b_{j0}\Big)^3 + 3 \sum_{j,k \in \ngb_i} b_{j0} \eta_{k1}^2 }. 
\end{split}
\]
Since $\EE{\eta_{k1}^2 \mid Y_0, W_0} = b_{k0} (1-b_{k0}) \leq b_{k0}$, the above expression can further be bounded by 
$ \EE{\p{ \sum_{j \in \ngb_i} b_{j0}}^3} + 3 \EE{\p{\sum_{j \in \ngb_i} b_{j0}}^2 }$. To bound $b_{j0}$, we note that
\[b_{j0} = \abs{f\p{Y_{i0}, W_{i0}, \Qs_{i}} - f\p{Y_{i0}, W_{i0}, Z_{i0}}} \leq L_n \abs{Z_{i0} - \Qs_{i} }. \]
Again, by recalling that $Z_{i0} = \sum_{j \in \mathcal{N}_i} Y_{i0}$, $\Qs_i = \sum_{j \in \mathcal{N}_i} \Ps_i$ and $Y_{0 i} \sim \operatorname{Ber}(\Ps_i)$, we have that 
\[\EE{b_{j0}^3} \leq L_n^3 \EE{\abs{Z_{i0} - \Qs_{i} }^3} \leq L_n^3 D_n^{\frac{3}{2}}. \]
Thus,  
\[ \EE{\Big(\sum_{j \in \ngb_i} b_{j0}\Big)^3} \leq D_n^2  \EE{\sum_{j \in \ngb_i} b_{j0}^3} \leq L_n^3 D_n^{\frac{9}{2}}. \]
This further implies that
\[ 3 \EE{\Big(\sum_{j \in \ngb_i} b_{j0}\Big)^2} \leq 3 \EE{\Big(\sum_{j \in \ngb_i} b_{j0}\Big)^3}^{\frac{2}{3}} \leq 3 L_n^2 D_n^3 \leq 3 L_n^3 D_n^{\frac{9}{2}}.\]
Therefore, 
\[
\begin{split}
d_{E,3}(\tilde{Y}_{1}, Y_{1})^3 &= \max_i \EE{\Big( \sum_{j \in \ngb_i} \abs{\tilde{Y}_{j1} - Y_{j1}} \Big)^3}\\
& \leq \max_i\p{\mathbb{E}\bigg[\Big( \sum_{j \in \ngb_i} b_{j0}\Big)^3\bigg] + 3 \mathbb{E}\bigg[\Big( \sum_{j \in \ngb_i} b_{j0}\Big)^2\bigg] } 
\leq 4 L_n^3 D_n^{\frac{9}{2}}. 
\end{split}
\]
Therefore, 
\begin{align*}
 W_{d_{E,3}}\p{\law(Y_0), \law(Y_1)} 
= W_{d_{E,3}}\p{\law(\widetilde{Y}_1), \law(Y_1)} 
\leq d_{E,3}\p{\widetilde{Y}_1, Y_1} 
\leq 2 L_n D_n^{\frac{3}{2}}. 
\end{align*}

\end{proof}

\begin{lemm}
\label{lemma:exp_ergodic}
Consider random variables $A_t$, $t \in \cb{1, 2, \dots, T}$ and a constant $\alpha$. Assume that there exists a constant $\theta < 1$ such that
\begin{equation}
\label{eqn:lemma_exp_ergodic}
\abs{\EE{(A_t - \alpha)(A_s - \alpha)}} \leq C_0 \theta^{\abs{s-t}} \textnormal{ for any } s, t \in \cb{1,2, \dots, T}. 
\end{equation}
Then 
\begin{equation}
\EE{\p{\frac{1}{T}\sum_{t = 1}^T A_t - \alpha}^2} \leq \frac{(1+\theta)C_0}{(1-\theta)T} . 
\end{equation}
\end{lemm}

\begin{proof}
Note that
\[
\begin{split}
\EE{\p{\frac{1}{T}\sum_{t = 1}^T A_t - \alpha}^2}
= \frac{1}{T^2}\sum_{s,t} \EE{(A_t - \alpha)(A_s - \alpha)} 
\end{split}
\]
This can be further bounded by
\[\frac{1}{T^2}\sum_{s,t} \EE{(A_t - \alpha)(A_s - \alpha)}  \leq \frac{1}{T^2} \sum_{s,t} C_0 \theta^{\abs{s-t}} = \frac{C_0}{T} + \frac{2 C_0}{T^2} \sum_{t = 1}^T (T - t) \theta^{t} \leq \frac{(1+\theta)C_0}{(1-\theta)T}. \]
\end{proof}

\begin{lemm}
\label{lemma:markov_exp_ergodic}
Let $B_t$ be a discrete time Markov chain with stationary distribution $\mu$. Define a random process $A_t = h(B_t, B_{t+1}) \in \mathbb{R}$. Let $\alpha = \EE[\mu]{A_t}$. Suppose that there exist constants $C_1 >0$, $C_2 >0$, $\theta \in (0,1)$ such that
\begin{equation}
\abs{A_t - \alpha} \leq C_1 \textnormal{ for } t \geq 0 \textnormal{, and } \abs{\EE{A_t - \alpha}} \leq C_2 \theta^t \textnormal{ for } t \geq 1,
\end{equation}
then
\begin{equation}
\EE{\p{\frac{1}{T}\sum_{t = 1}^T A_t - \alpha}^2} \leq \frac{(1+\theta)C_1 (C_1 + C_2)}{\theta(1-\theta)T} . 
\end{equation}

In particular, assume that there exist two processes $B^1_t$ and $B^2_t$ (with the same transition probability as $B_t$), such that $B^2_0 \sim \mu$. Let $A^1_t =  h(B^1_t, B^1_{t+1})$ and $A^2_t =  h(B^2_t, B^2_{t+1})$. If 
\begin{equation}
\abs{A^1_t - A^2_t} \leq C_1 \textnormal{ for } t \geq 0 \textnormal{, and } \abs{\EE{A^1_t - A^2_t}} \leq C_2 \theta^t \textnormal{ for } t \geq 1,
\end{equation}
then 
\begin{equation}
\EE{\p{\frac{1}{T}\sum_{t = 1}^T A^1_t - \alpha}^2} \leq \frac{(1+\theta)C_1 (C_1 + C_2)}{\theta(1-\theta)T} . 
\end{equation}
\end{lemm}

\begin{proof}
We will make use of Lemma \ref{lemma:exp_ergodic}. To show condition \eqref{eqn:lemma_exp_ergodic}, we note that due to the Markovian nature of the processes, it suffice to focus on $\EE{(A_t - \alpha)(A_0 - \alpha)}$. A closer look at $\EE{(A_t - \alpha)(A_0 - \alpha)}$ shows that for $t \geq 2$, 
\[
\begin{split}
\abs{\EE{(A_t - \alpha)(A_0 - \alpha)}} &= \abs{\EE{\EE{(A_t - \alpha) \mid B_0, B_1} (A_0 - \alpha)}}\\
&\leq \EE{\abs{\EE{(A_t - \alpha) \mid B_0, B_1}} \abs{A_0 - \alpha}}
 \leq C_1 C_2 \theta^{t - 1}. 
\end{split}
\]
At the same time, for $t = 0, 1$, $\abs{\EE{(A_t - \alpha)(A_0 - \alpha)}} \leq C_1^2$. Combining the two results, we get 
\[\abs{\EE{(A_t - \alpha)(A_0 - \alpha)}} \leq C_1\p{\frac{C_2}{\theta} + C_1} \theta^t \leq \frac{C_1(C_1 + C_2)}{\theta} \theta^{t}.  \]
Applying Lemma \ref{lemma:exp_ergodic} gives the desired result. 

The second part of the Lemma is a direct corollary of the first part. 

\end{proof}

\begin{lemm}
\label{lemma:little_bound_ratio}
Let $a_1$, $b_1$, $a_2$, $b_2$ be non-negative random variables. If $b_2 \geq c_0 > 0$, $a_1 \leq c_1 b_1$, then
\begin{equation}
\begin{split}
\EE{\abs{\frac{a_1}{b_1} - \frac{a_2}{b_2}}} 
&\leq \frac{c_1\EE{\abs{b_1 - b_2}} + \EE{\abs{a_1 - a_2}}}{c_0}\\
&\leq \frac{c_1 \sqrt{\EE{(b_1 - b_2)^2}}+ \sqrt{\EE{(a_1 - a_2)^2}}}{c_0},
\end{split}
\end{equation}
and
\begin{equation}
\begin{split}
\EE{\p{\frac{a_1}{b_1} - \frac{a_2}{b_2}}^2} 
&\leq \frac{2c_1^2 \EE{(b_1 - b_2)^2}+ 2\EE{(a_1 - a_2)^2}}{c_0^2},
\end{split}
\end{equation}
\end{lemm}
\begin{proof}
Note that
\[
\frac{a_1}{b_1} - \frac{a_2}{b_2} = \frac{a_1 b_2 - a_2 b_1}{b_1 b_2} = \frac{a_1(b_2 - b_1)}{b_1 b_2} + \frac{(a_1 - a_2) b_1}{b_1 b_2}
= \frac{a_1}{b_1}\frac{1}{b_2} (\beta_1 - b_2) + \frac{a_1 - a_2}{b_2}. 
\]
Then the conclusion follows directly from Cauchy Schwarz inequality. 
\end{proof}

\begin{lemm}
\label{lemma:l_infty_norm}
Let $M$ be a $n \times n$ matrix. Assume that there exist constants $a > b >0$ such that for any  $i \in \cb{1, \dots, n}$, $M_{i,i} \geq a$, and $\sum_{k \neq i}\abs{M_{i,k}} \geq - b$. 
\begin{enumerate}
\item If $M$ is invertible, then for any vector $\bu$,
\[ \Norm{M^{-1} \bu}_{\infty} \leq \Norm{\bu}_{\infty}/(a-b).  \]
\item Furthermore, if we assume that each row $M_{i \cdot}$ of the matrix $M$ has all its elements non-negative or non-positive at the same time, then the smallest singular value of $M$ (in absolute value) is bounded below by $a - b > 0$, and thus $M$ is invertible. 
\end{enumerate}
\end{lemm}
\begin{proof}
To show the bound on norm of $M^{-1} \bu$, 
we use proof by contradiction. Let $\bv$ be a vector such that $M \bv = \bu$. Assume that $\Norm{\bv}_{\infty} > \Norm{\bu}_{\infty}/(a-b)$. Without loss of generality, assume that $|v_1|= \max_i |v_i| >  \Norm{\bu}_{\infty}/(a-b)$. Then 
\[
\abs{u_1} = \abs{M_{1,1} v_1 + \sum_{i = 2}^n M_{1,i}v_i}
\geq M_{1,1} \abs{v_1} - \sum_{i = 2}^n \abs{M_{1,i}} \abs{v_i}
\geq M_{1,1} \abs{v_1} - \abs{v_1}\sum_{i = 2}^n \abs{M_{1,i}} 
\geq (a-b)  \abs{v_1} >  \Norm{\bu}_{\infty}. 
\]
This is impossible, which implies our assumption that $\Norm{\bv}_{\infty} > a-b$ is wrong. Thus $\Norm{\bv}_{\infty} \leq \Norm{\bu}_{\infty}/(a-b)$. 

For 2, note that we can write $M = M_1 + M_2$, where $M_1$ is a diagonal matrix and $M_2$ has its diagonal being 0. Then the smallest singular value of $M_1$ is bounded below by $a$, while $\Norm{M_2}_{\operatorname{op}} \leq b$ by the Perron–Frobenius theorem. Thus the smallest singular value of $M$ is bounded below by $a - b > 0$.  
\end{proof}

\begin{lemm}
\label{lemma:l_infty_norm2}
Let $M$ be a $n \times n$ matrix. Let $a_1, b_1, a_2, b_2 \dots a_n, b_n$ be sequences of real numbers such that $a_i > b_i > 0$. Assume that there exists  a constant $\eta > 0$ such that $b_i \leq (1 - \eta) a_i$, and there exists another constant $\kappa > 0$ such that $a_i \geq \kappa$. Assume that $M$ satisfies the following: for any  $i \in \cb{1, \dots, n}$, $M_{i,i} \geq a_i$, the row $M_{i \cdot}$ has all its elements non-negative or non-positive at the same time, and $\sum_{k \neq i}\abs{M_{i,k}} \geq - b_i$. Then the smallest singular value of $M$ (in absolute value) is bounded below by $\eta \kappa$, and thus $M$ is invertible. Furthermore, for any vector $\bu$,
\[ \Norm{M^{-1} \bu}_{\infty} \leq \Norm{\bu}_{\infty}/(\eta \kappa).  \]
\end{lemm}

\begin{proof}
Let $A$ be a diagonal matrix with diagonal elements $a_1, \dots, a_n$. Then the matrix $M$ can be written as $M = A \widetilde{M}$. Then $\widetilde{M}$ satisfies the condition of Lemma \ref{lemma:l_infty_norm} with $a = 1$ and $b = 1 - \eta$. Thus the smallest singular value of $\widetilde{M}$ is bounded below by $\eta$, $\widetilde{M}$ is invertible, and for any vector $\bu$, $\Norm{\widetilde{M}^{-1} \bu}_{\infty} \leq \Norm{\bu}_{\infty}/\eta$. This then implies that $M$ is invertible. Specifically, the largest singular value of $M^{-1}$ is bounded above by that of $\widetilde{M}^{-1}\kappa$. Thus the smallest singular value of $M$ (in absolute value) is bounded below by $\eta \kappa$. In terms of the infinity norm, 
\[ \Norm{M^{-1} \bu}_{\infty} \leq \Norm{\widetilde{M}^{-1} \bu}_{\infty} /\kappa \leq \Norm{\bu}_{\infty}/(\eta \kappa ).  \]
\end{proof}

\begin{lemm}
\label{lemma:matrix_inverse}
Let $\alpha, \beta_1, \beta_2 \in \mathbb{R}^n$ be three vectors. Let $X_1$, $X_2 \in \mathbb{R}^{n \times n}$ be invertible matrices. Then
\[ \abs{\alpha\trans X_1^{-1} \beta_1 - \alpha\trans X_2^{-1} \beta_2 } 
\leq \frac{\Norm{\alpha} \Norm{\beta_1 - \beta_2}}{\lambda_{\operatorname{smallest}}(X_1)}
+ \Norm{(X_2\trans)^{-1}\alpha}_{\infty}\Norm{(X_1\trans)^{-1}\beta_2}_{\infty} \sum_{i,j} \abs{X_{1,i,j} - X_{2,i,j}},
 \]
where $\lambda_{\operatorname{smallest}}(\cdot)$ is the smallest singular value (in absolute value), and $\mathbf{1} \in \mathbb{R}^{n}$ is the vector with all elements one. 
\end{lemm}
\begin{proof}
The result follows easily from the decomposition of $\abs{a\trans X_1^{-1} \beta_1 - \alpha\trans X_2^{-1} \beta_2}$:
\[
\begin{split}
\abs{\alpha\trans X_1^{-1} \beta_1 - \alpha\trans X_2^{-1} \beta_2}
&= \alpha \trans X_1^{-1} (\beta_1 - \beta_2) + \alpha \trans (X_1^{-1} - X_2^{-1}) \beta_2\\
&=  \alpha \trans X_1^{-1} (\beta_1 - \beta_2) + \alpha \trans (X_2\trans)^{-1} (X_2 - X_1) X_1^{-1} \beta_2.
\end{split}
\]
\end{proof}

\begin{lemm}
Under Assumption \ref{assu:MDP} - \ref{assu:contraction}, 
\[
\EE{\p{\frac{1}{T} \sum_{t = 1}^T Y_{it} - \EE[\mu\p{\pi}]{ Y_{it}}}^2} \leq \frac{2(1+C)}{C(1-C)T}.  
\]
\label{lemma:p_hat_converge}
\end{lemm}
\begin{proof}
We will make use of Lemma \ref{lemma:markov_exp_ergodic}. 
Let $X_t$ and $Y_t$ be two processes satisfying Assumption \ref{assu:MDP} and \ref{assu:Bern}. Assume that the initial distribution of $X$ is the stationary distribution $\mu(\pi)$. Lemma \ref{lemma:contraction_d_e} implies that $W_{d_E}(\law(X_t), \law(Y_t)) \leq C^t W_{d_E}(\law(X_0), \law(Y_0))$ for any $t \geq 0$, where $C<1$ is the constant in Assumption \ref{assu:contraction}. In particular, there exists a coupling of the processes $X$ and $Y$ such that $d_E(X_t, Y_t) \leq C^t W_{d_E}(\law(X_0), \law(Y_0))$ for any $t$. Note that we can further bound $C^t W_{d_E}(\law(X_0), \law(Y_0))$ by $C^t D_n$. 
With such coupling, we can bound the difference between $\EE{ Y_{it}}$ and $\EE{X_{it}}$. In particular, 
\[
\begin{split}
  \abs{\EE{Y_{it}} - \EE[\pi]{X_{it}}}
=\abs{\EE{X_{it} - Y_{it}}}
 \leq  L_n d_E(X_{t-1}, Y_{t-1}) \leq C^{t-1} L_n D_n \leq C^t. 
\end{split}
\]
Thus Lemma \ref{lemma:markov_exp_ergodic} implies that 
\[
\EE{\p{\frac{1}{T} \sum_{t = 1}^T Y_{it} - \EE[\mu\p{\pi}]{ Y_{it}}}^2} \leq \frac{2(1+C)}{C(1-C)T}.  
\]
\end{proof}

\begin{lemm}
\label{lemma:distance_between_mu_ps}
Under Assumptions \ref{assu:MDP} - \ref{assu:contraction}, for any $y, w \in \cb{0,1}$, 
\begin{enumerate}
\item 
\[ \abs{\EE[\mu(\pi)]{f_i(y,w,Z_i)} -  f_i(y,w, \Qs(\pi))} \leq \sqrt{L_n} C/(2(1-C)).\]
\item For any $g \in \cb{a_1, \dots, d_i}$\[ \abs{\EE[\mu(\pi)]{g(Z_i)} -  g(\Qs(\pi))} \leq 2\sqrt{L_n} C/(1-C).\]
\item \[ \abs{\EE[\mu(\pi)]{Y_i} - \Ps(\pi) } \leq C_1 \sqrt{L_n},\]
 for some constant $C_1$. 
\end{enumerate}
\end{lemm}
\begin{proof}
For 1, we have that
\[ 
\begin{split}
& \abs{\EE[\mu(\pi)]{f_i(y,w,Z_i)} -  f_i(y,w, \Qs(\pi))}
 \leq L_n \EE[\mu(\pi)]{\abs{Z_i - \Qs(\pi)}}\\
&\qquad \qquad \leq L_n \sqrt{D_n} \sqrt{C}/(2(1-C)) 
\leq \sqrt{L_n} C/(2(1-C)),
\end{split}
\]
where the second inequality follows from Theorem \ref{theo:mean_field2}. Then 2 is then a direct corollary of 1. 

For 3, we note that 
\[
\begin{split}
\EE[\mu(\pi)]{Y_i}
&= \EE[\mu(\pi)]{a_i(Z_i)} + \pi_i \EE[\mu(\pi)]{b_i(Z_i)}  + \EE[\mu(\pi)]{c_i(Z_i) Y_i} + \pi \EE[\mu(\pi)]{d_i(Z_i) Y_i} \\
& = \EE[\mu(\pi)]{a_i(Z_i)} + \pi_i \EE[\mu(\pi)]{b_i(Z_i)}  + \EE[\mu(\pi)]{c_i(Z_i)}\EE[\mu(\pi)]{Y_i} + \pi_i \EE[\mu(\pi)]{d_i(Z_i)}\EE[\mu(\pi)]{Y_i} + \operatorname{error},
\end{split}
\]
where $\abs{\operatorname{error}} \leq \Cov[\mu(\pi)]{c_i(Z_i) + \pi_i d_i(Z_i), Y_i} \leq \EE[\mu]{\abs{ c_i(Z_i) +d_i(Z_i)  \pi_i  - \EE[\mu]{c_i(Z_i)  + d_i(Z_i)  \pi_i}  }} \leq 3L_n^2 D_n^{\frac{3}{2}}/(1-C) \leq 3C^{\frac{3}{2}}\sqrt{L_n}/(1-C)$ by \eqref{eqn:corr_c_i_Y}. Therefore,
\[
\begin{split}
\EE[\mu(\pi)]{Y_i} = \frac{\EE[\mu(\pi)]{a_i(Z_i)} + \pi_i \EE[\mu(\pi)]{b_i(Z_i)} + \operatorname{error}}{1 - \EE[\mu(\pi)]{c_i(Z_i)} - \pi_i \EE[\mu(\pi)]{d_i(Z_i)}}
= \frac{a_i(\Qs(\pi)) + \pi_i b_i(\Qs(\pi)) }{1 - c_i(\Qs(\pi)) - \pi_i d_i(\Qs(\pi)) } + \operatorname{error}_1,
\end{split}
\]
where $\abs{\operatorname{error}_1} \leq C_1\sqrt{L_n}$ for some constant $C_1$. The inequality is a result of the second part of the lemma and the fact that $c_i(Z_i) + \pi_i d_i(Z_i) \leq C < 1$. Finally, we note that since 
\[\Ps_{i}(\pi) = a_i\p{\Qs_{i}(\pi)} + b_i\p{\Qs_{i}(\pi)} \pi_i +  c_i\p{\Qs_{i}(\pi)} \Ps_{i}(\pi) + d_i\p{\Qs_{i}(\pi)} \pi_i \Ps_{i}(\pi), \]
\[\Ps_{i}(\pi) = \frac{a_i(\Qs(\pi)) + \pi_i b_i(\Qs(\pi)) }{1 - c_i(\Qs(\pi)) - \pi_i d_i(\Qs(\pi)) }.\]
The desired result then follows.

\end{proof}

\subsection{Proof of Proposition \ref{prop:MDP_stationary}}
The existence of a stationary distribution follows directly from the fact that the induced Markov chain is finite-state and time homogenous. 

We show the uniqueness of the stationary distribution using proof by contradiction. Assume that there are two different stationary distributions $\mu_1$ and $\mu_2$. Then the $L_1$-Wasserstein distance defined in \eqref{eqn:wasser_L1} between the two measures are not zero: $W_{L_1}(\mu_1, \mu_2) > 0$. Let $X_t$ and $Y_t$ be two processes satisfying Assumptions  \ref{assu:MDP} and \ref{assu:Bern}. Furthermore, assume that $X_t \sim \mu_1$ and $Y_t \sim \mu_2$. Then by Lemma \ref{lemma:L1_contraction},
\[ W_{L_1}(\mu_1, \mu_2) = W_{L_1}(\law(X_{t+1}), \law(Y_{t+1}))
\leq C W_{L_1}(\law(X_{t}), \law(Y_{t})) = W_{L_1}(\mu_1, \mu_2),
 \]
which is impossible. This implies that our assumption that there are two different stationary distributions is wrong. Thus the stationary distribution is unique.

To show convergence in distribution, we will again make use of Lemma \ref{lemma:L1_contraction}. Let $\mu$ be the unique stationary distribution. Let $X_t$ and $Y_t$ be two processes satisfying Assumptions  \ref{assu:MDP} and \ref{assu:Bern}. Furthermore, assume that $X_0 \sim \mu$. Then by Lemma \ref{lemma:L1_contraction},
\[
\begin{split}
 W_{L_1}(\mu, \law(Y_{t})) 
&= W_{L_1}(\law(X_{t}), \law(Y_{t}))
\leq C W_{L_1}(\law(X_{t-1}), \law(Y_{t-1}))
\leq \dots\\
&\leq C^t W_{L_1}(\law(X_0), \law(Y_0))
\leq C^t W_{L_1}(\mu, \law(Y_0)).
\end{split}
 \]
Therefore, $ W_{L_1}(\mu, \law(Y_{t})) \to 0$ as $t \to \infty$. Hence $Y_t \Rightarrow \mu$.

\subsection{Proof of Proposition \ref{prop:system_fixed}}
The existence of the fixed point follows directly from Brouwer's fixed-point theorem. 

For uniqueness, like in the proof of Proposition \ref{prop:MDP_stationary}, we use proof by contradiction. Assume that there are two different fixed point $\Ps_1$ and $\Ps_2$. Then for any $i$,
\begin{equation}
\label{eqn:dyn_contraction1}
\begin{split}
\abs{\Ps_{1,i} - \Ps_{2,i}} &= 
\abs{f_i(\Ps_{1,i}, W_i, \Qs_{1,i})
- f_i(\Ps_{2,i}, W_i, \Qs_{2,i})}
\leq B \abs{\Ps_{1,i} - \Ps_{2,i}} + L_n \abs{\Qs_{1,i} - \Qs_{2,i}}\\
& \leq B \abs{\Ps_{1,i} - \Ps_{2,i}} + L_n \sum_{j \in \ngb_i} \abs{\Ps_{1,j} - \Ps_{2,j}}. 
\end{split}
\end{equation}
This then implies that
\begin{equation}
\label{eqn:dyn_contraction2}
\begin{split}
\sum_{i =1}^n \abs{\Ps_{1,i} - \Ps_{2,i}}
&\leq \sum_{i =1}^n\abs{\Ps_{1,i} - \Ps_{2,i}} + L_n \sum_{i =1}^n \sum_{j \in \ngb_i} \abs{\Ps_{1,j} - \Ps_{2,j}}\\
&= \sum_{i =1}^n\abs{\Ps_{1,i} - \Ps_{2,i}} + L_n \sum_{j =1}^n \sum_{i \in \ngb_j} \abs{\Ps_{1,j} - \Ps_{2,j}}\\
&\leq \sum_{i =1}^n\abs{\Ps_{1,i} - \Ps_{2,i}} + L_n D_n \sum_{j =1}^n \abs{\Ps_{1,j} - \Ps_{2,j}}
\leq C \sum_{i =1}^n \abs{\Ps_{1,i} - \Ps_{2,i}}. 
\end{split}
\end{equation}
But this cannot be true for two different fixed points $\Ps_1$ and $\Ps_2$. Thus, the fixed point must be unique.

To show the convergence, we note that for any process $P_t$ satisfying \eqref{eqn:dynamic_system}, we have 
\[
\begin{split}
\sum_{i = 1}^n \abs{P_{t,i} - \Ps_i}
\leq C \sum_{i = 1}^n \abs{P_{t-1,i} - \Ps_i}
\leq \dots
\leq C^t \sum_{i = 1}^n \abs{P_{0,i} - \Ps_i}. 
\end{split}
 \]
Here the inequality follows from the same arguments as in \eqref{eqn:dyn_contraction1} and \eqref{eqn:dyn_contraction2}. 
Thus $\sum_{i = 1}^n \abs{P_{t,i} - \Ps_i} \to 0$ as $t \to \infty$. Hence, $P_t \to \Ps$ as $t \to \infty$. 

\subsection{The distances defined in \eqref{eqn:wasser_L1}, \eqref{eqn:dek_def} and \eqref{eqn:wasser_dek}}
\label{subsection:triangular}
In this section, we show that the distances defined in \eqref{eqn:wasser_L1}, \eqref{eqn:dek_def} and \eqref{eqn:wasser_dek} are well-defined metrics. They clearly satisfy the identity of indiscernibles and symmetry. We will show triangular inequalities.

For $d_{E,k}$, let $X,Y,Z$ be three random vectors in $\cb{0,1}^n$. Then
\[
\begin{split}
d_{E,k}(X,Y) + d_{E,k}(Y,Z)
& = \max_i \bigg(\mathbb{E}\bigg[ \Big(\sum_{j \in \mathcal{N}_i} \abs{X_j - Y_j}\Big)^k\bigg]\bigg)^{\frac{1}{k}}
+ \max_i  \bigg(\mathbb{E}\bigg[ \Big(\sum_{j \in \mathcal{N}_i} \abs{Y_j - Z_j}\Big)^k\bigg]\bigg)^{\frac{1}{k}}\\
& \geq \max_i \cb{ \bigg(\mathbb{E}\bigg[ \Big(\sum_{j \in \mathcal{N}_i} \abs{X_j - Y_j}\Big)^k\bigg]\bigg)^{\frac{1}{k}}
+ \bigg(\mathbb{E}\bigg[ \Big(\sum_{j \in \mathcal{N}_i} \abs{Y_j - Z_j}\Big)^k\bigg]\bigg)^{\frac{1}{k}} }\\
& \geq \max_i \cb{ \bigg(\mathbb{E}\bigg[ \Big(\sum_{j \in \mathcal{N}_i} \abs{X_j - Y_j} + \sum_{j \in \mathcal{N}_i} \abs{Y_j - Z_j}\Big)^k\bigg]\bigg)^{\frac{1}{k}} }\\
& \geq \max_i \cb{ \bigg(\mathbb{E}\bigg[ \Big(\sum_{j \in \mathcal{N}_i} \abs{X_j - Z_j}\Big)^k\bigg]\bigg)^{\frac{1}{k}} }  = d_{E,k}(X,Z),
\end{split}
\]   
where the third line follows from the fact that $\EE{(\cdot)^k}^{\frac{1}{k}}$ is a well-defined norm. Finally, $W_{L_1}$ is the first Wasserstein distance corresponding to the $L_1$ norm, while $W_{d_{E,k}}$ is the first Wasserstein distance corresponding to the $d_{E,k}$ metric. Thus the triangular inequalities for them follow from standard analysis for the Wasserstein metric.

\subsection{Proof of Theorem \ref{theo:mean_field} and Theorem \ref{theo:mean_field2}}
Assume $X_t$ and $Y_t$ are two processes satisfying Assumptions \ref{assu:MDP} and \ref{assu:Bern}. Suppose the initial distributions of the two are different: $Y_{0i} \sim \operatorname{Ber}(\Ps_i)$ independently, while $X_0 \sim \mu(\pi)$. With this new set of notations, it suffices to show that $W_{L_1}\p{\law(X_0), \law(Y_0)} \leq n L_n \sqrt{D_n}/(2(1-C))$ and $W_{d_E}\p{\law(X_0), \law(Y_0)} \leq L_n D_n^{\frac{3}{2}}/(2(1-C))$. 

For $W_{L_1}$, we have
\begin{align*}
W_{L_1}\p{ \law(Y_0), \law(Y_1)}
&\geq W_{L_1}\p{ \law(Y_0), \law(X_0)} - W_{L_1}\p{ \law(Y_1), \law(X_0)}\\
& = W_{L_1}\p{ \law(Y_0), \law(X_0)} - W_{L_1}\p{ \law(Y_1), \law(X_1)}\\
& \geq W_{L_1}\p{ \law(Y_0), \law(X_0)} - C W_{L_1}\p{ \law(Y_0), \law(X_0)}\\
& = (1 - C) W_{L_1}\p{ \law(Y_0), \law(X_0)},
\end{align*}
where the first line follows from the triangular inequality, second line follows from the fact that $X_0 \sim \mu(\pi)$ is the stationary distribution, and the third line follows from Lemma \ref{lemma:L1_contraction}.   
Thus by Lemma \ref{lemma:close_dyn_MDP}, 
\[W_{L_1}\sqb{ \law(X_0), \law(Y_0} \leq \frac{1}{1 - C} W_{L_1}\p{ \law(Y_0), \law(Y_1)} \leq  n L_n \sqrt{D_n}/(2(1-C)). \]

Similarly, for $W_{d_E}$, 
\begin{align*}
W_{d_E}\p{ \law(Y_0), \law(Y_1)}
&\geq W_{d_E}\p{ \law(Y_0), \law(X_0)} - W_{d_E}\p{ \law(Y_1), \law(X_0)}\\
& = W_{d_E}\p{ \law(Y_0), \law(X_0)} - W_{d_E}\p{ \law(Y_1), \law(X_1)}\\
& \geq W_{d_E}\p{ \law(Y_0), \law(X_0)} - C W_{d_E}\p{ \law(Y_0), \law(X_0)}\\
& = (1 - C) W_{d_E}\p{ \law(Y_0), \law(X_0)},
\end{align*}
where the first line follows from the triangular inequality, second line follows from the fact that $X_0 \sim \mu(\pi)$ is the stationary distribution, and the third line follows from Lemma \ref{lemma:contraction_d_e}.    
Thus by Lemma \ref{lemma:close_dyn_MDP}, 
\[W_{d_E}\p{\law(X_0), \law(Y_0)} \leq \frac{1}{1 - C} W_{d_E}\p{ \law(Y_0), \law(Y_1)} \leq L_n D_n^{\frac{3}{2}}/(2(1-C)). \]

For $W_{d_{E,3}}$, 
\begin{align*}
W_{d_{E,3}}\p{ \law(Y_0), \law(Y_1)}
&\geq W_{d_{E,3}}\p{ \law(Y_0), \law(X_0)} - W_{d_{E,3}}\p{ \law(Y_1), \law(X_0)}\\
& = W_{d_{E,3}}\p{ \law(Y_0), \law(X_0)} - W_{d_{E,3}}\p{ \law(Y_1), \law(X_1)}\\
& \geq W_{d_{E,3}}\p{ \law(Y_0), \law(X_0)} - C W_{d_{E,3}}\p{ \law(Y_0), \law(X_0)} - 1\\
& = (1 - C) W_{d_{E,3}}\p{ \law(Y_0), \law(X_0)} - 1,
\end{align*}
where the first line follows from the triangular inequality, second line follows from the fact that $X_0 \sim \mu(\pi)$ is the stationary distribution, and the third line follows from Lemma \ref{lemma:contraction_d_e3}.    
Thus by Lemma \ref{lemma:close_dyn_MDP}, 
\[W_{d_{E,3}}\p{\law(X_0), \law(Y_0)} \leq \frac{1}{1 - C} \sqb{W_{d_E}\p{ \law(Y_0), \law(Y_1)} + 1} \leq (2L_n D_n^{\frac{3}{2}} + 1)/(1-C). \]

\subsection{Proof of Theorem \ref{theo:SDE_estimation}}

Simple calculation shows that this estimator $\htau_{\operatorname{IPW}, t}$ is unbiased for $\tau_{\SDE,t}$ conditional on $Y_t$.
\begin{equation}
\begin{split}
\EE{\htau_{\operatorname{IPW}, t} \mid Y_t} &= \frac{1}{n}\sum_{i=1}^n \EE{Y_{i (t+1)} \p{\frac{W_{t}}{\pi_i} - \frac{1 - W_{t}}{1 - \pi_i}} \mid Y_t}\\
& =  \frac{1}{n}\sum_{i=1}^n \EE{Y_{i (t+1)} \frac{W_{t}}{\pi_i} - Y_{i (t+1)} \frac{1 - W_{t}}{1 - \pi_i} \mid Y_t}\\
& = \frac{1}{n}\sum_{i=1}^n \EE{f_{i}(Y_{it}, 1, Z_{it}) \frac{W_{t}}{\pi_i} - f_{i}(Y_{it}, 0, Z_{it})\frac{1 - W_{t}}{1 - \pi_i} \mid Y_t}\\
& = \frac{1}{n}\sum_{i=1}^n f_{i}(Y_{it}, 1, Z_{it}) \EE{\frac{W_{t}}{\pi_i}\mid Y_t} - f_{i}(Y_{it}, 0, Z_{it})\EE{\frac{1 - W_{t}}{1 - \pi_i} \mid Y_t}\\
& = \frac{1}{n}\sum_{i=1}^n f_{i}(Y_{it}, 1, Z_{it}) - f_{i}(Y_{it}, 0, Z_{it}) = \tau_{\SDE, t}.
\end{split}
\end{equation}
Then we bound the conditional variance of $\htau_{\operatorname{IPW}, t}$ given $Y_t$. Note that 
\begin{equation}
\begin{split}
&\Var{Y_{i(t+1)}  \p{\frac{W_{t}}{\pi_i} - \frac{1 - W_{t}}{1 - \pi_i}} \mid Y_t }\\
& \quad = \Var{\EE{Y_{i(t+1)}  \p{\frac{W_{t}}{\pi_i} - \frac{1 - W_{t}}{1 - \pi_i}}  \mid W_t, Y_t} \mid Y_t}\\
& \qquad\qquad\qquad \qquad \qquad  + \EE{\Var{Y_{i(t+1)}  \p{\frac{W_{t}}{\pi_i} - \frac{1 - W_{t}}{1 - \pi_i}}  \mid W_t, Y_t} \mid Y_t} \\
& \quad = \Var{f_i(Y_{it}, W_{it}, Z_{it}) \p{\frac{W_{t}}{\pi_i} - \frac{1 - W_{t}}{1 - \pi_i}} \mid Y_t}\\
& \qquad\qquad\qquad \qquad \qquad + \EE{f_i(Y_{it}, W_{it}, Z_{it})(1 - f_i(Y_{it}, W_{it}, Z_{it}))\p{\frac{W_{t}}{\pi_i} - \frac{1 - W_{t}}{1 - \pi_i}}^2 \mid Y_t}\\
& \quad \leq 2 \Var{\frac{W_{t}}{\pi_i} - \frac{1 - W_{t}}{1 - \pi_i}} = 2. 
\end{split}
\end{equation}
Since $Y_{i(t+1)}  \p{W_{t}/\pi_i - (1 - W_{t})/(1 - \pi_i)}$ are independent from each other conditioning on $Y_t$, we have that $\Var{\htau_{\operatorname{IPW}, t} \mid Y_t} \leq 2/n$. Therefore, our inverse propensity weighted estimator is consistent of the short-term direct effect. In particular, 
\begin{equation}
\htau_{\operatorname{IPW}} =  \tau_{\SDE, t} + \oo_p\p{\frac{1}{\sqrt{n}}}. 
\end{equation}

\subsection{Proof of Theorem \ref{theo:LDE_char}}
In this section, we often write $a_i(Z_i)$, $b_i(Z_i)$, $c_i(Z_i)$ and $d_i(Z_i)$ as $a_i$, $b_i$, $c_i$ and $d_i$ for short. 

For subject $i$, since $\mu(\pi_i = \gamma, \pi_{-i})$ is the stationary distribution, the expectation satisfies 
\[\EE[\mu(\pi_i = \gamma, \pi_{-i})]{Y_i} = \EE[\mu(\pi_i = \gamma, \pi_{-i})]{a_i + b_i\gamma + \p{c_i + d_i \gamma} Y_i}. \]
Rearranging the terms gives
\begin{equation}
\label{eqn:equil_eqn}
\begin{split}
&\EE[\mu(\pi_i = \gamma, \pi_{-i})]{Y_i} \EE[\mu(\pi_i = \gamma, \pi_{-i})]{1 - c_i - d_i \gamma} \\
&\qquad \qquad =  \EE[\mu(\pi_i = \gamma, \pi_{-i})]{a_i + b_i \gamma}  + \Cov[\mu(\pi_i = \gamma, \pi_{-i})]{c_i + d_i \gamma, Y_i}.
\end{split}
\end{equation}

We will first show that the distribution of $a_i, b_i, c_i$ and $d_i$ is not too different under $\mu(\pi)$ and $\mu(\pi_i = \gamma, \pi_{-i})$. Let $Y$ and $\widetilde{Y}$ be two processes satisfying Assumption \ref{assu:MDP} with the same initial distribution $Y_0 = \widetilde{Y}_0 \sim \mu(\pi)$. 
Assume that at time $0$, they have the same treatment vector except for the $i$-th one, i.e., $\widetilde{W}_j = W_j \sim \operatorname{Ber}(\pi_j)$ for $j \neq i$ while  $\widetilde{W}_i \sim \operatorname{Ber}(\gamma)$ and $W_j \sim \operatorname{Ber}(\pi_j)$ independently. Suppose further that they share the same random seed $U_0 \sim \operatorname{Unif}[0,1]$ and that $Y_{i1} = \mathbbm{1}\p{U_0 \leq f_i(Y_{i0}, W_{i0}, Z_{i0})}$ and $\widetilde{Y}_{i1} = \mathbbm{1}(U_0 \leq f_i(Y_{i0}, \widetilde{W}_{i0}, Z_{i0}))$. Then we can immediately verify that $Y_1$ and $\tilde{Y}_1$ can only differ at index $i$. Thus 
\[
d_E\p{Y_{1} , \widetilde{Y}_{1}} \leq  \EE{\abs{Y_{i1} - \widetilde{Y}_{i1}}} \leq 1,
\]
This further implies that 
\[W_{d_E} \p{ \law(Y_{1}), \law(\widetilde{Y}_{1})} \leq 1. \]
By the same argument as in the proof of Theorem \ref{theo:mean_field2}, we have that
\[W_{L_1} \sqb{ \mu(\pi_i = \gamma, \pi_{-i}), \mu}  \leq 1/(1-C) \textnormal{ and } W_{d_E} \sqb{ \mu(\pi_i = \gamma, \pi_{-i}), \mu}  \leq 1/(1-C) .\]
Thus there exist a coupling of random variables $Y$ and $\breve{Y}$, such that $Y \sim \mu(\pi)$, $\breve{Y} \sim \mu(\pi_i = \gamma, \pi_{-i})$, and $d_E(Y, \breve{Y} ) \leq 1/(1-C)$.  With this coupling, we have
\[ \abs {\EE[\mu(\pi_i = \gamma, \pi_{-i})]{a_i + b_i \gamma} - \EE[\mu]{a_i + b_i \gamma} } \leq 2L_n d_E(Y, \breve{Y} ) \leq \frac{2L_n}{1 - C}. \]
Similarly, $\abs {\EE[\mu(\pi_i = \gamma, \pi_{-i})]{1 - c_i  - d_i \gamma} - \EE[\mu]{1 - c_i - d_i \gamma} } \leq 2L_n d_E(Y, \breve{Y} ) \leq \frac{2L_n}{1 - C}. $

Plugging the above results back into \eqref{eqn:equil_eqn}, we get
\begin{equation}
\label{eqn:LDE_key_equation}
\begin{split}
&\abs{\EE[\mu(\pi_i = \gamma, \pi_{-i})]{Y_i} -  \frac{\EE[\mu]{a_i + b_i \gamma}}{\EE[\mu]{1 - c_i - d_i \gamma}} }\\
&\leq \abs{\frac{\EE[\mu]{a_i + b_i \gamma}}{\EE[\mu]{1 - c_i - d_i \gamma}}  - \frac{\EE[\mu(\pi_i = \gamma, \pi_{-i})]{a_i + b_i \gamma}}{\EE[\mu(\pi_i = \gamma, \pi_{-i})]{1 - c_i - d_i \gamma}}} 
 + \frac{\abs{\Cov[\mu(\pi_i = \gamma, \pi_{-i})]{c_i + d_i \gamma, Y_i}}}{\EE[\mu(\pi_i = \gamma, \pi_{-i})]{1 - c_i - d_i \gamma}}\\
&\leq \frac{2L_n}{1 - C} \p{ \frac{1}{(1-B)^2} + \frac{1}{1-B}} + \frac{1}{1-B}\abs{\Cov[\mu(\pi_i = \gamma, \pi_{-i})]{c_i + d_i \gamma, Y_i}}. 
\end{split}
\end{equation}

To study the covariance term, note that as $Y_i$ is bounded above by 1 and below by 0, the covariance term can be bounded by
\begin{align*}
&\abs{\Cov[\mu(\pi_i = \gamma, \pi_{-i})]{c_i + d_i \gamma, Y_i} }\\
& \qquad \qquad =\abs{ \EE[\mu(\pi_i = \gamma, \pi_{-i})]{ \p{Y_i - \EE[\mu(\pi_i = \gamma, \pi_{-i})]{Y_i}} \p{c_i + d_i \gamma - \EE[\mu(\pi_i = \gamma, \pi_{-i})]{c_i + d_i \gamma}} }}\\
& \qquad \qquad \leq \EE[\mu(\pi_i = \gamma, \pi_{-i})]{\abs{ c_i+d_i \gamma  - \EE[\mu(\pi_i = \gamma, \pi_{-i})]{c_i + d_i \gamma}  }}
\end{align*}
By the same argument as above, this term can be approximated by 
$\EE[\mu]{\abs{ c_i+d_i \gamma  - \EE[\mu]{c_i + d_i \gamma}  }}$ with a difference bounded by $4L_n/(1-C)$. 


We move on to show that $\EE[\mu]{\abs{ c_i+d_i \gamma  - \EE[\mu]{c_i + d_i \gamma}  }}$ is small. One the one hand, we show that $\EE[\mu]{\abs{ c_i+d_i \gamma  - \EE[\mu]{c_i + d_i \gamma}  }}$ is 
close to $\EE[\nu]{\abs{ c_i+d_i \gamma  - \EE[\nu]{c_i + d_i \gamma}  }}$, where under $\nu$, $Y_i \sim \operatorname{Ber}(\Ps_i)$ independently. Recall that $\Ps$ is the fixed point of the system \eqref{eqn:dynamic_system}. Theorem \ref{theo:mean_field2} implies that there exist random vectors $Y$ and $\Ys$, such that $Y \sim \mu$, $Y^* \sim \nu$, and $\max_i \sum_{j \in \mathcal{N}_j} \EE{\abs{Y_j - \Ys_j}} \leq L_n D_n^{\frac{3}{2}}/(2(1-C))$. Then some simple algebra gives 
\begin{align*}
&\abs{\EE[\mu]{\abs{ c_i+d_i \gamma  - \EE[\mu]{c_i + d_i \gamma}  }} 
- \EE[\nu]{\abs{ c_i+d_i \gamma  - \EE[\nu]{c_i + d_i \gamma}  }}}\\
&\qquad \qquad \leq 2\EE{c_i(Z_i) + d_i(Z_i) \gamma - \p{c_i(\Zs_i) + d_i(\Zs_i) \gamma}}\\
&\qquad \qquad \leq 4L_n  \sum_{j \in \ngb_i} \EE{\abs{Y_j - Y_j^*}} \leq  2L_n^2 D_n^{\frac{3}{2}}/(1-C). 
\end{align*}

On the other hand, we note that the term $\EE[\nu]{\abs{ c_i+d_i \gamma  - \EE[\nu]{c_i + d_i \gamma}  }}$ is small. Specifically, under $\nu$, $Y_i$'s are independent. Thus
\[
\begin{split}
\EE[\nu]{\abs{ c_i+d_i \gamma  - \EE[\nu]{c_i + d_i \gamma}  }}
&\leq \Var[\nu]{c_i + d_i \gamma}
\leq 4L_n^2 \operatorname{Var}_{\nu}\Big[\smash{\sum_{j \in \ngb_i} Y_j}\Big]\\
& = 4L_n^2 \sum_{j \in \ngb_i} \Ps_j(1 - \Ps_j) \leq L_n^2 D_n. 
\end{split}
\]
Combining the two results, we get 
\begin{equation}
\label{eqn:corr_c_i_Y}
\EE[\mu]{\abs{ c_i+d_i \gamma  - \EE[\mu]{c_i + d_i \gamma}  }} \leq 3L_n^2 D_n^{\frac{3}{2}}/(1-C).
\end{equation}

Finally, plugging things back into \eqref{eqn:LDE_key_equation}, we get 
\[
\begin{split}
&\abs{\EE[\mu(\pi_i = \gamma, \pi_{-i})]{Y_i} -  \frac{\EE[\mu]{a_i + b_i \gamma}}{\EE[\mu]{1 - c_i - d_i \gamma}} }\\
&\qquad \qquad \leq \frac{2L_n}{1 - C} \p{2 + \frac{1}{(1-B)^2} + \frac{1}{1-B}} + \frac{1}{1-B}3L_n^2 D_n^{\frac{3}{2}}/(1-C)\\
& \qquad \qquad = \oo\p{L_n\sqrt{D_n}} = \pp\p{\sqrt{L_n}},
\end{split}
\]
since $D_n L_n \leq B \leq C < 1$.

\subsection{Proof of Proposition \ref{prop:estimate_LDE}}

We will focus on $\widehat{f_i(0,0)}$. Results for other combinations of $y$ and $w$ can be derived similarly. Note that $f_i(0,0, Z_i) = a(Z_i)$. For notation simplicity and consistency with the text later, we also write $\hat{a}_i = \widehat{f_i(0,0)}$. 

We will take the following three steps in this proof. 
\begin{enumerate}
\item Establish that $\frac{1}{T} \sum_{t = 1}^T (1 - W_{i t}) (1 - Y _{i t})$ converges to $\EE[\mu(\pi)]{(1 - W_{i t}) (1 - Y _{i t})}$. 
\item Establish that $\frac{1}{T} \sum_{t = 1}^T Y_{i(t+1)}(1 - W_{i t}) (1 - Y _{i t})$ converges to $\EE[\mu(\pi)]{Y_{i(t+1)}(1 - W_{i t}) (1 - Y _{i t})}$. 
\item Show that $\EE[\mu(\pi)]{a_i(Z_{it}) \mid Y_{it} = 1}$ is close to $\EE[\mu(\pi)]{a_i(Z_{it})}$. 
\end{enumerate}

We will start with 1. We will make use of Lemma \ref{lemma:markov_exp_ergodic}. 
Let $X_t$ and $Y_t$ be two processes satisfying Assumption \ref{assu:MDP} and \ref{assu:Bern}. Assume that the initial distribution of $X$ is the stationary distribution $\mu(\pi)$. Lemma \ref{lemma:contraction_d_e} implies that $W_{d_E}(\law(X_t), \law(Y_t)) \leq C^t W_{d_E}(\law(X_0), \law(Y_0))$ for any $t \geq 0$, where $C<1$ is the constant in Assumption \ref{assu:contraction}. In particular, there exists a coupling of the processes $X$ and $Y$ such that $d_E(X_t, Y_t) \leq C^t W_{d_E}(\law(X_0), \law(Y_0))$ for any $t$. Note that we can further bound $C^t W_{d_E}(\law(X_0), \law(Y_0))$ by $C^t D_n$. 
With such coupling, we can bound the difference between $\EE{(1 - W_{it})(1 - Y_{it})}$ and $\EE{(1 - W_{it})(1 - X_{it})}$. In particular, 
\[
\begin{split}
&  \abs{\EE{(1 - W_{it})(1 - Y_{it})} - \EE[\pi]{(1 - W_{it})(1 - X_{it})}}
= (1-\pi_i) \abs{\EE{X_{it} - Y_{it}}}\\
&\qquad \qquad \leq (1-\pi_i) L_n d_E(X_{t-1}, Y_{t-1}) \leq C^{t-1} L_n D_n \leq C^t. 
\end{split}
\]
Thus Lemma \ref{lemma:markov_exp_ergodic} implies that 
\begin{equation}
\label{eqn:ergodic_denom}
\EE{\p{\frac{1}{T} \sum_{t = 1}^T (1 - W_{it})(1 - Y_{it}) - \EE[\mu\p{\pi}]{ (1 - W_{it})(1 - Y_{it})}}^2} \leq \frac{2(1+C)}{C(1-C)T}.  
\end{equation}

To show 2, we use similar methods as used in 1. It is slightly more complicated though since $Y_{i(t+1)}$ is not independent of $Y_{it}$. Again, we will make use of Lemma \ref{lemma:contraction_d_e} and Lemma \ref{lemma:markov_exp_ergodic}. Lemma \ref{lemma:contraction_d_e} implies that for any initial distribution $\mu_1$, there exist two processes $X$ and $Y$ satisfying Assumptions \ref{assu:MDP} and \ref{assu:Bern} such that $X_0 \sim \mu(\pi)$, $Y_0 \sim \mu_1$, $d_E(X_t, Y_t) \leq C^t d_E(X_0, Y_0)$ and $\EE{|X_{it} - Y_{it}|} = \abs{\EE{X_{it} - Y_{it}}}$ for any $t \geq 1$,. The last equality comes from the fact that in construction of the coupling, $X$ and $Y$ share the same random seed $U$. Thus,
\[
\begin{split}
& \abs{\EE{Y_{i(t+1)}(1 - W_{it})(1 - Y_{it})}- \EE{X_{i(t+1)}(1 - W_{it})(1 - X_{it})}}\\
&  \qquad \qquad= (1 - \pi_i) \abs{\EE{a_i(Z_{it}) (1 - Y_{it}) - a_i(V_{it})(1 - X_{it})} }\\
&  \qquad \qquad \leq (1 - \pi_i)  \p{\abs{\EE{(a_i(Z_{it}) - a_i(V_{it}))(1 - Y_{it})} } +  \abs{\EE{a_i(V_{it}) ( Y_{it}- X_{it})} } }\\
&  \qquad \qquad \leq \EE{\abs{a_i(Z_{it}) - a_i(V_{it})}}+  \EE{\abs{Y_{it}- X_{it}}}, 
\end{split}
\]
where as usual, $Z_{it} = \sum_{j \in \ngb_i }Y_{it}$ and $V_{it} =  \sum_{j \in \ngb_i }X_{it}$. 
The first term $\EE{\abs{a_i(Z_{it}) - a_i(V_{it})}}$ can be bounded by
\[
\EE{\abs{a_i(Z_{it}) - a_i(V_{it})}} \leq L_n d_E(X_t, Y_t) \leq C^t L_n d_E(X_0, Y_0) \leq C^t L_n D_n \leq C^{t+1}. 
\]
The second term $\EE{\abs{Y_{it}- X_{it}}}$ can be bounded by 
\[\EE{\abs{Y_{it}- X_{it}}} \leq L_n d_E(X_{t-1}, Y_{t-1}) \leq L_n C^{t-1} d_E(X_0, Y_0) \leq L_n D_n C^{t-1} \leq C^t. \]
Combining the two terms, we get 
\[\abs{\EE{Y_{i(t+1)}(1 - W_{it})(1 - Y_{it})}- \EE{X_{i(t+1)}(1 - W_{it})(1 - X_{it})}} \leq (1+C) C^t. \]
Then, by Lemma \ref{lemma:markov_exp_ergodic}, 
\begin{equation}
\label{eqn:ergodic_numer}
 \EE{\p{\frac{1}{T} \sum_{t = 1}^T Y_{i(t+1)}(1 - W_{it})(1 - Y_{it}) - \EE[\mu\p{\pi}]{ Y_{i(t+1)}(1 - W_{it})(1 - Y_{it})}}^2} \leq \frac{(1+C)(2+C)}{C(1-C)T}. 
\end{equation}

The two inequalities \eqref{eqn:ergodic_denom} and \eqref{eqn:ergodic_numer} together show that $\hat{a}_i$ is close to the following quantity
\[
\begin{split}
&\frac{\EE[\mu\p{\pi}]{ Y_{i(t+1)}(1 - W_{it})(1 - Y_{it})}}{\EE[\mu\p{\pi}]{ (1 - W_{it})(1 - Y_{it})}}
= \frac{\EE[\mu\p{\pi}]{ a_i(Z_{it})(1 - W_{it})(1 - Y_{it})}}{\EE[\mu\p{\pi}]{ (1 - W_{it})(1 - Y_{it})}}\\
& \qquad \qquad  = \frac{(1 - \pi_i) \EE[\mu\p{\pi}]{ a_i(Z_{it})(1 - Y_{it})}}{ (1 - \pi_i) \EE[\mu\p{\pi}]{ (1 - Y_{it})}} = \EE[\mu\p{\pi}]{ a_i(Z_{it}) \mid Y_{it} =0 }. 
\end{split}
\]
Specifically, by Lemma \ref{lemma:little_bound_ratio}, we can show that if $\pi_i < 1$ and $\EE[\mu\p{\pi}]{ Y_{it}} < 1$, then
\[ \EE{\p{\hat{a}_i - \EE[\mu\p{\pi}]{a_i(Z_{it} )\mid Y_{it} =0 }}^2} \leq \frac{C_0}{T},  \]
for some constant $C_0$. 

It then remains to study point 3, i.e., to show that $\EE[\mu(\pi)]{a_i(Z_{it}) \mid Y_{it} = 0}$ is close to $\EE[\mu(\pi)]{a_i(Z_{it})}$. Note that 
\[
\begin{split}
& \abs{\EE[\mu(\pi)]{a_i(Z_{it}) \mid Y_{it} = 0} - \EE[\mu(\pi)]{a_i(Z_{it})}} \\
& \qquad \qquad = \abs{\frac{\EE[\mu(\pi)]{a_i(Z_{it}) (1 - Y_{it})} - \EE[\mu(\pi)]{a_i(Z_{it})} \EE[\mu(\pi)]{1 - Y_{it}} }{\EE[\mu(\pi)]{1 - Y_{it}}}}\\
& \qquad \qquad = \frac{\abs{\Cov[\mu(\pi)]{a_i(Z_i), Y_{it}}}}{\EE[\mu(\pi)]{1 - Y_{it}}}. 
\end{split}
\]
To show that $\Cov[\mu(\pi)]{a_i(Z_i), Y_{it}}$ is small, we follow the proof of Theorem \ref{theo:LDE_char}, where we showed that $\abs{\Cov[\mu(\pi_i = \gamma, \pi_{-i})]{c_i + d_i \gamma, Y_i} }$ is small. With the exact same steps (for simplicity, we omit details here), we get $\Cov[\mu(\pi)]{a_i(Z_i), Y_{it}} \leq 2L_n^2 D_n^{\frac{3}{2}}/(1-C) \leq 2 C^{\frac{3}{2}} \sqrt{L_n}/(1-C)$. 

Therefore, 
\[ \EE{\p{\hat{a}_i - \EE[\mu\p{\pi}]{a_i(Z_{it} )}}^2} \leq C_1 \p{\frac{1}{T} + L_n}.  \]

\subsection{Proof of Theorem \ref{theo:LTE_char2}}
We will show that $\frac{1}{n}\sum_{i = 1}^n\p{ \Ps_i(\pi + \Delta\bv) - \Ps_i(\pi)}$ is close to $\frac{\Delta}{n}\sum_{i = 1}^n \p{\nabla_{\pi} \Ps_i(\pi) \trans \bv}$. For notation simplicity, we write $g_i(\Delta) =  \Ps_i(\pi + \Delta\bv)$. 
Our proof consists of three main steps:
\begin{enumerate}
\item Let $g'(\Delta)$ be the vector of $(g'_1(\Delta), \dots, g'_n(\Delta)$. We will show that $\max_{\Delta} \Norm{g'(\Delta)}_{\infty}$ is small.
\item Let $g''(\Delta)$ be the vector of $(g''_1(\Delta), \dots, g''_n(\Delta)$. We will show that $\max_{\Delta} \Norm{g''(\Delta)}_{\infty}$ is small based on the results in the previous step.
\item Using the bound on the second derivative, we will establish that $\frac{1}{n}\sum_{i = 1}^n\p{ \Ps_i(\pi + \Delta\bv) - \Ps_i(\pi) }
\approx \frac{\Delta}{n}\sum_{i = 1}^n \p{\nabla_{\pi} \Ps_i(\pi) \trans \bv}$ and derive \eqref{eqn:LTE_char1}. 
\end{enumerate}

We start with the first step. Taking derivative of both hand sides of \eqref{eqn:fix_point} with respect to $\Delta$ gives
\begin{equation}
\label{eqn:LTE_char_first_dir}
\begin{split}
 g_i'(\Delta)
&= b_i\p{\Qs_{i}(\pi + \Delta \bv)}v_i  + d_i\p{\Qs_{i}(\pi+ \Delta \bv)} \Ps_{i}(\pi+ \Delta \bv) v_i \\
& \qquad \qquad + \sqb{c_i\p{\Qs_{i}(\pi+ \Delta \bv)} + d_i\p{\Qs_{i}(\pi+ \Delta \bv)} (\pi_i+ \Delta v_i)}  g_i'(\Delta) \\
& \qquad \qquad + f_i'(\Ps_i(\pi + \Delta \bv), \pi_i + \Delta v_i, \Qs_i(\pi + \Delta \bv)) \sum_{j \in \ngb_i} g_j'(\Delta). 
\end{split}
\end{equation}
Similar to what we have in \eqref{eqn:first_dir_matrix}, we get that $g'(\Delta)$ can be written as
\begin{equation}
g'(\Delta) = (I - D(\Delta) A - W(\Delta)) ^{-1} \bu(\Delta),  
\end{equation}
where $A$ is the adjacency matrix of the interference graph, $D(\Delta) = \operatorname{diag}\big(f_i'(\Ps_i(\pi + \Delta \bv), \pi_i + \Delta v_i, \Qs_i(\pi + \Delta \bv))\big)$, $W(\Delta) =  \operatorname{diag} \big( c_i\p{\Qs_{i}(\pi+ \Delta\bv)} + d_i\p{\Qs_{i}(\pi+ \Delta\bv)} (\pi_i + \Delta v_i)\big)$, and $\bu(\Delta) = \operatorname{vec} \big( v_i \p{b_i\p{\Qs_{i}(\pi + \Delta\bv)} + d_i\p{\Qs_{i}(\pi+ \Delta\bv)} \Ps_{i}(\pi+ \Delta\bv)}  \big)$.
By Assumption \ref{assu:bound_fi}, we have $\big|f_i'(\Ps_i(\pi + \Delta \bv), \pi_i + \Delta v_i, \Qs_i(\pi + \Delta \bv))\big| \leq L_n$. Thus $D(\Delta)$ is a diagonal matrix, all of whose diagonal entries have their absolute value bounded by $L_n$. For $W(\Delta)$, again by Assumption \ref{assu:bound_fi}, we have $\abs{c_i\p{\Qs_{i}(\pi+ \Delta\bv)} + d_i\p{\Qs_{i}(\pi+ \Delta\bv)} (\pi_i + \Delta v_i)} \leq B$. Thus all of $W(\Delta)$'s diagonal entries have their absolute value bounded by $B$. Therefore, the matrix $M = I - D(\Delta) A - W(\Delta)$ satisfies the following
\[
M_{i,i} \geq 1 - B, \qquad \sum_{j \neq i}\abs{M_{i,j}} \geq - L_n D_n. 
\]
For the vector $\bu(\Delta)$, the absolute value of its $i$-th entry is bounded by $\abs{u_i(\Delta)} = v_i \big| f_i(\Ps_i(\pi+ \Delta\bv), 1, \Qs_{i}(\pi+ \Delta\bv)) - f_i(0, 1, \Qs_{i}(\pi+ \Delta\bv))\big| \leq v_i \leq C_v$. Thus by Lemma \ref{lemma:l_infty_norm}, 
\[ \Norm{g'(\Delta)}_{\infty} \leq C_v/(1 - B - L_n D_n) \leq \frac{C_v}{1 - C}. \]


Now we move on to the second step. Taking derivative of both hand sides of \eqref{eqn:LTE_char_first_dir}, we get
\begin{equation}
\label{eqn:LTE_char_second_dir}
\begin{split}
 g_i''(\Delta)
&= \sqb{c_i\p{\Qs_{i}(\pi+ \Delta \bv)} + d_i\p{\Qs_{i}(\pi+ \Delta \bv)} (\pi_i+ \Delta v_i)}  g_i''(\Delta) \\
& \qquad \qquad + f_i'(\Ps_i(\pi + \Delta \bv), \pi_i + \Delta v_i, \Qs_i(\pi + \Delta \bv)) \sum_{j \in \ngb_i} g_j''(\Delta)\\
&  \qquad \qquad + f_i''(\Ps_i(\pi + \Delta \bv), \pi_i + \Delta v_i, \Qs_i(\pi + \Delta \bv)) \big(\sum_{j \in \ngb_i} g_j'(\Delta)\big)^2\\
& \qquad \qquad +2 d_i\p{\Qs_{i}(\pi+ \Delta \bv)} v_i  g_i'(\Delta)  \\
&  \qquad \qquad +  2\sqb{b_i'\p{\Qs_{i}(\pi + \Delta \bv)} + d_i'\p{\Qs_{i}(\pi+ \Delta \bv)} \Ps_{i}(\pi+ \Delta \bv) } v_i  \sum_{j \in \ngb_i} g_j'(\Delta)\\
& \qquad \qquad + 2\sqb{c_i'\p{\Qs_{i}(\pi+ \Delta \bv)} + d_i'\p{\Qs_{i}(\pi+ \Delta \bv)} (\pi_i+ \Delta v_i)}  g_i'(\Delta) \sum_{j \in \ngb_i} g_j'(\Delta). 
\end{split}
\end{equation}
This then implies that 
\begin{equation}
g'(\Delta) = (I - D(\Delta) A - W(\Delta)) ^{-1} \bxi(\Delta),  
\end{equation}
where $\bxi$ is a vector with entries $\xi_i$ being the sum of the last four lines of \eqref{eqn:LTE_char_second_dir}. We can then bound $\abs{\xi_i}$ by
\[
\begin{split}
\abs{\xi_i} 
&\leq L_{2,n} D_n^2 \Norm{g'(\Delta)}^2_{\infty} + 4 \Norm{g'(\Delta)}_{\infty} + 2L_n D_n \Norm{g'(\Delta)}_{\infty} + 2 L_n D_n (\Norm{g'(\Delta)}_{\infty})^2\\
&\leq C_{\xi 0} (L_{2,n} D_n^2 + L_nD_n + 1)
\leq C_{\xi},
\end{split}
\]
for some constant $C_{\xi 0}$ and $C_{\xi}$ not depending on $n$. Then again, by Lemma \ref{lemma:l_infty_norm}, we have
\[ \Norm{g''(\Delta)}_{\infty} \leq C_{\xi}/(1 - B - L_n D_n) \leq \frac{C_{\xi}}{1 - C}. \]

Finally, we note that the above implies that for any $i \in \cb{1, \dots, n}$, $g_i''(\Delta) \leq \frac{C_{\xi}}{1 - C}$. Therefore by Taylor expansion, we can show that
\[ \abs{ \frac{g_i(\Delta) - g_i(0)}{\Delta} - g'_i(\Delta) }
=  \abs{\frac{g_i(0) + g'_i(0)\Delta + \frac{1}{2}g''_i(\tilde{\Delta}} \Delta^2 - g_i(0)}{\Delta} - g'_i(\Delta)
= \abs{ \frac{1}{2\Delta} g''_i(\tilde{\Delta}) \Delta^2  }
\leq \frac{C_{\xi}}{2(1 - C)} \Delta. 
 \]
Switching back to the $\Ps$ notation, we have that
\[\abs{\p{ \Ps_i(\pi + \Delta\bv) - \Ps_i(\pi) }
- \Delta \p{\nabla_{\pi} \Ps_i(\pi) \trans \bv}} \leq \frac{C_{\xi}}{2(1 - C)} \Delta,\]
and thus 
\[\abs{\frac{1}{n}\sum_{i = 1}^n\p{ \Ps_i(\pi + \Delta\bv) - \Ps_i(\pi) }
- \frac{\Delta}{n}\sum_{i = 1}^n \p{\nabla_{\pi} \Ps_i(\pi) \trans \bv}} \leq \frac{C_{\xi}}{2(1 - C)} \Delta. \]

The rest follows directly from \eqref{eqn:theo_LTE_char0} and \eqref{eqn:theo_LTE_char1}. 

\subsection{Proof of Theorem \ref{theo:LTE_estimation}}
We focus on $f'_i(1,1, \Qs_i(\pi))$. The rest can be shown with the same proof techniques. For $y=1$ and $w=1$, we can write $\widehat{f'_i(y,w)}_{\delta_T}$ in a slightly simpler way:
\begin{equation}
\label{eqn:f_prime_11}
\begin{split}
\widehat{f'_i(1,1)}_{\delta_T} & = \frac{\sum_{t = 1}^T Y_{it} W_{it} \p{Y_{i(t+1)} - \bar{Y_i}(1,1)} \p{Z_{it} - \bar{Z_i}(y,w) }}{D_n T \delta_T \vee \sum_{t = 1}^T Y_{it} W_{it} \p{ Z_{it} -  \bar{Z_i}(1,1) }^2}\\
& = \frac{\sum_{t = 1}^T Y_{it} W_{it} Y_{i(t+1)} Z_{it} - T_{1,1} \bar{Y_i}(1,1) \bar{Z_i}(1,1)}{D_n T \delta_T \vee \p{\sum_{t = 1}^T Y_{it} W_{it} Z_{it}^2 - T_{1,1} \bar{Z_i}(1,1)^2}},
\end{split}
\end{equation}
where $T_{1,1} = \sum_{t = 1}^T Y_{it} W_{it}$,
$\bar{Y_i}(1,1) = \sum_{t = 1}^T Y_{it} W_{it} Y_{i(t+1)} /T_{1,1}$, and
$\bar{Z_i}(1,1) = \sum_{t = 1}^T Z_{it} W_{it} Z_{it} /T_{1,1}$. 

We will follow the following steps in the proof. 
\begin{enumerate}
\item Analyze each term in \eqref{eqn:f_prime_11} and show that $\widehat{f'_i(1,1)}_{\delta_T} $ converges to a limit free of $T$ as $T \to \infty$.
\item Analyze the behavior of the above limit. 
\end{enumerate}

We start with analyzing the behavior of $T_{1,1}$, $\bar{Y_i}(1,1)$ and $\bar{Z_i}(1,1)$. We will make use of Lemma \ref{lemma:contraction_d_e} and Lemma \ref{lemma:markov_exp_ergodic}. Lemma \ref{lemma:contraction_d_e} implies that there exists a process $X$ satisfying Assumptions \ref{assu:MDP} and \ref{assu:Bern} such that $X_0 \sim \mu(\pi)$, $d_E(X_t, Y_t) \leq C^t d_E(X_0, Y_0)$ and $\EE{|X_{it} - Y_{it}|} = \abs{\EE{X_{it} - Y_{it}}}$ for any $t \geq 1$. The last equality comes from the fact that in construction of the coupling, $X$ and $Y$ share the same random seed $U$. With such coupling, we have
\[
\begin{split}
& \abs{\EE{Y_{it}W_{it} - X_{it}W_{it} }} \leq (1 - \pi_i) \EE{\abs{Y_{it}  - X_{it}}},  
\end{split}
\]
where $\EE{\abs{Y_{it} - X_{it}}}$ can be bounded by
\begin{equation}
\label{eqn:bound_E_Y_X}
\EE{\abs{Y_{it} - X_{it}}} \leq L_n d_E(X_{t-1}, Y_{t- 1}) \leq L_n C^{t-1} d_E(X_0, Y_0) \leq L_n D_n C^{t-1} \leq C^t.  
\end{equation}
Therefore, Lemma \ref{lemma:markov_exp_ergodic} implies that 
\begin{equation}
\label{eqn:T11_bound}
\EE{ \p{ \frac{1}{T}\sum_{t= 1}^T Y_{it}W_{it} - \EE[\mu(\pi)]{Y_{it}W_{it} } }}^2 \leq \frac{2}{C(1-C) T}. 
\end{equation}
In words, this shows that $T_{1,1}$ is close to $T \EE[\mu(\pi)]{Y_{it}W_{it} }$. 

Now we can analyze $\bar{Y_i}(1,1)$ similarly. Using the same coupling as above, we get, 
\[
\begin{split}
& \abs{\EE{Y_{it}W_{it}Y_{i(t+1)} - X_{it}W_{it}X_{i(t+1)}}}\\
& \qquad \qquad = \abs{\EE{Y_{it}W_{it}f_i(1,1,Z_{it}) - X_{it}W_{it}f_i(1,1,V_{it})}}\\
& \qquad \qquad = \pi_i \abs{\EE{Y_{it} f_i(1,1,Z_{it}) - X_{it}f_i(1,1,V_{it})}}\\
& \qquad \qquad \leq \pi_i \cb{ \EE{\abs{Y_{it} - X_{it}}} + \EE{\abs{f_i(1,1,Z_{it}) - f_i(1,1,V_{it})}}},
\end{split}
\]
where as usual, $Z_{it} = \sum_{j \in \ngb_i} Y_{it}$ and $V_{it} = \sum_{j \in N_i} X_{it}$. 
We have shown in \eqref{eqn:bound_E_Y_X} that the first term $\EE{\abs{Y_{it} - X_{it}}} \leq C^t$. 
The second term $\EE{\abs{f_i(1,1,Z_{it}) - f_i(1,1,V_{it})}}$ can be bounded 
\begin{equation}
\label{eqn:bound_f_Z_V}
\EE{\abs{f_i(1,1,Z_{it}) - f_i(1,1,V_{it})}} \leq L_n d_E(X_t, Y_t) \leq C^{t} L_n d_E(X_0, Y_0) \leq C^t L_n D_n \leq C^{t+1}.  
\end{equation}
Combining the two terms, we get for any $t \geq 1$, 
\[ \abs{\EE{Y_{it}W_{it}Y_{i(t+1)} - X_{it}W_{it}X_{i(t+1)}}} \leq (1+C) C^t. \]
Therefore, Lemma \ref{lemma:markov_exp_ergodic} implies that 
\begin{equation}
\label{eqn:Y_bar_bound1}
\EE{ \p{ \frac{1}{T}\sum_{t= 1}^T Y_{it}W_{it}Y_{i(t+1)} - \EE[\mu(\pi)]{Y_{it}W_{it}Y_{i(t+1)}}}  }^2 \leq \frac{(1+C)(2+C)}{C(1-C) T}. 
\end{equation}
Applying Lemma \ref{lemma:little_bound_ratio} to \eqref{eqn:Y_bar_bound1} and \eqref{eqn:T11_bound}, we get 
\begin{equation}
\EE{\p{\bar{Y_i}(1,1) - \EE{f_i(1,1,Z_i) \mid Y_1 = 1, W_i = 1}}^2} \leq \p{\frac{1}{ \pi_i \EE[\mu(\pi)]{Y_i}}}^2 \frac{16}{C(1-C)} \frac{1}{T}. 
\end{equation}

We will then analyze $\bar{Z_i}(1,1)$. Again, using the same coupling of $X$ and $Y$, we have
\[
\begin{split}
& \abs{\EE{Y_{it}W_{it}Z_{it} - X_{it}W_{it}V_{it}}}  = \pi_i \abs{\EE{Y_{it} Z_{it} - X_{it}V_{it}}}\\
& \qquad \qquad \leq \pi_i \cb{ D_n \EE{\abs{Y_{it} - X_{it}}} + \EE{\abs{Z_{it} - V_{it} }}},
\end{split}
\]
The first term $\EE{\abs{Y_{it} - X_{it}}} \leq C^t$ by \eqref{eqn:bound_E_Y_X}, while the second term $\EE{\abs{Z_{it} - V_{it} }}$ can be bounded by 
\begin{equation}
\label{eqn:bound_Z_V}
\EE{\abs{Z_{it} - V_{it} }} \leq d_E(X_{t}, Y_{t}) \leq C^t d_E(X_{0}, Y_{0}) \leq C^{t} D_n.  
\end{equation}
Therefore, $ \abs{\EE{Y_{it}W_{it}Z_{it} - X_{it}W_{it}V_{it}}} \leq 2C^t D_n$. 
Therefore, Lemma \ref{lemma:markov_exp_ergodic} implies that 
\begin{equation}
\label{eqn:Z_bar_bound1}
\EE{ \p{ \frac{1}{T}\sum_{t= 1}^T Y_{it}W_{it}Z_{it} - \EE[\mu(\pi)]{Y_{it}W_{it}Z_{it}} } }^2 \leq \frac{3D_n^2}{C(1-C) T}. 
\end{equation}
Applying Lemma \ref{lemma:little_bound_ratio} to \eqref{eqn:Z_bar_bound1} and \eqref{eqn:T11_bound}, we get 
\begin{equation}
\EE{\p{\bar{Z_i}(1,1) - \EE{Z_i \mid Y_1 = 1, W_i = 1}}^2} \leq \p{\frac{1}{ \pi_i \EE[\mu(\pi)]{Y_i Z_i}}}^2 \frac{8 D_n^2}{C(1-C)} \frac{1}{T}.
\end{equation}

Using similar methods, we can also analyze $\sum_{t = 1}^T Y_{it} W_{it} Y_{i(t+1)} Z_{it}$ and $\sum_{t = 1}^T Y_{it} W_{it} Z_{i}^2$. Using the coupling of $X$ and $Y$, we have 
\[
\begin{split}
& \abs{\EE{Y_{it}W_{it}Y_{i(t+1)}Z_{it} - X_{it}W_{it}X_{i(t+1)}V_{it}}}\\
& \qquad \qquad = \abs{\EE{Y_{it}W_{it}f_i(1,1,Z_{it})Z_{it} - X_{it}W_{it}f_i(1,1,V_{it}) V_{it} } }\\
& \qquad \qquad = \pi_i \abs{\EE{Y_{it} f_i(1,1,Z_{it}) Z_{it} - X_{it}f_i(1,1,V_{it}) V_{it}}}\\
& \qquad \qquad \leq \pi_i \cb{ D_n \EE{\abs{Y_{it} - X_{it}}} + D_n \EE{\abs{f_i(1,1,Z_{it}) - f_i(1,1,V_{it})}} + \EE{\abs{Z_{it} - V_{it}}} }\\
& \qquad \qquad \leq \pi_i\cb{D_n C^t + D_n C^{t+1} + D_n C^t}
\leq 3 \pi_i D_n C^t,
\end{split}
\]
where the last line follows from \eqref{eqn:bound_E_Y_X}, \eqref{eqn:bound_f_Z_V} and \eqref{eqn:bound_Z_V}. Thus Lemma \ref{lemma:markov_exp_ergodic} implies that 
\begin{equation}
\label{eqn:YZ_bar_bound1}
\EE{ \p{ \frac{1}{T}\sum_{t= 1}^T Y_{it}W_{it}Y_{i(t+1)}Z_{it} - \EE[\mu(\pi)]{Y_{it}W_{it}Y_{i(t+1)} Z_{it}}}}^2 \leq \frac{4 D_n^2}{C(1-C) T}. 
\end{equation}

For $\sum_{t = 1}^T Y_{it} W_{it} Z_{i}^2$, we have \[
\begin{split}
& \abs{\EE{Y_{it}W_{it}Z_{it}^2 - X_{it}W_{it}X_{i(t+1)}V_{it}^2}}
 = \pi_i \abs{\EE{Y_{it} Z_{it}^2 - X_{it}V_{it}^2}}\\
& \qquad \qquad \leq \pi_i \cb{ D_n^2 \EE{\abs{Y_{it} - X_{it}}} +  2D_n \EE{\abs{Z_{it} - V_{it}}} }\\
& \qquad \qquad \leq \pi_i\cb{D_n C^t + 2 D_n^2 C^t}
\leq 3 \pi_i D_n^2 C^t,
\end{split}
\]
where the last line follows from \eqref{eqn:bound_E_Y_X}, \eqref{eqn:bound_f_Z_V} and \eqref{eqn:bound_Z_V}. Thus Lemma \ref{lemma:markov_exp_ergodic} implies that 
\begin{equation}
\label{eqn:Z_bar2_bound1}
\EE{ \p{ \frac{1}{T}\sum_{t= 1}^T Y_{it}W_{it}Z_{it}^2 - \EE[\mu(\pi)]{Y_{it}W_{it} Z_{it}^2 }}}^2 \leq \frac{4 D_n^4}{C(1-C) T}. 
\end{equation}

Now we are ready to combine the above results and analyze $\widehat{f'(1,1)}_{\delta_T}$. The numerator of \eqref{eqn:f_prime_11} (scaled by $1/T$) can be written as
\begin{equation}
\label{eqn:bound_LT_num}
\begin{split}
&\frac{1}{T}\sum_{t = 1}^T Y_{it} W_{it} Y_{i(t+1)} Z_{it} - \frac{1}{T} T_{1,1} \bar{Y_i}(1,1) \bar{Z_i}(1,1)\\
&\qquad = \frac{1}{T}\sum_{t = 1}^T Y_{it} W_{it} Y_{i(t+1)} Z_{it} - \frac{T}{T_{1,1}}  \p{ \frac{1}{T}\sum_{t= 1}^T Y_{it}W_{it}Y_{i(t+1)} } \p{ \frac{1}{T}\sum_{t= 1}^T Y_{it}W_{it}Z_{it} }\\
&\qquad = \EE[\mu(\pi)]{Y_{it}W_{it}Y_{i(t+1)} Z_{it}} + \frac{ \EE[\mu(\pi)]{Y_{it}W_{it}Y_{i(t+1)}} \EE[\mu(\pi)]{Y_{it}W_{it}Z_{it}} }{\EE[\mu(\pi)]{Y_{it}W_{it}}} + \operatorname{error}_{\operatorname{num}}\\
&\qquad = \Cov[\mu(\pi)]{Y_{i(t+1)}, Z_{it} \mid Y_{it} =1 , W_{it} = 1} + \operatorname{error}_{\operatorname{num}},
\end{split}
\end{equation}
where $\EE{\operatorname{error}_{\operatorname{num}}^2} \leq C_{\operatorname{num}}D_n^2/T$ for some constant $C_{\operatorname{num}}$. 
Here, the third line is a result of applying Lemma \ref{lemma:little_bound_ratio} to  \eqref{eqn:T11_bound}, \eqref{eqn:Y_bar_bound1} and \eqref{eqn:Z_bar_bound1}. 
Similarly, the denominator of \eqref{eqn:f_prime_11} (scaled by $1/T$ ignoring the $\delta_t$ term) can be written as
\begin{equation}
\label{eqn:bound_LT_dem}
\begin{split}
&\frac{1}{T} \sum_{t = 1}^T Y_{it} W_{it} Z_{it}^2 - T_{1,1} \bar{Z_i}(1,1)^2 = \frac{1}{T}\sum_{t = 1}^T Y_{it} W_{it} Z_{it}^2 - \frac{T}{T_{1,1}}  \p{ \frac{1}{T}\sum_{t= 1}^T Y_{it}W_{it}Z_{it} }^2\\
& \qquad \qquad  = \EE[\mu(\pi)]{Y_{it}W_{it} Z_{it}^2} - \frac{ \p{\EE[\mu(\pi)]{Y_{it}W_{it}Z_{it}}}^2 }{\EE[\mu(\pi)]{Y_{it}W_{it}}} + \operatorname{error}_{\operatorname{den}}\\
&\qquad  \qquad = \Var[\mu(\pi)]{Z_{it} \mid Y_{it} =1 , W_{it} = 1} + \operatorname{error}_{\operatorname{den}}\\
&\qquad  \qquad = \Var[\mu(\pi)]{Z_{it} \mid Y_{it} =1} + \operatorname{error}_{\operatorname{den}},\\
\end{split}
\end{equation}
where $\EE{\operatorname{error}_{\operatorname{den}}^2} \leq C_{\operatorname{den}}D_n^4/T$ for some constant $C_{\operatorname{den}}$. 
Before moving on with taking the ratio of the two, we analyze  $\Var[\mu(\pi)]{Z_{it} \mid Y_{it} =1}$. 
Note that
\begin{equation}
\label{eqn:LTE_var_lower_bound}
\begin{split}
& \Var[\mu(\pi)]{Z_{it} \mid Y_{it} =1} \geq \EE[\mu(\pi)]{\Var[\mu(\pi)]{Z_{it} \mid Y_{it} =1, W_{it} = 1, Y_{i(t-1)}}}\\
& \qquad \qquad = \EE[\mu(\pi)]{\Var[\mu(\pi)]{Z_{it} \mid Y_{i(t-1)}}}\\
& \qquad \qquad = \sum_{j \in \ngb_i}\EE[\mu(\pi)]{f_j(Y_{j(t-1)}, W_{j(t-1)}, Z_{j(t-1)})\p{1- f_j(Y_{j(t-1)}, W_{j(t-1)}, Z_{j(t-1)})}} \\
& \qquad \qquad \geq \abs{\ngb _i} C_f (1-C_f) \geq C_l C_f (1-C_f) D_n.
\end{split}
\end{equation}
Now we are ready to take the ratio of the numerator and the denominator:
\begin{equation}
\label{eqn:bound_T_LTE}
\begin{split}
\widehat{f'_i(1,1)}_{\delta_T} &= \frac{\frac{1}{T}\sum_{t = 1}^T Y_{it} W_{it} Y_{i(t+1)} Z_{it} - \frac{1}{T} T_{1,1} \bar{Y_i}(1,1) \bar{Z_i}(1,1)}{D_n \delta_T \vee\p{  \frac{1}{T}\sum_{t = 1}^T Y_{it} W_{it} Z_{it}^2 - T_{1,1} \bar{Z_i}(1,1)^2}}\\
&= \frac{ \Cov[\mu(\pi)]{Y_{i(t+1)}, Z_{it} \mid Y_{it} =1 , W_{it} = 1} + \operatorname{error}_{\operatorname{num}}}{D_n \delta_T \vee \p{ \Var[\mu(\pi)]{Z_{it} \mid Y_{it} =1 , W_{it} = 1} + \operatorname{error}_{\operatorname{den}}}}. 
\end{split}
\end{equation}
We will again apply Lemma \ref{lemma:little_bound_ratio} to the above ratio. In particular, we take $b_1 = \Var[\mu(\pi)]{Z_{it} \mid Y_{it} =1 , W_{it} = 1}$, $a_1 = \Cov[\mu(\pi)]{Y_{i(t+1)}, Z_{it} \mid Y_{it} =1 , W_{it} = 1}$, $b_2$ to be the above denominator and $a_1$ the above numerator. 
Note that the ratio $ \Cov[\mu(\pi)]{Y_{i(t+1)}, Z_{it} \mid Y_{it} =1 , W_{it} = 1} / \Var[\mu(\pi)]{Z_{it} \mid Y_{it} =1 , W_{it} = 1} \leq 1/ \sqrt{\Var[\mu(\pi)]{Z_{it} \mid Y_{it} =1 , W_{it} = 1}} \leq \sqrt{C_l C_f (1-C_f) D_n}$. Thus by Lemma  \ref{lemma:little_bound_ratio},
\begin{equation}
\label{eqn:f_prime_key}
 \widehat{f'_i(1,1)}_{\delta_T} =  \frac{\Cov[\mu(\pi)]{Y_{i(t+1)}, Z_{it} \mid Y_{it} =1 , W_{it} = 1} }{ \Var[\mu(\pi)]{Z_{it} \mid Y_{it} =1 , W_{it} = 1} } + \epsilon_{1,i}, 
\end{equation}
where 
\[\EE{\epsilon_{1,i}^2} \leq \frac{C_1 \p{L_n^2 \EE{\operatorname{error}_{\operatorname{den}}^2} + \EE{\operatorname{error}_{\operatorname{num}}^2}}}{D_n^2 \delta_T^2}
\leq \frac{C_2}{T \delta_T^2},
\]
when $\delta_T < C_l C_f (1-C_f)$ for some constants $C_1$ and $C_2$. Since $\delta_T \to 0$ as $T \infty$, we have 
\[\EE{\epsilon_{1,i}^2} = \oo\p{ \frac{1}{T \delta_T^2}}.
\]

Let 
\[\widetilde{f_i'(1,1)} =  \frac{\Cov[\mu(\pi)]{Y_{i(t+1)}, Z_{it} \mid Y_{it} =1 , W_{it} = 1} }{ \Var[\mu(\pi)]{Z_{it} \mid Y_{it} =1 , W_{it} = 1} } . \]
It then remains to show that $\widetilde{f_i'(1,1)}$ is close to $f'(1,1,\Qs(\pi))$. To this end, we analyze the term $\Cov[\mu(\pi)]{Y_{i(t+1)}, Z_{it} \mid Y_{it} =1 , W_{it} = 1}$. Note that
\begin{equation}
\label{eqn:LTE_remove_W_i}
\begin{split}
&\Cov[\mu(\pi)]{Y_{i(t+1)}, Z_{it} \mid Y_{it} =1 , W_{it} = 1}\\
& \qquad = \EE[\mu(\pi)]{ Y_{i(t+1)} Z_{it} \mid Y_{it} =1 , W_{it} = 1 }\\
& \qquad \qquad \qquad - \EE[\mu(\pi)]{ Y_{i(t+1)} \mid Y_{it} =1 , W_{it} = 1 } \EE[\mu(\pi)]{ Z_{i(t+1)} \mid Y_{it} =1 , W_{it} = 1 } \\
& \qquad = \EE[\mu(\pi)]{ f_i(1,1,Z_{it}) Z_{it} \mid Y_{it} =1 , W_{it} = 1 }\\
& \qquad \qquad \qquad  - \EE[\mu(\pi)]{ f_i(1,1,Z_{it}) \mid Y_{it} =1 , W_{it} = 1 } \EE[\mu(\pi)]{ Z_{i(t+1)} \mid Y_{it} =1 , W_{it} = 1 } \\
&  \qquad = \EE[\mu(\pi)]{ f_i(1,1,Z_{it}) Z_{it} \mid Y_{it} =1 } \\
& \qquad = \EE[\mu(\pi)]{ f_i(1,1,Z_{it}) Z_{it} \mid Y_{it} =1 , W_{it} = 1 }\\
& \qquad \qquad \qquad - \EE[\mu(\pi)]{ f_i(1,1,Z_{it}) \mid Y_{it} =1} \EE[\mu(\pi)]{ Z_{i(t+1)} \mid Y_{it} =1} \\
& \qquad = \Cov[\mu(\pi)]{f_i(1,1,Z_{it}),  Z_{it} \mid Y_{it} =1}. 
\end{split}
\end{equation}
We then do a Taylor expansion of $f_i(1,1,Z_{it})$ around the conditional mean $\EE[\mu(\pi)]{ Z_{i t} \mid Y_{it} =1}$. 
\[
\begin{split}
f_i(1,1,Z_{it}) &= f_i(1,1, \EE[\mu(\pi)]{ Z_{i t} \mid Y_{it} =1})\\
& \qquad \qquad + f_i'(1,1, \EE[\mu(\pi)]{ Z_{it} \mid Y_{it} =1}) (Z_{it} -  \EE[\mu(\pi)]{ Z_{i t} \mid Y_{it} =1})\\
& \qquad \qquad + \frac{1}{2} f_i'(1,1, \tilde{Z}_{it}) (Z_{it} -  \EE[\mu(\pi)]{ Z_{i t} \mid Y_{it} =1})^2,
\end{split}
\]
where $\tilde{Z}_{it}$ is between $Z_{it}$ and  $\EE[\mu(\pi)]{ Z_{i t} \mid Y_{it} =1}$. 
Therefore, taking covariance of both hand sides with $Z_{it}$, we get
\[
\begin{split}
&\Cov[\mu(\pi)]{f_i(1,1,Z_{it}) , Z_{it} \mid Y_{it} =1} \\
& \qquad  = \Cov[\mu(\pi)]{ f_i(1,1, \EE[\mu(\pi)]{ Z_{i t} \mid Y_{it} =1}) , Z_{it} \mid Y_{it} =1}\\
& \qquad \qquad + f_i'(1,1, \EE[\mu(\pi)]{ Z_{it} \mid Y_{it} =1}) \Cov[\mu(\pi)]{Z_{it} -  \EE[\mu(\pi)]{ Z_{i t} \mid Y_{it} =1} , Z_{it} \mid Y_{it} =1}\\
& \qquad \qquad + \frac{1}{2} \Cov[\mu(\pi)]{ f_i'(1,1, \tilde{Z}) (Z_{it} -  \EE[\mu(\pi)]{ Z_{i t} \mid Y_{it} =1})^2 , Z_{it} \mid Y_{it} =1}\\
&\qquad =  f_i'(1,1, \EE[\mu(\pi)]{ Z_{it} \mid Y_{it} =1}) \Var[\mu(\pi)]{Z_{it} \mid Y_{it}=1}\\
&\qquad \qquad + \frac{1}{2} \Cov[\mu(\pi)]{ f_i'(1,1, \tilde{Z}) (Z_{it} -  \EE[\mu(\pi)]{ Z_{i t} \mid Y_{it} =1})^2 , Z_{it} \mid Y_{it} =1}. 
\end{split}
\]
This then implies that 
\begin{equation}
\label{eqn:LTE_bound_n1}
\begin{split}
&\abs{\frac{\Cov[\mu(\pi)]{f_i(1,1,Z_{it}) , Z_{it} \mid Y_{it} =1}}{ \Var[\mu(\pi)]{Z_{it} \mid Y_{it}=1}} - f_i'(1,1, \EE[\mu(\pi)]{ Z_{it} \mid Y_{it} =1})}\\
&\qquad \qquad \leq \frac{L_{2,n} \EE[\mu(\pi)]{ \abs{Z_{it} -  \EE[\mu(\pi)]{ Z_{i t} \mid Y_{it} =1}}^3 \mid Y_{it}=1} }{2\Var[\mu(\pi)]{Z_{it} \mid Y_{it}=1}}. 
\end{split}
\end{equation}

We will then study $\EE[\mu(\pi)]{ \abs{Z_{it} -  \EE[\mu(\pi)]{ Z_{i t} \mid Y_{it} =1}}^3 \mid Y_{it}=1} $. Note that
\begin{equation}
\label{eqn:LTE_bound_n2}
\begin{split}
&\EE[\mu(\pi)]{ \abs{Z_{it} -  \EE[\mu(\pi)]{ Z_{i t} \mid Y_{it} =1}}^3 \mid Y_{it}=1} \\
&\qquad \qquad = \frac{1}{\EE[\mu(\pi)]{Y_{it}} }  \EE[\mu(\pi)]{ \abs{Z_{it} -  \EE[\mu(\pi)]{ Z_{i t} \mid Y_{it} =1}}^3 Y_{it} } \\
&\qquad \qquad \leq \frac{1}{C_y}  \EE[\mu(\pi)]{ \abs{Z_{it} -  \EE[\mu(\pi)]{ Z_{i t} \mid Y_{it} =1}}^3 }. 
\end{split}
\end{equation}
For $\EE[\mu(\pi)]{ \abs{Z_{it} -  \EE[\mu(\pi)]{ Z_{i t} \mid Y_{it} =1}}^3 }$, we can try to bound it using results from Theorem \ref{theo:mean_field2}. Specifically, Theorem \ref{theo:mean_field2} implies that there exist random vectors $X$ and $Y$ such that 
$Y \sim \mu(\pi)$, $X_i \sim \operatorname{Ber}(\Ps(\pi))$ independently, and $d_{E,3}(X,Y) \leq (2L_n D_n^{\frac{3}{2}} + 1)/(1-C)$. We write $Z_i = \sum_{j \in \ngb_i} Y_j$ and $V_i = \sum_{j \in \ngb_i} X_j$. Then
\begin{equation}
\label{eqn:bound_num}
\begin{split}
&\EE[\mu(\pi)]{ \abs{Z_{it} -  \EE[\mu(\pi)]{ Z_{i t} \mid Y_{it} =1}}^3 }
= \EE{\abs{Z_i - \EE{Z_i \mid Y_1 = 1}}^3}\\
&\leq 16\p{\EE{\abs{Z_i - V_i}^3} + \EE{\abs{V_i - \Qs_i(\pi)}^3} + \abs{\Qs_i(\pi) - \EE{Z_i}}^3 + \abs{\EE{Z_{i} \mid Y_{i} = 1} - \EE{Z_i}}^3}. 
\end{split}
\end{equation}
By Theorem \ref{theo:mean_field2}, we have $\EE{\abs{Z_i - V_i}^3}  \leq \p{(2L_n D_n^{\frac{3}{2}} + 1)/(1-C)}^3 \leq C_1 D_n^{\frac{3}{2}}$ for some constant $C_1$. For the second term $\EE{\abs{V_i - \Qs_i(\pi)}^3}$, note that $V_i$ is a sum of independent Bernoulli random variables, and $\Qs$ is the sum of the corresponding means; thus $\EE{\abs{V_i - \Qs_i(\pi)}^3} \leq D_n^\frac{3}{2}$. For the third term, the first part of Theorem \ref{theo:mean_field2} implies that 
\begin{equation}
\label{eqn:LTE_bound3}
\abs{\Qs_i(\pi) - \EE{Z_i}} = \abs{\EE{Z_i - V_i}} \leq W_{d,E}(\law(Y), \law(X)) \leq C/(2(1-C)) D_n^{\frac{1}{2}}.
\end{equation}\
Thus, $\abs{\Qs_i(\pi) - \EE{Z_i}}^3 \leq C_3 D_n^{\frac{3}{2}}$ for some constant $C_3$. Finally, for the fourth term $\abs{\EE{Z_{i} \mid Y_{i} = 1} - \EE{Z_i}}^3$, note that
\[\abs{\EE{Z_{i} \mid Y_{i} = 1} - \EE{Z_i}}
= \frac{1}{\EE{Y_i}} \abs{\Cov{Z_i, Y_i}}
\leq \frac{1}{C_y}  \EE{ \abs{Z_i - \EE{Z_i}}}. 
  \]
Interestingly, we can use the first three terms to bound this term. 
\begin{equation}
\label{eqn:LTE_bound4}
 \p{\EE{ \abs{Z_i - \EE{Z_i}}}}^3
\leq 9\p{\EE{\abs{Z_i - V_i}^3} + \EE{\abs{V_i - \Qs_i(\pi)}^3} + \abs{\Qs_i(\pi) - \EE{Z_i}}^3}
\leq C_4 D_n^{\frac{3}{2}},
\end{equation}
for some constant $C_4$. 
Combining the results of the four terms and plugging them back into \eqref{eqn:bound_num}, we get
\[ 
\EE[\mu(\pi)]{ \abs{Z_{it} -  \EE[\mu(\pi)]{ Z_{i t} \mid Y_{it} =1}}^3 } \leq C_{\operatorname{num}}D_n^{\frac{3}{2}},
 \]
for some constant $C_{\operatorname{num}}$. 
Together with \eqref{eqn:LTE_bound_n1} and \eqref{eqn:LTE_bound_n2}, the above bound shows 
\[
\begin{split}
\abs{\frac{\Cov[\mu(\pi)]{f_i(1,1,Z_{it}) , Z_{it} \mid Y_{it} =1}}{ \Var[\mu(\pi)]{Z_{it} \mid Y_{it}=1}} - f_i'(1,1, \EE[\mu(\pi)]{ Z_{it} \mid Y_{it} =1})}
 \leq \frac{L_{2,n} C_{\operatorname{num}}D_n^{\frac{3}{2}} }{2 C_y \Var[\mu(\pi)]{Z_{it} \mid Y_{it}=1}}. 
\end{split}
\]
Note also that $\Var[\mu(\pi)]{Z_{it} \mid Y_{it} =1}$ has a nice lower bound \eqref{eqn:LTE_var_lower_bound}. 
\[
\Var[\mu(\pi)]{Z_{it} \mid Y_{it} =1} \geq C_l C_f (1-C_f) D_n.
\]
Thus 
\[
\begin{split}
\abs{\frac{\Cov[\mu(\pi)]{f_i(1,1,Z_{it}) , Z_{it} \mid Y_{it} =1}}{ \Var[\mu(\pi)]{Z_{it} \mid Y_{it}=1}} - f_i'(1,1, \EE[\mu(\pi)]{ Z_{it} \mid Y_{it} =1})}
 \leq C_{\operatorname{dir}}L_{2,n} D_n^{\frac{1}{2}}. 
\end{split}
\]
for some constant $C_{\operatorname{dir}}$. This implies that
\begin{equation}
\label{eqn:LTE_final_but_two}
\widetilde{f'_i(1,1)} =  f_i'(1,1, \EE[\mu(\pi)]{ Z_{it} \mid Y_{it} =1}) + \epsilon_{2,i},  
\end{equation}
where $\abs{\epsilon_{2,i}} \leq C_2 L_{2,n} D_n^{\frac{1}{2}}$ for some constant $C_2$.

Finally, it remains to show that $f_i'(1,1, \EE[\mu(\pi)]{ Z_{it} \mid Y_{it} =1}) $ is close to $f_i'(1,1, \Qs(\pi))$. Since the second derivative of $f_i$ is bounded above by $L_{2,n}$, we can bound the difference of the two by
\[
\abs{f_i'(1,1, \EE[\mu(\pi)]{ Z_{it} \mid Y_{it} =1}) - f_i'(1,1, \Qs(\pi))}
\leq L_{2,n} \abs{ \EE[\mu(\pi)]{ Z_{it} \mid Y_{it} =1} - \Qs(\pi)}. 
\] 
By \eqref{eqn:LTE_bound3} and \eqref{eqn:LTE_bound4}, we can easily establish that $\abs{\EE[\mu(\pi)]{ Z_{it} \mid Y_{it} =1} - \Qs(\pi)} \leq C_3 D_n^{\frac{1}{2}}$ for some constant $C_3$. Therefore, we have that $ f_i'(1,1, \EE[\mu(\pi)]{ Z_{it} \mid Y_{it} =1}) = f_i'(1,1, \Qs(\pi)) + \epsilon_{3,i}$, where $\abs{\epsilon_{3,i}} \leq L_{2,n} D_n^{\frac{1}{2}}$. Together with \eqref{eqn:f_prime_key} and \eqref{eqn:LTE_final_but_two}, we get
\[ \widehat{f'_i(1,1)}_{\delta_T} =  f_i'(1,1, \Qs(\pi)) + \epsilon_i,   \]
where $\EE{\epsilon_i^2} = \oo\p{1/(T \delta_T^2) + D_n^{-3}}$ for some constant $C_{\operatorname{div}} $. 

\subsection{Proof of Theorem \ref{theo:LTE_estimation2}}
We start with analyzing the difference between $\mathbf{1}\trans (M_{\eta_n} - \Dh A - \Wh_{\kappa_n}) ^{-1} \uh$ and $\mathbf{1}\trans (I - D A - W) ^{-1} \bu$. We will make use of Lemma \ref{lemma:matrix_inverse}. Specifically, we take $X_1 = M_{\eta_n} - \Dh A - \Wh_{\kappa_n}$, $X_2 = I - D A - W$, $\alpha = \mathbf{1}/n$, $\beta_1 = \uh$ and $\beta_2 = \bu$. Lemma \ref{lemma:matrix_inverse} shows that 
\begin{equation}
\label{eqn:LTE_est_matrix_inverse}
\begin{split}
\abs{\alpha\trans X_1^{-1} \beta_1 - \alpha\trans X_2^{-1} \beta_2 } 
\leq  \frac{\Norm{\alpha} \Norm{\beta_1 - \beta_2}}{\lambda_{\operatorname{smallest}}(X_1)}
+ \Norm{(X_2\trans)^{-1}\alpha}_{\infty}\Norm{(X_1)^{-1}\beta_2}_{\infty} \sum_{i,j} \abs{X_{1,i,j} - X_{2,i,j}}. 
\end{split}
\end{equation}
We will analyze the terms one by one. We start with $X_1$. Note that by construction, each diagonal element of $M_{\eta_n} - \Wh_{\kappa_n}$ is bounded below by $\kappa_n$. At the same time, $M_{\eta_n, i,i} - \Wh_{\kappa_n,i,i} \geq \hat{d}_i D_n/(1-\eta_n) \geq \sum_{j}\sqb{\hat{D}A}_{i,j}$. Therefore by Lemma \ref{lemma:l_infty_norm2}, $X_1$ is invertible, $\lambda_{\operatorname{smallest}}(X_1) \geq \eta_n \kappa_n$. Further, $\Norm{(X_1)^{-1}\beta_2}_{\infty} \leq \Norm{\beta_2}_{\infty}/\eta_n \kappa_n \leq 1/\eta_n \kappa_n$. Now for $X_2$, by Lemma \ref{lemma:l_infty_norm}, $X_2$ is invertible and $\Norm{(X_1\trans)^{-1}\alpha}_{\infty} \leq \Norm{\alpha}_{\infty}/(1-C) = 1/(n(1-C))$. For $\alpha$, we clearly have $\Norm{\alpha} = 1/\sqrt{n}$. 

It remains to analyze $\Norm{\beta_1 - \beta_2}$ and $\sum_{i,j} \abs{X_{1,i,j} - X_{2,i,j}}$. For $\Norm{\beta_1 - \beta_2}$, we note that by Proposition \ref{prop:estimate_LDE}, Lemma \ref{lemma:p_hat_converge} and Lemma \ref{lemma:distance_between_mu_ps}, 
\[
\begin{split}
\EE{\Norm{\beta_1 - \beta_2}^2}
= \sum_i \EE{\sqb{\p{ \hat{b}_i v_i + \hat{d}_i \Ph_{i}(\pi) v_i} - \p{ b_i(\Qs(\pi)) v_i + d_i(\Qs(\pi)) \Ps_{i}(\pi) v_i}}^2}
\leq C_1 n \p{\frac{1}{T} + L_n},
\end{split}
\]
for some constant $C_1$. For $\sum_{i,j} \abs{X_{1,i,j} - X_{2,i,j}}$, we study $\hat{D}$, $\hat{W}_{\kappa_n}$ and $M_{\eta_n}$ separately. For $\hat{D}$, note that 
\begin{equation}
\label{eqn:D_bound}
\EE{\p{\hat{D}_i - D_i}^2} = \oo\p{ 1/(T\delta_t^2) + D_n^{-3} }
\end{equation}
by Theorem \ref{theo:LTE_estimation}, Lemma \ref{lemma:p_hat_converge} and Lemma \ref{lemma:distance_between_mu_ps}. Therefore, 
 \[\sum_{i,j} \EE{\abs{(\hat{D}A)_{i,j} - (DA)_{2,i,j}}} =  \oo\p{ n D_n \sqrt{1/(T\delta_t^2)} + n D_n^{-\frac{1}{2}} }. \]
For $\hat{W}_{\kappa_n}$, when $\kappa_n \leq C$, we have that by Proposition \ref{prop:estimate_LDE} and Lemma \ref{lemma:distance_between_mu_ps},
\begin{equation}
\label{eqn:W_bound}
\EE{\p{\hat{\omega}_{i, \kappa_n}- W_{i,i}}^2}
\leq \EE{\p{\hat{c}_i + \hat{d_i}\pi_i - W_{i,i}}^2}
\leq C_2 \p{\frac{1}{T} + L_n},
\end{equation}
for some constant $C_2$. Thus,
\[\sum_{i,j} \EE{\abs{\hat{W}_{\kappa_n, i,j} - W_{i,j}}}
= \sum_i \EE{\abs{\hat{\omega}_{i, \kappa_n}- W_{i,i}}}
= \oo\p{\frac{n}{\sqrt{T}} + n\sqrt{L_n}},
 \]
by Proposition \ref{prop:estimate_LDE} and Lemma \ref{lemma:distance_between_mu_ps}. Then we move on to study $M_{\eta_n}$, for each $i$, $M_{\eta_n, i,i} = \max\big(1, \hat{D}_i D_n/(1- \eta_n) + \hat{\omega}_{i, \kappa_n}\big)$. Let $\hat{m}_i = \hat{D}_i D_n/(1- \eta_n) + \hat{\omega}_{i, \kappa_n}$. When $L_n D_n/(1 - \eta_n) + B < 1$, 
\[
\begin{split}
0 \leq \EE{M_{\eta_n, i,i} - 1}
&= \EE{\hat{m}_i - 1; \hat{m}_i  \geq 1}
\leq \EE{\hat{m}_i - (D_{i,i}D_n/(1- \eta_n) + W_{i,i}); \hat{m}_i  \geq 1}\\
&\leq \EE{\abs{\hat{m}_i - (D_{i,i}D_n/(1- \eta_n) + W_{i,i})}
\abs{\hat{m}_i }}\\
& \leq \sqrt{\EE{\p{\hat{m}_i - (D_{i,i}D_n/(1- \eta_n) + W_{i,i})}^2}}\sqrt{\EE{
\p{\hat{m}_i}^2}}\\
&= \oo\p{D_n\sqrt{1/(T\delta_t^2)} + 1/\sqrt{D_n}}. 
\end{split}
\]
where the last line follows from \eqref{eqn:D_bound} and \eqref{eqn:W_bound}. Thus,  
\[\sum_{i,j} \EE{\abs{M_{\eta_n, i,j} - I_{i,j}}}
= \sum_{i} \EE{\abs{M_{\eta_n, i,i} - I1}}
= \sum_i \EE{M_{\eta_n, i,i} - 1}
= \oo\p{n D_n/\p{\sqrt{T}\delta_t} + n/\sqrt{D_n}}.
 \]

Combining the above results, we get 
\[
\begin{split}
&\EE{\sum_{i,j} \abs{X_{1,i,j} - X_{2,i,j}}} \\
&\qquad \leq 
\sum_{i,j} \EE{\abs{(\hat{D}A)_{i,j} - (DA)_{2,i,j}}} + \sum_{i,j} \EE{\abs{\hat{W}_{\kappa_n, i,j} - W_{i,j}}}
+ \sum_{i,j} \EE{\abs{M_{\eta_n, i,j} - I_{i,j}}}\\
&\qquad = \oo\p{n D_n/\p{\sqrt{T}\delta_t} + n/\sqrt{D_n}}.
\end{split}
\]
Finally, plugging everything in \eqref{eqn:LTE_est_matrix_inverse}, we get
\[
\begin{split}
&\EE{\abs{\mathbf{1}\trans (M_{\eta_n} - \Dh A - \Wh_{\kappa_n}) ^{-1} \uh-  \mathbf{1}\trans (I - D A - W) ^{-1} \bu}}\\
& = \oo\p{ \frac{1}{\eta_n \kappa_n} \p{\frac{1}{\sqrt{T}} + \sqrt{L_n}} }
+  \frac{1}{\eta_n \kappa_n n } \p{nD_n/\p{\sqrt{T}\delta_t} + n/\sqrt{D_n}}\\
& = \oo\p{\frac{1}{\eta_n \kappa_n} \p{\frac{D_n}{\sqrt{T}\delta_t} + \frac{1}{\sqrt{D_n}}}}. 
\end{split}
\]

Together with Theorem \ref{theo:LTE_char2}, we get
\[
\begin{split}
\frac{\hat{\tau}_{\LTE}( \pi + \Delta_n \bv,  \pi)}{\Delta_n} &= \frac{\tau_{\LTE}( \pi + \Delta_n \bv,  \pi)}{\Delta_n} + \oo_p\p{\frac{\sqrt{L_n}}{\Delta_n} + \Delta_n + \frac{1}{\eta_n \kappa_n} \p{\frac{D_n}{\sqrt{T}\delta_t} + \frac{1}{\sqrt{D_n}}}}\\
&= \frac{\tau_{\LTE}( \pi + \Delta_n \bv,  \pi)}{\Delta_n} + \oo_p\p{\frac{1}{\Delta_n \sqrt{D_n}} + \Delta_n + \frac{1}{\eta_n \kappa_n} \p{\frac{D_n}{\sqrt{T}\delta_t} + \frac{1}{\sqrt{D_n}}}}.
\end{split}
\]

\end{appendix}

\end{document}